\documentclass[10pt]{amsart}
\usepackage{amssymb}
\usepackage[mathscr]{eucal}
\usepackage{verbatim}
\usepackage[T1]{fontenc}
\usepackage{courier}
\usepackage{hyperref}
\usepackage{fancyhdr}

\numberwithin{equation}{section}

\newtheorem{theorem}{Theorem}[section]
\newtheorem{proposition}{Proposition}[section]
\newtheorem{lemma}[proposition]{Lemma}
\newtheorem{corollary}[proposition]{Corollary}

\theoremstyle{definition}
\newtheorem{definition}{Definition}[section]
\newtheorem{remark}{Remark}[section]

\theoremstyle{remark}

\newtheorem{example}{Example}[section]

\renewcommand{\theproposition}{\arabic{section}-\arabic{proposition}}
\newcommand{\eqdef}{\overset{\mbox{\tiny{def}}}{=}}
\newcommand{\sforspace}{\mathbf{s}}

\newcommand{\Mg}{\underline{g}{}}
\newcommand{\Ent}{{\eta}}

\newcommand{\wtv}{\widetilde{v}}

\newcommand{\wtg}{\widetilde{\gamma}_c}

\newcommand{\wtRc}{\widetilde{R}_c}
\newcommand{\wtP}{\widetilde{P}}
\newcommand{\wtQc}{\widetilde{Q}_c}

\newcommand{\vsubc}{v}
\newcommand{\Qc}{Q_c}
\newcommand{\Qinfinity}{Q_{\infty}}
\newcommand{\Rc}{R_c}
\newcommand{\Rinfinity}{R_{\infty}}

\newcommand{\gammac}{\gamma_c}

\newcommand{\Vbold}{\mathbf{V}}
\newcommand{\Vboldt}{\widetilde{\Vbold}}
\newcommand{\Vb}{\bar{\Vbold}}
\newcommand{\Vdot}{\dot{\Vbold}}
\newcommand{\VID}{\mathring{\Vbold}}
\newcommand{\VIDSmoothed}{{^{(0)}{\VID}}}

\newcommand{\scrV}{\boldsymbol{\mathscr{V}}}
\newcommand{\scrVb}{\bar{\scrV}}
\newcommand{\scrVtilde}{\boldsymbol{\widetilde{\mathscr{V}}}}

\newcommand{\scrVID}{\mathring{\scrV}}

\newcommand{\Wbold}{\mathbf{W}}
\newcommand{\Wb}{\bar{\Wbold}}
\newcommand{\Wboldt}{\widetilde{\Wbold}}
\newcommand{\Wdot}{\dot{\Wbold}}
\newcommand{\WID}{\mathring{\Wbold}}
\newcommand{\WIDSmoothed}{{^{(0)}{\WID}}}

\newcommand{\scrW}{\boldsymbol{\mathscr{W}}}
\newcommand{\scrWb}{\bar{\scrW}}
\newcommand{\scrWID}{\mathring{\scrW}}

\newcommand{\scrWtilde}{\boldsymbol{\widetilde{\mathscr{W}}}}
\newcommand{\scrWdot}{\boldsymbol{\dot{\mathscr{W}}}}

\newcommand{\Phidot}{\dot{\Phi}}

\newcommand{\Phit}{\widetilde{\Phi}}
\newcommand{\PhiID}{\mathring{\Phi}}

\newcommand{\Phibar}{\bar{\Phi}}

\newcommand{\PsiID}{\mathring{\Psi}}

\newcommand{\vID}{\mathring{v}}

\newcommand{\pID}{\mathring{p}}

\newcommand{\pIDSmoothed}{{^{(0)}{\pID}}}
\newcommand{\pbar}{\bar{p}}

\newcommand{\EntID}{\mathring{\Ent}}
\newcommand{\Entdot}{\dot{\Ent}}
\newcommand{\EntIDSmoothed}{{^{(0)}{\EntID}}}
\newcommand{\Entbar}{\bar{\Ent}}

\newcommand{\Pbar}{\bar{P}}
\newcommand{\Pdot}{\dot{P}}

\newcommand{\Jscrdot}{\dot{\mathscr{J}}}
\newcommand{\Jscrdotc}{{^{(c)}{\dot{\mathscr{J}}}}}
\newcommand{\Jscrdotinfinity}{{^{(\infty)}{\dot{\mathscr{J}}}}}

\newcommand{\Ac}{{_c\mathscr{A}}}
\newcommand{\Ainfinity}{{_{\infty}\mathscr{A}}}

\newcommand{\leqc}{\lesssim}

\begin{document}
\pagestyle{fancy}

\title{The Non-Relativistic Limit of the Euler-Nordstr\"{o}m System with Cosmological Constant}

\author{Jared Speck}

\address{Department of Mathematics, Rutgers University\footnote{This article was
finalized while the author was a postdoctoral researcher in the Princeton University math department.},
Hill Center, 110 Frelinghuysen Rd., Piscataway, NJ 08854, USA}  
\email{jspeck@math.princeton.edu}

\begin{abstract}
In this paper the author studies the singular limit $c \to \infty$ of the family
of Euler-Nordstr\"{o}m systems indexed by the parameters $\kappa^2$
and $c$ (EN$_{\kappa}^c$), where $\kappa^2 > 0$ is the cosmological
constant and $c$ is the speed of light. Using Christodoulou's
techniques to generate energy currents, the author develops
Sobolev estimates that show that initial data belonging to an appropriate Sobolev space
launch unique solutions to the EN$_{\kappa}^c$ system that converge to corresponding unique solutions of the
Euler-Poisson system with the cosmological constant $\kappa^2$ as $c$ tends to infinity.
\end{abstract}



\date{Version of \today}
\maketitle

\tableofcontents

\section{Introduction}
The Euler-Nordstr\"{o}m system models the evolution of a relativistic perfect fluid
with self-interaction mediated by Nordstr\"{o}m's theory of scalar gravity. In \cite{jS2008a}, we introduced the
system in dimensionless units and showed that the Cauchy problem is locally well-posed in the
Sobolev space\footnote{More precisely, we showed local well-posedness in a suitable affine shift of $H^N$ for 
$N \geq 3,$ where by ``affine shift'' of $H^N$ we mean the collection of all functions $F$ such
that $\|F - \Vb\|_{H^N} < \infty,$ where $\Vb$ is a fixed constant array; see Section \ref{S:Remarks}
for further discussion of this function space.} $H^N$ for $N \geq 3.$ In this article, we study the non-relativistic (also known as the ``Newtonian'') limit of the family of Euler-Nordstr\"{o}m systems indexed by the parameters $\kappa$ and
$c$ (EN$_{\kappa}^c$), where $\kappa^2$ is the cosmological
constant\footnote{The parameter $\kappa^2 > 0 $ is fixed throughout
this article. Remark \ref{R:KernelRemark} contains an explanation of why
our proof breaks down in the case $\kappa^2 =0.$ } and $c$ is the speed of light. The limit $c \to
\infty$ is singular because the EN$_{\kappa}^c$ system is hyperbolic
for all finite $c,$ while the limiting system, namely the Euler-Poisson system with
a cosmological constant (EP$_{\kappa}$), is not hyperbolic.
Using Christodoulou's techniques \cite{dC2000} to generate energy
currents, together with elementary harmonic analysis, we develop Sobolev estimates and use them to 
study the singular limit $c \to \infty.$ 

Before introducing our main theorem, we place this article in context by mentioning some related works. We remark that our list of references is not exhaustive. In \cite{sKaM1981}, Klainerman and Majda study singular limits in quasilinear symmetric hyperbolic systems, and in particular the incompressible limit (as the Mach number tends to $0$) of compressible fluids. In \cite{aR1994}, Rendall studies the singular limit $c \to \infty$ of the Vlasov-Einstein system and proves
that a class of data launches solutions to this system that converge to corresponding solutions of the Vlasov-Poisson system
as $c \to \infty$, thereby obtaining the first rigorous existence proof for the $c \to \infty$ limit of the Einstein equations coupled to a matter field. In \cite{sChL2004}, Calogero and Lee study the singular limit $c \to \infty$ of the Vlasov-Nordstr\"{o}m system and prove that a class of data launches solutions to this system that converge to corresponding solutions of the Vlasov-Poisson system at the rate $O(c^{-1}),$ a result analogous to our main theorem. In \cite{sB2005}, Bauer improves the rate of convergence to $O(c^{-4}),$ which is known as a ``1.5 post-Newtonian approximation.'' In \cite{sBmKgRaR2006}, Bauer, Kunze, Rein, and Rendall study the Vlasov-Maxwell and Vlasov-Nordstr\"{o}m systems and obtain a formula that relates the radiation flux at infinity to the motion of matter and that is analogous to the Einstein quadrupole formula (see e.g. \cite{nS1984}) in general relativity. In \cite{tO2007}, Oliynyk studies the singular limit $c \to \infty$ of the Euler-Einstein system. He exhibits a class of data that launches solutions that converge to corresponding solutions of the Euler-Poisson system as $c \to \infty,$ while in \cite{tO2008}, he improves the rate of convergence by showing that the ``first post-Newtonian expansion'' is valid.   

Our main theorem is in the spirit of the above results. We state it loosely here, and we state and prove it rigorously as Theorem \ref{T:NewtonianLimit}:

\begin{changemargin}{.25in}{.25in} 
\textbf{Main Theorem.} \ 
Let $N \geq 4$ be an integer, and assume that $\kappa^2 > 0.$ Then initial data belonging to a 
suitable affine shift of the Sobolev space $H^N$ launch unique solutions to the EN$_{\kappa}^c$ system that converge uniformly on a spacetime slab $[0,T] \times \mathbb{R}^3$ to corresponding unique solutions of the EP$_{\kappa}$ system as the speed of light $c$ tends to infinity.
\end{changemargin}
\noindent We remark that although we explicitly discuss only the EN$_{\kappa}^c$ system in this article, the techniques we apply can be generalized under suitable hypotheses to study singular limits of hyperbolic systems that derive from a Lagrangian and that feature a small parameter\footnote{The small parameter is $c^{-2}$ in the case of the EN$_{\kappa}^c$ system.}.

As discussed in \cite{jS2008a}, we consider the EN$_{\kappa}^c$ system to be a mathematical scalar 
caricature of the Euler-Einstein system with cosmological constant (EE$_{\kappa}^c$). We now provide some justification for this point of view. First of all, like the EE$_{\kappa}^c$ system, the EN$_{\kappa}^c$ system is a metric theory of gravity featuring gravitational waves that propagate along null cones. Second, the main theorem stated above shows that if $\kappa^2 > 0,$ then the Newtonian limit of the EN$_\kappa^c$ system is the EP$_{\kappa}$ system. Furthermore, as previously mentioned, Oliynyk's work \cite{tO2007} shows that
the Newtonian limit of the EE$_{0}^c$ system is the EP$_0$ system. Based on these considerations, we therefore expect\footnote{We temper this expectation by recalling that our proof does not work in the case $\kappa^2 =0$ and that in contrast to the initial value problem studied here, Oliynyk considers the case $\kappa^2 = 0$ with compactly supported data under an adiabatic equation of state. This special class of equations of state allows one to make a ``Makino'' change of variables that regularizes the equations and overcomes the singularities that typically occur in the equations in regions where the proper energy density vanishes. Furthermore, this change of variables enables one to write the relativistic Euler equations in symmetric hyperbolic form. See \cite{tM1986} and \cite{aR1992a} for additional examples of this change of variables in the context of various fluid models.} that achieving an understanding of the evolution of solutions to the EN$_{\kappa}^c$ system will provide insight into the behavior of solutions to the vastly more complicated EE$_{\kappa}^c$ system. 

\subsection{Outline of the structure of the paper} 

Before proceeding, we outline the structure of this article. In Section \ref{S:Remarks}, we introduce some notation 
that we use throughout our discussion. In Section \ref{S:Origin}, we derive the EN$_{\kappa}^c$ equations with the parameter $c$ and then rewrite the equations using Newtonian state-space variables, a change of variables that is essential for comparing
the relativistic system EN$_{\kappa}^c$ to the non-relativistic system EP$_{\kappa}.$ In Section \ref{S:FormalLimit}, we provide for convenience the EN$_{\kappa}^c$ and EP$_{\kappa}$ systems in the form used for the remainder of the article. From this form, it is clear that formally, $\lim_{c \to \infty}$ EN$_{\kappa}^c$ $=$ EP$_{\kappa}.$ In Section \ref{SS:EOVc}, we introduce standard PDE matrix notation and discuss the Equations of Variation (EOV$_{\kappa}^c$), which are the linearization of the EN$_{\kappa}^c$ and EP$_{\kappa}$ systems. In Section \ref{S:cDependence}, we provide an extension of the Sobolev-Moser calculus that is useful for bookkeeping powers of $c.$ We also introduce some hypotheses on the $c-$dependence of the equation of state that are sufficient to prove our main theorem. We then apply the calculus to the EN$_{\kappa}^c$ system by proving several preliminary lemmas that are useful in the technical estimates that appear later. Roughly speaking, the lemmas describe the $c \to \infty$ asymptotics of the EN$_{\kappa}^c$ equations. 

In Section \ref{S:EnergyCurrentsc}, we introduce the energy currents that are used to control the Sobolev norms of the solutions. One of the essential features of the currents that we use is that they have a positivity property that is uniform for all large $c.$ In Section \ref{S:IVPc}, we describe a class of initial data for which our main theorem holds, and in Section \ref{S:SmoothingtheData}, we smooth the initial data for technical reasons. In Section \ref{S:UniforminTimeLocalExistence}, we recall the local existence result  \cite{jS2008a} for the EN$_{\kappa}^c$ system and prove an important precursor to our main theorem. Namely, we prove that solutions to the EN$_{\kappa}^c$ system exist on a common interval of time $[0,T]$ for all large $c.$ This proof is separated into two parts. The first part is a continuous induction argument based on some technical lemmas. The second part is the proof of these technical lemmas, which are a series of energy estimates derived with the aid of the calculus developed in Section \ref{S:cDependence}. The two basic tools we use for generating the energy estimates are energy currents and the estimate $\| f \|_{H^2} \leq C \cdot \|( \Delta  - \kappa^2)f \|_{L^2},$ for $f \in H^2.$ In Section \ref{S:NonrelativisticLimit}, we state and prove our main theorem.

\section{Remarks on the Notation} \label{S:Remarks}
    We introduce here some notation that is used throughout this
    article, some of which is non-standard. We assume that the reader is familiar with standard notation for the
    $L^p$ spaces and the Sobolev spaces $H^k.$ Unless otherwise stated, the symbols $L^p$ and $H^k$ refer to
    $L^p(\mathbb{R}^3)$ and $H^k(\mathbb{R}^3)$ respectively.

\subsection{Notation regarding differential operators}
    If $F$ is a scalar or \\
    finite-dimensional array-valued function on
    $\mathbb{R}^{1 + 3},$ then
    $D^{(a)}F$ denotes the array consisting of all $a^{th}-$order spacetime coordinate partial derivatives (including
    partial derivatives with respect to time) of every component of
    $F,$ while $\partial^{(a)} F$ denotes the array of
    consisting of all $a^{th}-$order \emph{spatial} coordinate partial derivatives of every component of $F.$ We
    write $DF$ and $\partial F$ respectively instead of $D^{(1)} F$ and $\partial^{(1)} F.$
    $\nabla$ denotes the Levi-Civita connection corresponding to the spacetime metric $g$ defined in 
   	\eqref{E:NordstromMetric}.

\subsection{Index conventions}                  \label{SS:IndexConventions}
    We adopt Einstein's convention that diagonally repeated Latin
    indices are summed from $1 \ \mbox{to} \ 3,$ while diagonally repeated Greek
    indices are summed from $0 \ \mbox{to} \ 3.$ Indices are raised an lowered using the spacetime metric
    $g,$ which is defined in \eqref{E:NordstromMetric}, or the Minkowski metric $\Mg,$ depending on context.

\subsection{Notation regarding norms and function spaces} \label{SS:NormsandFunctionSpaces}
	If $E \subset \mathbb{R}^3$ and $\bar{\mathbf{V}} \subset \mathbb{R}^n$ is a constant array, we use the notation
    \begin{align} \label{E:LpVbNormDef}
        & \|F\|_{L_{\Vb}^p(E)}
        \overset{\mbox{\tiny{def}}}{=} \|F - \Vb
        \|_{L^p(E)},
    \end{align}
    and we denote the set of all (array-valued) Lebesgue measurable functions $F$ such \\ 
    that $\|F\|_{L_{\Vb}^p(E)} < \infty$
    by $L_{\Vb}^p(E).$ We also define the $H^j_{\Vb}(E)$ norm of $F$ by
    \begin{align}     \label{E:NJVbNormDef}
    \|F\|_{H^j_{\Vb}(E)}
        \eqdef \Big( \sum_{|\vec{\alpha}| \leq j} \|\partial_{\vec{\alpha}}(F - \Vb)\|_{L^2(E)}^2 \Big)^{1/2}, 
   	\end{align}
   	where $\partial_{\vec{\alpha}}$ is a multi-indexed operator representing repeated partial differentiation
    with respect to \emph{spatial} coordinates. Unless we indicate otherwise, we assume that $E = 
   	\mathbb{R}^3$ when the set $E$ is not explicitly written. 
    
    \begin{remark} \label{R:Norm}
    		Technically speaking, the $\| \cdot \|_{H^j_{\Vb}}$ are not norms in general, since for example
       	$\| \mathbf{0} \|_{H^j_{\Vb}} = \infty$ unless $\Vb = \mathbf{0}.$ This is not a problem because in this article, we only
        study the $\| \cdot \|_{H^j_{\Vb}}$ ``norm'' of functions $F$ that by design feature 
        $\|F\|_{H^j_{\Vb}} < \infty.$
    \end{remark}

    If $F$ is a map from $[0,T]$ into the normed function space $X,$ we use the notation
    \begin{equation}
        \mid\mid\mid F \mid\mid\mid_{X,T} \ \overset{\mbox{\tiny{def}}}{=} \underset{t \in
        [0,T]}{\sup} \|F(t)\|_{X} \label{E:HNprimesupovert}.
    \end{equation}
    We also use the notation
    $C^j([0,T],X)$ to denote the set of $j$-times continuously differentiable maps from $(0,T)$ into
    $X$ that, together with their derivatives up to order $j,$ extend continuously to $[0,T].$

    If $\mathfrak{D} \subset \mathbb{R}^n,$ then
    $C^j_b(\mathfrak{D})$ denotes the set of $j-$times continuously differentiable
    functions (either scalar or array-valued, depending on context) on Int$(\mathfrak{D})$ with bounded derivatives up to
    order $j$ that extend continuously to the closure of $\mathfrak{D}.$ The norm of a function $\mathfrak{F} 
    \in C^j_b(\mathfrak{D})$ is defined by
    \begin{equation} \label{E:CbkNormDef}
        |\mathfrak{F}|_{j,\mathfrak{D}} \overset{\mbox{\tiny{def}}}{=} \sum_{|\vec{I}|\leq j} \sup_{z \in \mathfrak{D}}
        |\partial_{\vec{I}}\mathfrak{F}(z)|,
    \end{equation}
    where $\partial_{\vec{I}}$ is a multi-indexed operator representing repeated partial differentiation
    with respect to the arguments $z$ of $\mathfrak{F},$ which may be either spacetime coordinates or state-space variables
    depending on context.

\subsection{Notation for $c-$independent inequalities}
    If $A_c$ is a quantity that depends on the parameter $c,$ and $X$
    is a quantity such that $A_c \leq X$ holds for all large $c,$ then we indicate this by writing 
    \begin{align} \label{E:leqcdef}
    	A_c \leqc X. 
    \end{align}

\subsection{Notation regarding constants} \label{SS:Constants}
    We use the symbol $C$ to denote a generic constant in the estimates
    below which is free to vary from line to line. If the constant
    depends on quantities such as real numbers $N,$ subsets $\mathfrak{D}$ of
    $\mathbb{R}^n,$ functions $\mathfrak{F}$ of the state-space variables, etc., that are peripheral
    to the argument at hand, we sometimes
    indicate this dependence by writing $C(N,\mathfrak{D},\mathfrak{F}),$
    etc. We explicitly show the dependence on such quantities when it is (in
    our judgment) illuminating, but we often omit the dependence on such quantities
    when it overburdens the notation without being illuminating. Occasionally, we shall use additional symbols
    such as $\Lambda_1, Z, L_2,$ etc., to denote constants that play a distinguished role in the discussion.

\section{The Origin of the EN$_{\kappa}^c$ System} \label{S:Origin}
        In this section, we insert both the speed of light $c$ and Newton's universal
        gravitational constant $G$ into the Euler-Nordstr\"{o}m system with a cosmological constant and perform a 
        Newtonian change of variables, which brings the system into the form \eqref{E:ENkappac1} - \eqref{E:PDefcII}. A similar analysis 
        for the Vlasov-Nordstr\"{o}m system\footnote{The Vlasov-Nordstr\"{o}m (VN) model describes a particle density function $f$ on
        physical space $\times$ momentum space that evolves due to self-interaction mediated by Nordstr\"{o}m's theory
        of gravity. Various aspects of this system are studied, for example, in \cite{sC2003}, and \cite{sC2006}.} 	
        is carried out in \cite{sChL2004}.

    \subsection{Deriving the equations with $c$ as a parameter} \label{SS:ENkappacDerivation} 
		We assume that spacetime is a four-dimensional Lorentzian manifold $\mathcal{M}$ and 
    furthermore, that there is a global rectangular (inertial) coordinate system on $\mathcal{M}$. We use the notation
    \begin{align}                               \label{E:SpacetimePoint}
        x=(x^0,x^1,x^2,x^3)
    \end{align}
    to denote the components of a spacetime point $x$ in this fixed coordinate system, and for this
    preferred time-space splitting, we identify $t=x^0$ with time and \\ 
    $\sforspace=(x^1,x^2,x^3)$ with space.
    Note that we are breaking with the usual convention, which is $x^0 =ct.$ The components of the Minkowski metric and its inverse
    in the inertial coordinate system are given by 
    \begin{align}
    	\Mg_{\mu \nu} &= \mbox{diag}(-c^2,1,1,1) \\
    	\Mg^{\mu \nu} &= \mbox{diag}(-c^{-2},1,1,1)
  	\end{align}
    respectively. We adopt Nordstr\"{o}m's postulate, namely that the \emph{spacetime metric} $g$
    is related to the Minkowski metric by a conformal scaling factor:
    \begin{align}
        g_{\mu \nu}= e^{2\phi} \Mg_{\mu \nu}.                       \label{E:NordstromMetric}
    \end{align}
    In \eqref{E:NordstromMetric}, $\phi$ is the \emph{dimensionless} \emph{cosmological-Nordstr\"{o}m potential}, a scalar 
    quantity.

    We now briefly introduce the notion of a relativistic perfect fluid. Readers may consult \cite{nAgG2007} or
    \cite{dC1995} for more background. For a perfect fluid model, the components of the energy-momentum-stress density
    tensor (which is commonly called the ``energy-momentum tensor'' in the literature) of matter read
    \begin{align}               \label{E:EMTensorcDef}
        T^{\mu \nu} = c^{-2}(\rho + p) u^{\mu} u^{\nu}
        + p g^{\mu \nu} = c^{-2}(\rho + p) u^{\mu} u^{\nu}
        + e^{-2 \phi}p \Mg^{\mu \nu},
     \end{align}
     where $\rho$ is the \emph{proper energy density} of the
     fluid, $p$ is the \emph{pressure} (this ``proper'' quantity is defined in a local rest frame), and $u$ is the
     \emph{four-velocity}, which is subject to the normalization constraint
     \begin{align}                               \label{E:uNormalizedcSquared}
        g_{\mu \nu} u^{\mu} u^{\nu} = e^{2\phi} \Mg_{\mu \nu} u^{\mu} u^{\nu}= -c^2.
     \end{align}
     The Euler equations for a perfect fluid are (see e.g. \cite{dC1995})
     \begin{align}
        \nabla_{\mu} T^{\mu \nu} &= 0 \qquad (\nu=0,1,2,3) \label{E:Euler} \\
        \nabla_{\mu}(n u^{\mu}) &= 0,  \label{E:nandulaw}
     \end{align}
     where $n$ is the \emph{proper number density} and $\nabla$ denotes the covariant derivative induced by the spacetime metric $g.$

    Nordstr\"{o}m's theory\footnote{Norstr\"{o}m's theory of gravity, although shown to be physically wrong through experiment, was 
    the first metric theory of gravitation.} \cite{gN1913} provides the 
    following evolution equation\footnote{Nordstr\"{o}m considered only the case $\kappa = 0.$} for $\phi:$ we
    define an auxiliary energy-momentum-stress density tensor
    \begin{align}
        T_{\mbox{\tiny{aux}}}^{\mu \nu} \eqdef e^{6\phi}T^{\mu \nu} = c^{-2}e^{6 \phi}(\rho + p) u^{\mu}
        u^{\nu} + e^{4 \phi}p \Mg^{\mu \nu},                                                   \label{E:EMAuxDefc}
    \end{align}
    and postulate that $\phi$ is a solution to
    \begin{align}
        \square \phi - {\kappa}^2 \phi = -4 \pi
        c^{-4} Ge^{4\phi} \mbox{tr}_g T
        = -4 \pi c^{-4}G \Mg_{\mu \nu}T_{\mbox{\tiny{aux}}}^{\mu \nu} =
        4 \pi c^{-4}Ge^{4 \phi}(\rho - 3p).                                                              \label{E:phi}
    \end{align}
    Note that
    \begin{align}
        \square \phi \eqdef \Mg^{\mu \nu} \partial_{\mu} \partial_{\nu} \phi= -c^{-2} \partial^2_t \phi + \triangle \phi
    \end{align}
    is the wave operator on flat spacetime applied to $\phi.$ The virtue of the postulate equation
    \eqref{E:phi}, as we shall see, is that it provides us with continuity equations \eqref{E:ENEMContinuityc} for an 
    energy-momentum-stress density tensor $\Theta$ in Minkowski space.

    We also introduce the \emph{entropy per particle}, a thermodynamic variable
    that we denote by $\Ent,$ and we close the system by supplying an 
    equation of state, which may depend on $c.$ A ``physical''
    equation of state for a perfect fluid state satisfies the following criteria (see e.g. \cite{yGsTZ1998}):
    
    \begin{enumerate}
  		
  		\item $\rho \geq 0$ is a function of $n \geq 0$ and ${\Ent} \geq 0.$

    	\item $p \geq 0$ is defined by
        \begin{align}
            p= n \left. \frac{\partial \rho}{\partial n} \right|_{\Ent} -
            \rho,                                                                                   \label{E:pressure}
        \end{align}
        where the notation $\left. \right|_{\cdot}$ indicates partial differentiation with $\cdot$ held
        constant.
        \item A perfect fluid satisfies
        \begin{align}
        \left. \frac{\partial \rho}{\partial n} \right|_{\Ent} >0, \left. \frac{\partial p}{\partial
        n} \right|_{\Ent}>0, \left. \frac{\partial \rho}{\partial {\Ent}} \right|_n \geq 0 \
        \mbox{with} \ ``='' \ \mbox{iff} \ \Ent=0.                             \label{E:EOSAssumptions}
        \end{align}
        As a consequence, we have that $\sigma,$ the speed
        of sound in the fluid, is always real for $\Ent > 0:$
            \begin{align}
                \sigma^2 \overset{def}{=} c^2 \left.\frac{\partial p}{\partial
                \rho}\right|_{\Ent} = c^2 \frac{{\partial p / \partial n}|_{\Ent}}{{\partial \rho / \partial
                n}|_{\Ent}} > 0.                                                                         \label{E:SpeedofSoundc}
            \end{align}
        \item We also demand that the speed of sound is positive and less than the speed of light
        whenever $n > 0$ and $\Ent > 0$:
            \begin{align} \label{E:Causalityc}
                n>0 \ \mbox{and} \ \Ent > 0 \implies 0 < \sigma < c.
            \end{align}
    \end{enumerate}
    
   	Postulates $(1) - (3)$ express the laws of thermodynamics and fundamental thermodynamic assumptions, while postulate 
    $(4)$ ensures that at each $x \in \mathcal{M},$ vectors that are causal with respect to the sound cone in $T_x \mathcal{M}$ are 
    necessarily causal with respect to the gravitational null cone in $T_x \mathcal{M};$ see Section \ref{SS:JscrdotcPositiveDefinite}.

    \begin{remark}
    We note that the assumptions $\rho \geq 0, p \geq 0$ together imply that
    the energy-momentum-stress density tensor \eqref{E:EMTensorcDef} satisfies both the \emph{weak energy 
    condition} ($T_{\mu \nu} X^{\mu} X^{\nu} \geq 0$ holds whenever $X$ is timelike and future-directed
    with respect to the gravitational null cone) 
    and the \emph{strong energy condition} \\
    ($[T_{\mu \nu} - 1/2 g^{\alpha \beta}T_{\alpha \beta} g_{\mu \nu}]X^{\mu}X^{\nu} \geq 0$ holds whenever $X$ 
    is timelike and future-directed with respect to the gravitational null cone). Furthermore, if we assume that the equation of state is 
    such that $p=0$ when
    $\rho = 0,$ then \eqref{E:SpeedofSoundc} and \eqref{E:Causalityc} guarantee that $p \leq \rho.$ It is then easy to check
    that $0 \leq p \leq \rho$ implies the \emph{dominant energy condition}
    ($-T^{\mu}_{\ \nu} X^{\nu}$ is causal and future-directed whenever $X$ is causal and future-directed with respect
    to the gravitational null cone).
		\end{remark}
    
    By \eqref{E:EOSAssumptions}, we can solve for $\sigma^2$ and
    $c^{-2}\rho$ as $c-$indexed functions $\mathfrak{S}_c^2$ and $\mathfrak{R}_c$ respectively of $\Ent$ and $p:$
        \begin{align}
            \sigma^2 &\eqdef \mathfrak{S}_c^2(\Ent,p)                                                         \label{E:SigmaSquaredc}\\
            c^{-2}\rho &\eqdef \mathfrak{R}_c(\Ent,p).                                                       \label{E:EnergyDensityc}
        \end{align}
    
		\noindent We also will make use of the following identity implied by \eqref{E:SpeedofSoundc}, \eqref{E:SigmaSquaredc}, and 
        \eqref{E:EnergyDensityc}:
        \begin{align} \label{E:partialRpartialpandStoNegativeTwoRelationshipc}
            \left. \frac{\partial \mathfrak{R}_c}{\partial p}(\Ent,p)\right|_{\Ent} = \mathfrak{S}_c^{-2}(\Ent,p).
        \end{align}

   \begin{remark} \label{R:NewtonianR}
        Note that $c^{-2}\rho$ has the dimensions of mass density.
        As we will see in Section \ref{S:cDependence}, $\lim_{c \to \infty}\mathfrak{R}_c(\Ent,p)$
        will be identified with the Newtonian mass density.
    \end{remark}

    We summarize by stating that \textit{the equations \eqref{E:NordstromMetric} - \eqref{E:nandulaw},
    \eqref{E:phi}, \eqref{E:pressure}, and \eqref{E:EnergyDensityc} constitute the EN$_{\kappa}^c$
    system}.

    \subsection{A reformulation of the EN$_{\kappa}^c$ system in Newtonian variables} \label{SS:NewtonianReformulationc} 
    In this section, we reformulate the EN$_{\kappa}^c$ system as a
    fixed background theory in flat Minkowski space and introduce a Newtonian change of state-space variables.
    The resulting system \eqref{E:ENkappac1} - \eqref{E:PDefcII} is an equivalent formulation of the EN$_{\kappa}^c$
    system. We remark that \emph{for the remainder of this article, all indices are raised and lowered
    with the Minkowski metric} $\Mg,$ so that $\partial^{\lambda} \phi = \Mg^{\mu \lambda}
    \partial_{\mu} \phi.$ To begin, we use the form of the metric \eqref{E:NordstromMetric} to
    compute that in our inertial coordinate system, the continuity equation \eqref{E:Euler} for the
    energy-momentum-stress density tensor \eqref{E:EMTensorcDef} is given by 
    \begin{align}
        0 & = \nabla_{\mu} T^{\mu \nu} = \partial_{\mu} T^{\mu \nu} + 6
            T^{\mu \nu} \partial_{\mu}\phi -
            \Mg_{\alpha \beta}T^{\alpha \beta}\partial^{\nu} \phi \notag\\
        & = \partial_{\mu} T^{\mu \nu} + 6 T^{\mu \nu}
            \partial_{\mu}\phi- e^{-6 \phi}\Mg_{\alpha \beta}
            T_{\mbox{\tiny{aux}}}^{\alpha \beta} \partial^{\nu}\phi \qquad (\nu=0,1,2,3),                     \label{E:TDiv}
    \end{align}
    where $T_{\mbox{\tiny{aux}}}^{\mu \nu}$ is define in \eqref{E:EMAuxDefc}. For this calculation we made use of the
    explicit form of the Christoffel symbols of $g$ in our rectangular coordinate system:
    \begin{align}
        \Gamma_{\mu \nu}^{\alpha} = \delta^{\alpha}_{\nu}
        \partial_{\mu}\phi + \delta^{\alpha}_{\mu}
        \partial_{\nu}\phi - \Mg_{\mu \nu}\Mg^{\alpha \beta} \partial_{\beta}
        \phi.                                                                                               \label{E:Christoffelc}
    \end{align}
    Using the postulated equation \eqref{E:phi} for $\phi,$
    \eqref{E:TDiv} can be rewritten as

    \begin{align}
        0 = e^{6\phi} \nabla_{\mu} T^{\mu \nu} = \partial_{\mu} \Big[T^{\mu
        \nu}_{\mbox{\tiny{aux}}}
        + \frac{c^4}{4 \pi G} \big(\partial^{\mu} \phi \partial^{\nu} \phi -\frac{1}{2} \Mg^{\mu \nu} \partial^{\alpha} \phi
        \partial_{\alpha} \phi - \frac{1}{2} \Mg^{\mu\nu} \kappa^2 \phi^2\big) \Big].                         \label{E:Divergencec}
    \end{align}
    Let us denote the terms from \eqref{E:Divergencec} that are inside the square brackets
    as $\Theta^{\mu \nu}.$ Since the coordinate-divergence of $\Theta$ vanishes, we are provided with local conservation laws in
    Minkowski space, and we regard $\Theta$ as an energy-momentum-stress density tensor. 
    
    We also introduce the following state-space variables that
    play a mathematical role\footnote{The ``physical'' quantities are $\mathfrak{R}_c$ and $p.$} in the sequel:
    \begin{align}
        \Rc &\eqdef c^{-2}\rho e^{4\phi}= e^{4\phi}\mathfrak{R}_c(\Ent,p)  \label{E:RcDef} \\
        P &\eqdef p e^{4\phi}.       \label{E:PDefc}
    \end{align}
    After we make this change of variables, the components of $\Theta$ read
    \begin{align} 		\label{E:ENEMThetacDef}
        \Theta^{\mu \nu} \eqdef \big[\Rc &+ c^{-2}P\big] e^{2\phi}u^{\mu}u^{\nu} + P \Mg^{\mu \nu}
        + \frac{c^4}{4 \pi G}\Big(\partial^{\mu}\phi \partial^{\nu}\phi
            -\frac{1}{2} \Mg^{\mu \nu}\partial^{\alpha}\phi\partial_{\alpha} \phi
            - \frac{1}{2} \Mg^{\mu \nu} \kappa^2 \phi^2 \Big), 
    \end{align}
    and we replace \eqref{E:Euler} with the equivalent equation
    \begin{align}           \label{E:ENEMContinuityc}
        \partial_{\mu} \Theta^{\mu \nu}=0 \qquad (\nu=0,1,2,3).
    \end{align}

    We also expand the covariant differentiation from \eqref{E:nandulaw} in terms of coordinate
    derivatives and the Christoffel symbols \eqref{E:Christoffelc}, arriving at the
    equation
    \begin{align}
        \partial_{\mu}\big(n e^{4\phi} u^{\mu} \big)=0 .                                               \label{E:ENcContinuity}
    \end{align}

    Our goal is to obtain the system EN$_{\kappa}^c$ in the form \eqref{E:ENkappac1} - \eqref{E:PDefcII}
    below. To this end, we project \eqref{E:ENEMContinuityc} onto the orthogonal complement\footnote{We are referring here to the 
    orthogonal complement defined by the Minkowski metric $\Mg$.} of $u$
    and in the direction of $u.$ We therefore introduce the rank 3 tensor $\Pi,$ which has the
    following components in our inertial coordinate system:
    \begin{align}
    \Pi^{\mu \nu} \eqdef c^{-2}e^{2\phi}u^{\mu}u^{\nu} + \Mg^{\mu \nu}.   \label{E:Projectionc}
    \end{align}
    $\Pi$ is the projection onto the orthogonal complement of $u:$
    \begin{align}
        \Pi^{\mu \nu}u^{\lambda}\Mg_{\lambda \mu} = 0 \qquad (\nu=0,1,2,3).
    \end{align}

    We now introduce the following Newtonian change of state-space
    variables\footnote{As suggested by Remark \ref{R:NewtonianR}, even though $\Rc$ is not a state-space variable, 
    equation \eqref{E:RcDef} also represents a Newtonian change of variables.}
    \begin{align}
        v^j &\eqdef {u^j}/{u^0} \qquad  (j= 1,2,3) \label{E:vdef}  \\
        \Phi &\eqdef c^2 \phi \label{E:PhiCHOVdef},
    \end{align}
    where $\mathbf{v}=(v^1,v^2,v^3)$ is the Newtonian velocity and
    $\Phi$ is the \emph{cosmological-Nordstr\"{o}m potential}. Relation \eqref{E:vdef} can be inverted to give
    \begin{align}
        u^0 &= e^{-\phi} \gamma_c \label{E:u0InTermsofv} \\
        u^j &= e^{-\phi} \gamma_c v^j \label{E:ujInTermsofv},
    \end{align}
    where
    \begin{align} \label{E:gammacDefI}
        \gamma_c(\mathbf{v}) \eqdef \frac{c}{(c^2 - |\mathbf{v}|^2)^{1/2}} .
    \end{align}

    \begin{remark}
        We provide here a brief elaboration on the Newtonian change of variables.
        Equation \eqref{E:vdef} provides the standard relationship
        between the Newtonian velocity $\mathbf{v}$ and the four-velocity $u$: if $x^{\nu}(t) \ (\nu = 0,1,2,3)$ are the 
        rectangular components of a timelike curve in $\mathcal{M}$ parameterized by $x^0=t,$ and $\tau$ denotes the proper 
        time parameter, then we have that $v^j = \partial_t x^j = (\partial \tau/\partial t )\cdot u^j = 
        u^j/u^0$ ($j=1,2,3$).

        Dimensional analysis suggests the approximate identification (for large $c$) of the 
        cosmological-Nordstr\"{o}m potential $\Phi$ from \eqref{E:PhiCHOVdef} and \eqref{E:ENkappac4} with the cosmological-Newtonian 
        potential $\Phi_{\infty}$, where $\Phi_{\infty}$ is the solution\footnote{We use the symbol
        $\Phi_{\infty}$ here to denote the solution to \eqref{E:EPkappa4} in order to distinguish 
        the cosmological-Newtonian potential from the cosmological-Nordstr\"{o}m potential.}
        to the non-relativistic equation \eqref{E:EPkappa4}: 
        $\Phi_{\infty}$ has the dimensions of $c^2,$ which suggests that when considering the limit $c \to \infty,$ 
        we should rescale the dimensionless cosmological-Nordstr\"{o}m potential $\phi$ as we did in \eqref{E:PhiCHOVdef}.
				Indeed, our main result, which is Theorem \ref{T:NewtonianLimit}, shows that with an 
				appropriate formulation of the initial value problems for the EN$_{\kappa}^c$ and EP$_{\kappa}$ systems, we have that $\lim_{c 
				\to \infty} \Phi = \Phi_{\infty}.$ Dimensional analysis also suggests the formal 
        identification of $\Rinfinity$ from \eqref{E:EPkappa2} - \eqref{E:QinfinityRelationship} with $
        \lim_{c \to \infty} \Rc = \lim_{c \to 
        \infty}\mathfrak{R}_{c}(\Ent,p)$ (for now assuming that this limit exists), where 
        $\mathfrak{R}_{c}(\Ent,p)$ is defined in \eqref{E:EnergyDensityc}. 
        
        Furthermore, these changes of variables can be justified through a formal expansion $c^{-2} \Phi \eqdef \phi =
        \phi_{(0)} + c^{-2} \phi_{(1)} + \cdots,$ $\Rc = R_{(0)} + c^{-2}
        R_{(1)} + \cdots,$ in powers of $c^{-2}$ in equation \eqref{E:ENkappac4}:
        equating the coefficients of powers of $c^{-2}$ on each
        side of the equation implies the formal identifications\footnote{Upon expansion, the formal equation satisfied by 
        $\phi_{(0)}$ is $(\Delta - \kappa^2)\phi_{(0)}=0,$ and by imposing vanishing boundary conditions at infinity, we conclude 
        that $\phi_{(0)}=0.$} $\phi_{(0)}=0$ and $(\Delta - \kappa^2)\phi_{(1)}=4 \pi GR_{(0)}.$ If we also consider
        equation \eqref{E:EPkappa4}, which reads $(\Delta - \kappa^2) \Phi_{\infty} = 4 \pi G \Rinfinity,$ then
        we are lead to the formal identifications $R_{(0)} \approx \Rinfinity$ and 
        $\Phi \eqdef c^2\phi \approx \phi_{(1)} \approx \Phi_{\infty}.$ A similar analysis for the 
        Vlasov-Nordstr\"{o}m system is carried out in \cite{sChL2004}.
    \end{remark}

    Upon making the substitutions \eqref{E:vdef} -
    \eqref{E:PhiCHOVdef} and lowering an index with $\Mg,$ the
    components of $\Pi$ in our inertial coordinate system read (for $1 \leq j,k \leq 3$):
    \begin{align}
        \Pi_0^0 &= - c^{-2}\gamma_c^2 |\mathbf{v}|^2 \label{E:Projection00}\\
        \Pi_j^0 &= c^{-2}\gamma_c^2 v^j    \label{E:Projectionj0}\\
        \Pi_0^j &= -\gamma_c^2 v^j    \label{E:Projection0j}\\
        \Pi_k^j &= c^{-2}\gamma_c^2 v^j v_k + \delta_k^j.   \label{E:Projectionjk}
    \end{align}
    Furthermore, we will also make use of the relation
    \begin{align} \label{E:gammaDifferentiated}
        \partial_{\lambda}\gamma_c = c^{-2} (\gamma_c)^3 v_k \partial_{\lambda} v^k \qquad (\lambda = 0,1,2,3).
    \end{align}

    Considering first the projection of \eqref{E:ENEMContinuityc} in the
    direction of $u,$ we remark that one may use \eqref{E:nandulaw} and \eqref{E:pressure}
    to conclude that for $C^1$ solutions, $u_{\nu}\partial_{\mu} \Theta^{\mu \nu}=0$ is equivalent to
    equation \eqref{E:ENkappac1}.

    We now project \eqref{E:ENEMContinuityc} onto the
    orthogonal complement of $u,$ which, with the aid of \eqref{E:phi}, gives the three equations
    $\Pi_{\nu}^{j} \partial_{\mu} \Theta^{\mu \nu}=0,$ $j=1,2,3:$
    \begin{align}  \label{E:ENEMContinutiyProjection}
        0&=\Pi_{\nu}^{j} \partial_{\mu} \Theta^{\mu \nu} =\Pi_{\nu}^{j} \big[\Rc + c^{-2}P\big]
            (e^{\phi}u^{\mu})\partial_{\mu}(e^{\phi}u^{\nu})+
            (\Pi_{\nu}^{j}\partial^{\nu}\phi) \frac{c^4}{4 \pi G}(\square \phi - \kappa^2 \phi) \\
        &=\Pi_{\nu}^{j} \big[\Rc + c^{-2}P\big]
            (e^{\phi}u^{\mu})\partial_{\mu}(e^{\phi}u^{\nu})+
            (\Pi_{\nu}^{j}\partial^{\nu}\Phi) (\Rc - 3c^{-2}P). \notag
    \end{align}
    After making the substitutions \eqref{E:PhiCHOVdef},
    \eqref{E:u0InTermsofv}, \eqref{E:ujInTermsofv}, and \eqref{E:gammacDefI}, and using relation
    \eqref{E:gammaDifferentiated}, it follows that for $C^1$ solutions, \eqref{E:ENEMContinutiyProjection} is
    equivalent to \eqref{E:ENkappac3}.

    We also introduce the nameless quantity $\Qc$ and make use of
    \eqref{E:pressure}, \eqref{E:SpeedofSoundc}, \eqref{E:SigmaSquaredc}, 
    \eqref{E:EnergyDensityc}, \eqref{E:partialRpartialpandStoNegativeTwoRelationshipc}, \eqref{E:RcDef}, 
    \eqref{E:PDefc}, and \eqref{E:PhiCHOVdef} to express it in the following form:
    \begin{align} \label{E:QcDef}
        \Qc \eqdef n \left. \frac{\partial P}{\partial n}\right|_{\Ent,\phi} &=
        \left. \frac{\partial P}{\partial (\rho/c^2)} \right|_{\Ent,\phi} \cdot n \left. \frac{\partial (\rho/c^2)}{\partial n}
        \right|_{\Ent} =\mathfrak{Q}_c(\Ent,p,\Phi),
    \end{align}
    where
    \begin{align} \label{E:QcFunctionDef}
        \mathfrak{Q}_c(\Ent,p,\Phi) \eqdef \mathfrak{S}_c^2(\Ent,p)e^{4 \Phi / c^2}[\mathfrak{R}_c(\Ent,p) + c^{-2}p]
        =\mathfrak{S}_c^2(\Ent,p)[\Rc + c^{-2}P].
    \end{align}
    Then we use the chain rule together with \eqref{E:nandulaw},
    \eqref{E:ENkappac1}, and \eqref{E:QcDef} to derive
    \begin{align} \label{E:PressureEvolution}
        e^{\phi}u^{\mu}\partial_{\mu}P + \Qc \partial_{\mu} (e^{\phi}u^{\mu}) = (4P-3\Qc) e^{\phi}u^{\mu} \partial_{\mu}\phi,
    \end{align}
    which we may use in place of \eqref{E:nandulaw}. Upon making
    the substitutions \eqref{E:RcDef}, \eqref{E:PDefc}, \eqref{E:PhiCHOVdef},
    \eqref{E:u0InTermsofv}, and \eqref{E:ujInTermsofv}, and using the relation \eqref{E:gammaDifferentiated},
    it follows that for $C^1$ solutions, \eqref{E:PressureEvolution} is equivalent to
    \eqref{E:ENkappac2}.

\section{The Formal Limit $c \to \infty$ of the EN$_\kappa^c$ System} \label{S:FormalLimit}
	For convenience, in this section we list the final form of the EN$_\kappa^c$ system as derived in sections
    \ref{SS:ENkappacDerivation} and \ref{SS:NewtonianReformulationc}. We also take the formal limit
    $c \to \infty$ to arrive at the EP$_{\kappa}$ system and introduce the equations of variation (EOV$_{\kappa}^c$).

    \subsection{A recap of the EN$_{\kappa}^c$ system} 
		The EN$_{\kappa}^c$ system is given by
        \begin{align}
            & \partial_t \Ent + v^k \partial_k \Ent =0        \label{E:ENkappac1}       \\
            & \partial_t P + v^k \partial_k P + \Qc \partial_k v^k + c^{-2}(\gamma_c)^2 \Qc v_k \big(\partial_t
                v^k + v^a \partial_a v^k \big) \label{E:ENkappac2}    \\
            & \hspace{.72in} = (4P- 3\Qc)\big[c^{-2}\partial_t \Phi + c^{-2} v^k \partial_k\Phi \big] \notag \\
            & (\gamma_c)^2(\Rc + c^{-2}P)\big[\partial_t v^j + v^k \partial_k
                v^j + c^{-2}(\gamma_c)^2 v^j v_k (\partial_t v^k + v^a \partial_a v^k)\big] \label{E:ENkappac3}      \\
            & \hspace{.90 in} + \partial_j P + c^{-2}(\gamma_c)^2 v^j(\partial_t P + v^k \partial_k P) \notag \\
            	& \hspace{.72in} = (3c^{-2}P - \Rc)\big(\partial_j \Phi + (\gamma_c)^{-2}v^j[c^{-2}\partial_t \Phi + c^{-2}v^k \partial_k
                \Phi]\big)  \notag \\
            & -c^{-2} \partial_t^2 \Phi + \Delta \Phi - \kappa^2
                \Phi = 4\pi G(\Rc - 3c^{-2}P),                                     \label{E:ENkappac4}
    \end{align}
    where $j=1,2,3,$
    \begin{align}
      	\gamma_c &= \gamma_c(\mathbf{v}) \eqdef \frac{c}{(c^2 - |\mathbf{v}|^2)^{1/2}} \label{E:gammacDefII}\\
       	\Rc & \eqdef e^{4 \Phi / c^2}\mathfrak{R}_c(\Ent,p) \label{E:RcDefII} \\
        \Qc &\eqdef \mathfrak{Q}_c(\Ent,p,\Phi) \eqdef \mathfrak{S}_c^2(\Ent,p)e^{4 \Phi / c^2}[\mathfrak{R}_c(\Ent,p) + c^{-2}p]  
        		\label{E:QcDefII} \\
       	& = \Big(\left. \frac{\partial \mathfrak{R}_c}{\partial p}(\Ent,p)\right|_{\Ent}\Big)^{-1} e^{4 \Phi / 
        	c^2}[\mathfrak{R}_c(\Ent,p) + c^{-2}p]  \notag \\
       	P &\eqdef e^{4 \Phi / c^2}p \label{E:PDefcII},
    \end{align}
    $c$ denotes the speed of light, $\mathfrak{S}_c (\Ent,p),$ which is defined in 
    \eqref{E:partialRpartialpandStoNegativeTwoRelationshipc}, is the 
    speed of sound, and the functions $\mathfrak{R}_c$ and $\mathfrak{S}_c$ derive from a $c-$indexed equation of state
    as discussed in Section \ref{SS:ENkappacDerivation}. The
    variables $\Ent,p,\mathbf{v}=(v^1,v^2,v^3),$ and $\Phi$ denote the
    entropy per particle, pressure, (Newtonian) velocity, and cosmological-Nordstr\"{o}m potential
    respectively. Section \ref{S:cDependence} contains a detailed discussion of the
    $c$-dependence of the EN$_{\kappa}^c$ System.

 \subsection{The EP$_{\kappa}$ system as a formal limit}    \label{SS:EPkappac} 
 		Taking the formal limit $c \to \infty$ in the EN$_{\kappa}^c$ system
    gives the Euler-Poisson system with a cosmological constant:
    \begin{align}
        \partial_t {\Ent} + v^k\partial_k {\Ent} &=  0                      \label{E:EPkappa1}\\
        \partial_t p + v^k \partial_k p + Q_{\infty}\partial_k v^k &= 0 \label{E:EPkappa2}\\
        \partial_t{\Rinfinity} +  \partial_k(\Rinfinity v^k) &=  0 \tag{\ref{E:EPkappa2}'}              \\
        \Rinfinity \big(\partial_t v_j + v^k \partial_k v^j\big)
        + \partial_j p  &= - \Rinfinity \partial_j \Phi \qquad  (j=1,2,3) \label{E:EPkappa3}\\
        \Delta \Phi - \kappa^2\Phi &= 4 \pi G \Rinfinity,                                  \label{E:EPkappa4}
    \end{align}
    where
    \begin{align}
        &\Rinfinity \eqdef \mathfrak{R}_{\infty}(\Ent,p) \label{E:RinfinityDef} \\
        &\Qinfinity \eqdef \mathfrak{Q}_{\infty}(\Ent,p) \eqdef \mathfrak{S}_{\infty}^2(\Ent,p)
        \mathfrak{R}_{\infty}(\Ent,p) = \Big(\left. \frac{\partial \mathfrak{R}_{\infty}}{\partial p}(\Ent,p)\right|_{\Ent}\Big)^{-1}
        \mathfrak{R}_{\infty}(\Ent,p), \label{E:QinfinityRelationship}
    \end{align}
    $\mathfrak{R}_{\infty}(\Ent,p)$ and $\mathfrak{S}_{\infty}^2(\Ent,p)$
    are the limits as $c \to \infty$ of $\mathfrak{R}_{c}(\Ent,p)$ and $\mathfrak{S}_{c}^2(\Ent,p)$
    respectively (see \eqref{E:EOScHypothesis1}, \eqref{E:EOScHypothesis2}, and \eqref{E:SpeedofSoundcDependence}), and
    $\Rinfinity$ is the Newtonian mass density. Since equations \eqref{E:partialRpartialpandStoNegativeTwoRelationshipc}
    and \eqref{E:SpeedofSoundcDependence} imply that
    $\partial \mathfrak{R}_{\infty}(\Ent,p) / {\partial p} = \mathfrak{S}_{\infty}^{-2}(\Ent,p),$
    it then follows with the aid of the chain rule that for $C^1$ solutions, equations \eqref{E:EPkappa2} and (\ref{E:EPkappa2}') are
    equivalent. We refer to the solution variable $\Phi$ from equation \eqref{E:EPkappa4} as the \emph{cosmological-Newtonian potential}.
    
    An introduction to the EP$_{\kappa}$ system can be found in
    \cite{mK2003}. In this article, Kiessling assumes an \emph{isothermal} equation of
    state ($p=c_s^2 \cdot \Rinfinity,$ where the constant $c_s$ denotes the speed of sound), and derives the
    Jeans dispersion relation that arises from linearizing
    (\ref{E:EPkappa2}'), \eqref{E:EPkappa3}, \eqref{E:EPkappa4} about a static state in which the
    background Newtonian mass density $\bar{R}_{\infty}$ is positive, followed by taking the limit $\kappa \to 0.$

    It is a standard result that the solution to \eqref{E:EPkappa4} is given by
    \begin{align}                                                       \label{E:EPkappaPotential}
        \Phi(t,\sforspace) = \Phibar_{\infty} -  G \int_{\mathbb{R}^3} \left( \frac{e^{-\kappa |\sforspace - \sforspace'|}}{|\sforspace -
        \sforspace'|} \right) \big[\mathfrak{R}_{\infty}(\Ent(t,\sforspace'),p(t,\sforspace')) - \mathfrak{R}_{\infty}(\Entbar,\pbar) \big]\, d^d \sforspace',
    \end{align}
    where the constants $\Phibar_{\infty}, \Entbar,$ and $\pbar,$ which are the values
    of $\Phi,$ $\Ent,$ and $p$ respectively in a constant background state, are discussed in Section
    \ref{S:IVPc}. The boundary conditions leading to this solution are that $\Phi(t,\cdot) - \Phibar_{\infty}$ vanishes at $\infty,$
    and we view $\Phi(t,\sforspace)$ as a (not necessarily small) perturbation of the constant potential $\Phibar_{\infty}.$

    \begin{remark} \label{R:KernelRemark}
        Consider the kernel $\mathcal{K}(\sforspace)=-G e^{-\kappa
        |\sforspace|}/|\sforspace|$ appearing in
        \eqref{E:EPkappaPotential}. An easy computation gives that
        $\mathcal{K}(\sforspace), \partial\mathcal{K}(\sforspace) \in
        L^1(\mathbb{R}^3).$ Therefore, a basic result from harmonic
        analysis (Young's inequality) implies that the map $f
        \rightarrow \mathcal{K} * f,$ where $*$ denotes
        convolution, is a bounded linear map\footnote{Our proof breaks down at this point in the case $\kappa = 0.$} from 
        $L^2(\mathbb{R}^3)$ to
        $H^{1}(\mathbb{R}^3).$ From this fact and Remark
        \ref{R:SobolevTaylorCalculusRemark} (alternatively consult Lemma \ref{L:HjbySubtractingConstant}), it follows that \\
        $\Phi(t,\cdot) \in H_{\Phibar}^{N+1}(\mathbb{R}^3)$ whenever
        $\big(\Ent(t,\cdot),p(t,\cdot)\big) \in
        H_{\Entbar}^{N}(\mathbb{R}^3) \times H_{\pbar}^{N}(\mathbb{R}^3).$
        By then applying Lemma \ref{L:AlterateHNnormestimate}, we can
        further conclude that $\Phi(t,\cdot) \in
        H_{\Phibar}^{N+2}(\mathbb{R}^3)$ whenever
        $\big(\Ent(t,\cdot),p(t,\cdot)\big) \in
        H_{\Entbar}^{N}(\mathbb{R}^3) \times H_{\pbar}^{N}(\mathbb{R}^3).$
    \end{remark}

\section{The Equations of Variation (EOV$_{\kappa}^c$)} \label{SS:EOVc}

        The EOV$_{\kappa}^c$ are formed by linearizing the EN$_\kappa^c$
        system (EP$_{\kappa}$ system if $c = \infty$) around a background solution (BGS) $\Vboldt$ of the form 
        $\Vboldt = (\widetilde{{\Ent}},\widetilde{P},\widetilde{v}^1,
        \cdots, \widetilde{\Phi}_2, \widetilde{\Phi}_3).$ Given such a $\Vboldt$ and inhomogeneous terms
        $f,g,h^{(1)},h^{(2)},h^{(3)},l,$ we define the EOV$_{\kappa}^c$ by
    \begin{align}
             \partial_t \dot{\Ent} + \wtv^k \partial_k \dot{\Ent} &= f        \label{E:EOVc1}       \\
             \partial_t \dot{P} +\wtv^k \partial_k \dot{P} + \wtQc \partial_k \dot{v}^k \label{E:EOVc2}
                + c^{-2}(\wtg)^2 \wtQc \wtv_k\big(\partial_t \dot{v}^k  + \wtv^a \partial_a \dot{v}^k \big)  &= g   \\
             (\wtg)^2(\wtRc + c^{-2}\wtP)\big[\partial_t \dot{v}^j + \wtv^k \partial_k
                \dot{v}^j + c^{-2}(\wtg)^2\wtv^j \wtv_k (\partial_t \dot{v}^k + \wtv^a \partial_a \dot{v}^k)\big]
                         \label{E:EOVc3}      \\
             + \partial_j \dot{P} + c^{-2}(\wtg)^2 \wtv^j(\partial_t \dot{P} + \wtv^k \partial_k \dot{P})  &=
                h^{(j)}   \notag \\
             -c^{-2} \partial_t^2 \Phidot + \Delta \Phidot - \kappa^2
                \Phidot &= l,                                  \label{E:EOVcKlein-Gordon}
        \end{align}
        where $\wtg \eqdef c/(c^2 - |\wtv|^2)^{1/2}, \widetilde{R}_c \eqdef
        e^{4 \Phit / c^2}\mathfrak{R}_c(\widetilde{\Ent},\widetilde{p}),$ etc.
        The unknowns are the components of $\Wdot \eqdef (\Entdot,\Pdot, \dot{v}^1, \dot{v}^2,
        \dot{v}^3)$ and $\Phidot.$

        \begin{remark}
                We place parentheses around the superscripts of the inhomogeneous terms $h^{(j)}$ in order to emphasize that
                we are merely labeling them, and that in general, we do not associate any transformation properties to them under 
                changes of coordinates.
            \end{remark}
				
				\subsection{PDE matrix/vector notation}
				
        Let us now provide a few remarks on our notation. We find it useful to analyze both the dependent variable $p$ and the 
        dependent variable $P$ when discussing solutions to \eqref{E:ENkappac1} - \eqref{E:ENkappac4}. Therefore, we will make use of all 
        four of the following arrays:
            \begin{align}
                \Wbold &\eqdef (\Ent,P,v^1,v^2,v^3) \label{E:Wboldarray}\\
                \Vbold & \eqdef (\Ent,P,v^1,v^2,v^3, \Phi, \partial_t \Phi, \partial_1 \Phi, \partial_2 \Phi, \partial_3 \Phi) 
                		\label{E:Vboldarray}\\
                \scrW &\eqdef (\Ent,p,v^1,v^2,v^3) \label{E:scrWboldarray} \\
                \scrV & \eqdef (\Ent,p,v^1,v^2,v^3, \Phi, \partial_t \Phi, \partial_1 \Phi, \partial_2
                \Phi, \partial_3\Phi), \label{E:scrVboldarray}
            \end{align}
        where $P \eqdef e^{4 \Phi / c^2}p.$ When discussing a BGS $\Vboldt \eqdef (\widetilde{{\Ent}},\widetilde{P},\widetilde{v}^1,
        \cdots, \widetilde{\Phi}_2, \widetilde{\Phi}_3)$ that defines the coefficients of the unknowns in the 
        EOV$_{\kappa}^c,$ we also use notation similar to that used in \eqref{E:Wboldarray} - \eqref{E:scrVboldarray}, including 
        $\scrVtilde \eqdef (\widetilde{\Ent},\widetilde{p},\widetilde{v}^1,\cdots,\partial_3
        \widetilde{\Phi}),$ $\Wboldt \eqdef (\widetilde{\Ent},\widetilde{P},\widetilde{v}^1,\widetilde{v}^2,\widetilde{v}^3),$
        where $\widetilde{p} \eqdef e^{-4\Phit/c^2}\widetilde{P},$ etc. When $c=\infty,$ we may also refer to
        $\scrWtilde \eqdef (\widetilde{{\Ent}},\widetilde{p},\widetilde{v}^1,\widetilde{v}^2,\widetilde{v}^3)$ as the BGS, since in
        this case, the left-hand sides of \eqref{E:EOVc1} - \eqref{E:EOVcKlein-Gordon} do not depend on $\widetilde{\Phi},$ and 
        furthermore, $\Wboldt = \scrWtilde.$ Additionally, we may refer to the unknowns in the EOV$_{\kappa}^c$ as $\scrWdot \eqdef 
        (\Entdot,\dot{p},\dot{v}^1, \dot{v}^2,\dot{v}^3)$ when $c = \infty;$ in this article, $\Phidot$ will always vanish at
        infinity, and in the case $c=\infty,$ rather than considering $\Phidot$ to be an ``unknown,'' we assume that the 
        solution variable $\Phidot$ has been constructed via the convolution 
        $\Phidot = \mathcal{K}*l,$ where the kernel $\mathcal{K}(\sforspace)$ is defined in Remark \ref{R:KernelRemark}, and $l$ 
        is the right-hand side of \eqref{E:EOVcKlein-Gordon}.

        We frequently adopt standard PDE matrix/vector notation. For example, we may write \eqref{E:ENkappac1} - \eqref{E:ENkappac3} as
            \begin{align}
                \Ac^{\mu}(\scrW,\Phi) \partial_{\mu} \Wbold = \mathbf{b},           \label{E:MatrixDef}
            \end{align}
        where each $\Ac^{\nu}(\cdot)$ is a $5 \times 5$ matrix with entries that are functions of $\scrW$ and $\Phi$, while 
        $\mathbf{b}=(f,g,\cdots,h^{(3)})$ is
        the 5-component column array on the right-hand side of \eqref{E:ENkappac1} - \eqref{E:ENkappac3}. It is instructive to see the 
        form of the $\Ac^{\nu}(\cdot),$ $\nu=0,1,2,3,$ for we will soon concern ourselves with their large$-c$ asymptotic
        behavior. Abbreviating $\alpha_{c} \eqdef (\gammac)^2\big(\Rc + c^{-2}P\big),$ \
        $\beta_{c}^{(i)} \eqdef c^{-2} (\gammac)^2 v^i, \ \beta_{c}^{(i,j)} \eqdef c^{-2} (\gammac)^2 v^i v^j,$ we have that

 \begin{align} \label{E:A0c}
                & {_c\mathscr{A}}^0(\scrW,\Phi) = 
                    & \begin{pmatrix}
                        1 & 0 & 0 & 0 &0     \\
                        0 & 1 & Q_c \beta_{c}^{(1)}
                            &  Q_c \beta_{c}^{(2)}
                            &  Q_c \beta_{c}^{(3)} \\
                        0 & Q_c \beta_{c}^{(1)}
                            & \alpha_{c}(1 + \beta_{c}^{(1,1)})
                            & \alpha_{c} \beta_{c}^{(1,2)}
                            & \alpha_{c} \beta_{c}^{(1,3)}\\
                        0 & Q_c \beta_{c}^{(2)}
                            & \alpha_{c} Q_c \beta_{c}^{(2,1)}
                            & \alpha_{c}(1 + \beta_{c}^{(2,2)})
                            & \alpha_{c} \beta_{c}^{(2,3)} \\
                        0 & Q_c \beta_{c}^{(3)}
                            & \alpha_{c} \beta_{c}^{(3,1)}
                            & \alpha_{c} \beta_{c}^{(3,2)}
                            & \alpha_{c}(1 + \beta_{c}^{(3,3)})
                    \end{pmatrix}
            \end{align}

            \begin{align} \label{E:A0Infinity}
                & {_{\infty}\mathscr{A}}^0(\scrW) = 
                    & \begin{pmatrix}
                        1 & 0 & 0 & 0 &0     \\
                        0 & 1 & 0
                            &  0
                            &  0 \\
                        0 & 0
                            &  \Rinfinity
                            & 0
                            & 0\\
                        0 & 0
                            & 0
                            &  \Rinfinity
                            & 0 \\
                        0 & 0
                            &0
                            & 0
                            & \Rinfinity
                    \end{pmatrix}
            \end{align}
    \begin{align} \label{E:Ajc}
                & {_c\mathscr{A}}^1(\scrW,\Phi) = 
                    & \begin{pmatrix}
                        v^1 & 0 & 0 & 0 &0     \\
                        0 & v^1 & Q_c (1 + \beta_{c}^{(1,1)})
                            & Q_c \beta_{c}^{(1,2)} 
                            & Q_c \beta_{c}^{(1,3)} \\
                        0 & 1 + \beta_{c}^{(1,1)}
                            &  \alpha_{c} v^1 (1 + \beta_{c}^{(1,1)})
                            & \alpha_{c} v^1 \beta_{c}^{(1,2)}
                            & \alpha_{c} v^1 \beta_{c}^{(1,3)}                         \\
                        0 & \beta_{c}^{(2,1)}
                            & \alpha_{c} v^1 \beta_{c}^{(2,1)}
                            &  \alpha_{c} v^1 (1 + \beta_{c}^{(2,2)})
                            & \alpha_{c} v^1 \beta_{c}^{(2,3)} \\
                        0 & \beta_{c}^{(3,1)}
                            & \alpha_{c} v^1 \beta_{c}^{(3,1)}
                            & \alpha_{c} v^1 \beta_{c}^{(3,2)}
                            &\alpha_{c} v^1 (1 + \beta_{c}^{(3,3)})
                    \end{pmatrix}
            \end{align}

\begin{align} \label{E:AjInfinity}
                & {_{\infty} \mathscr{A}}^1(\scrW) = 
                    & \begin{pmatrix}
                        v^1 & 0 & 0 & 0 &0     \\
                        0 & v^1 & \Qinfinity
                            &  0
                            &  0 \\
                        0 & 1
                            &   \Rinfinity v^1
                            & 0
                            & 0                               \\
                        0 & 0
                            & 0
                            & \Rinfinity v^1
                            & 0 \\
                        0 & 0
                            & 0
                            & 0
                            & \Rinfinity v^1
                    \end{pmatrix},
\end{align}
and similarly for ${_c\mathscr{A}}^2(\scrW,\Phi), \ {_{\infty} \mathscr{A}}^2(\scrW), \ {_c\mathscr{A}}^3(\scrW,\Phi),$ and ${_{\infty} \mathscr{A}}^3(\scrW).$

\section{On the $c-$Dependence of the EN$_{\kappa}^c$ System}               \label{S:cDependence}

      	In addition to appearing directly as the term $c^{-2},$ the constant $c$ appears in equations
        \eqref{E:ENkappac1} - \eqref{E:ENkappac4} through four terms: $\mathbf{i)} \ P = e^{4 \Phi / c^2}p,$ 
        $\mathbf{ii)} \ \gamma_c = c/(c^2 - |\mathbf{v}|^2)^{1/2},$ \\ $\mathbf{iii)} \ \Rc = e^{4 \Phi / c^2}\mathfrak{R}_c(\Ent,p),$ 
        and $\mathbf{iv)} \ \Qc= \mathfrak{S}_c^2(\Ent,p)e^{4 \Phi / c^2}[\mathfrak{R}_c(\Ent,p) + c^{-2}p].$ Because we want to recover 
        the EP$_{\kappa}$ system in the large $c$ limit, the first obvious requirement we have is that the function 
        $\mathfrak{R}_c(\Ent,p)$
        has a limit $\mathfrak{R}_{\infty}(\Ent,p)$ as $c \rightarrow \infty.$ For
        mathematical reasons, we will demand convergence in the norm $|\cdot|_{N+1,\mathfrak{C}}$ (see definition \eqref{E:CbkNormDef}) 
        at a rate of order $c^{-2},$ where $\mathfrak{C}$ is a compact subset of $\mathbb{R}^{+} \times \mathbb{R}^{+}$ that depends on 
        the Newtonian initial data $\scrV_{\infty}$ defined in \eqref{E:EPkappaData}; see \eqref{E:EOScHypothesis1} and 
        \eqref{E:EOScHypothesis2}. Although a construction of 
        $\mathfrak{C}$ is described in detail in Section \ref{SS:AdmissibleStateSpacec},
        let us now provide a preliminary description that is sufficient for our current purposes: for given initial data, we will 
        prove the existence of compact sets $\bar{\mathcal{O}}_2, \bar{\mathfrak{O}}_2,$ $[-a,a]^5, K \eqdef \bar{\mathcal{O}}_2 \times 
        [-a,a]^5,$ $\mathfrak{K} \eqdef \bar{\mathfrak{O}}_2 \times [-a,a]^5,$ and a time interval $[0,T]$ so that for all large $c,$ the 
        ($c-$dependent) solutions\footnote{Recall the notation \eqref{E:Wboldarray} - \eqref{E:scrVboldarray} which defines the arrays 
        $\Wbold, \Vbold, \scrW,$ and $\scrV$ respectively.} $\Vbold$ ($\scrV$) to 
        the EN$_{\kappa}^c$ system launched by the initial data exist on $[0,T] \times \mathbb{R}^3$ and satisfy $\Wbold([0,T] \times 
        \mathbb{R}^3) \subset \bar{\mathcal{O}}_2,$ $\scrW([0,T] \times \mathbb{R}^3) \subset 
        \bar{\mathfrak{O}}_2,$ $\Vbold([0,T] \times \mathbb{R}^3) \subset K,$ and $\scrV([0,T] \times \mathbb{R}^3) \subset 
        \mathfrak{K}.$ See Section \ref{SS:AdmissibleStateSpacec} for a detailed description of 
        $\bar{\mathcal{O}}_2$ and $\bar{\mathfrak{O}}_2,$ and \eqref{E:Kdef}, \eqref{E:frakKdef} for 
        the construction of $K$ and $\mathfrak{K}.$
        
       	The set $\mathfrak{C}$ from above, then, is the projection of $\bar{\mathfrak{O}}_2$ onto the first two axes (which
       	are the $\Ent,p$ components of $\scrV$). Intuitively, we would like the aforementioned four functions of the state-space 
       	variables to converge to $p,1,\Rinfinity,$ and $\Qinfinity$ respectively when their domains are restricted to an appropriate 
       	compact subset. 
       	In this section, we will develop and then assume hypotheses on the $c-$indexed equation of state that will allow us to prove 
       	useful versions of these kinds of convergence results.

        \subsection{Functions with $c-$independent properties: the definitions} 
        The main technical difficulty that we must confront is ensuring that
        the Sobolev estimates provided by the propositions appearing in Appendix \ref{A:SobolevMoser}
        can be made independently of all large $c.$ By examining these propositions,
        one could anticipate that this amounts to analyzing the $C_b^j$ norms (see definition \eqref{E:CbkNormDef})
        of various $c-$indexed families of functions $\mathfrak{F}_c$ appearing in the family of EN$_{\kappa}^c$ systems.
        We therefore introduce here some machinery that will allow us to easily discuss uniform-in-$c$ estimates. Following this, we
        use this machinery to prove some preliminary lemmas that will be used in the proofs of
        Theorem \ref{T:UniformLocalExistenceENkappac} and Theorem \ref{T:NewtonianLimit}, which are the two main theorems
        of this article. Before proceeding, we refer the reader to the notation defined in \eqref{E:leqcdef},
        which will be used frequently in the discussion that follows.

        \begin{definition} \label{D:RingDef}
          	Let $y^1,\cdots,y^n$ denote Cartesian coordinates on $\mathbb{R}^n,$ and let \\ 
          	$\mathfrak{D} \subset \mathbb{R}^n$ be a compact convex set.
          	We define $\mathcal{R}^j(c^{k};\mathfrak{D};y^1,\cdots,y^n)$ to be the ring consisting
            of all $c-$indexed families of functions $\mathfrak{F}_c(y^1,\cdots,y^n)$ such that for all large $c,$ $\mathfrak{F}_c 
            \in C_b^j(\mathfrak{D}),$ and such that the following estimate holds:
            \begin{align} \label{E:FcDecay}
                |\mathfrak{F}_c|_{j,\mathfrak{D}} \leqc c^{k} \cdot C(\mathfrak{D}).
            \end{align}
            We emphasize that the constant $C(\mathfrak{D})$ is
            allowed to depend on the family $\mathfrak{F}_c$ and the set
            $\mathfrak{D},$ but within a given family and on a fixed
            set, $C(\mathfrak{D})$ must be independent of all large $c.$
    	\end{definition}
      
      \begin{definition}  \label{D:qRingDef}
            Let $\mathfrak{D} \subset \mathbb{R}^n$ be a compact convex set.
            Let $q_1, \cdots, q_n$ be functions such that
            $(q_1,\cdots,q_n) \in H_{\bar{q}_1}^{j}(\mathbb{R}^3) \times \cdots \times 
            H_{\bar{q}_n}^{j}(\mathbb{R}^3)$ (see definition \eqref{E:NJVbNormDef}) and such that $\lbrace 
            \big(q_1(\sforspace),q_2(\sforspace),\cdots,q_n(\sforspace) \big)
            \ | \ \sforspace \in \mathbb{R}^3 \rbrace \subset \mathfrak{D},$ 
            where $\bar{q}_1, \bar{q}_2, \cdots, \bar{q}_n$ are constants
            such that $(\bar{q}_1, \bar{q}_2, \cdots, 
            \bar{q}_n) \in \mathfrak{D}.$ We define
            $\mathcal{R}^j(c^{k};\mathfrak{D};q_1,\cdots,q_n)$ to be the ring consisting
            of all $c-$indexed expressions that can be written as the composition of an element of 
            $\mathcal{R}^j(c^{k};\mathfrak{D};y^1,\cdots,y^n)$ with $(q_1, \cdots, q_n).$ 
            
            If $\mathfrak{F}_c$ is such an expression, then we indicate this by writing
            \begin{align} \label{E:qFcnotation}
            	\mathfrak{F}_c(q_1, \cdots, q_n) & \in \mathcal{R}^j(c^{k};\mathfrak{D};q_1,\cdots,q_n) \\
            	& or \notag \\
            	\mathfrak{F}_c & \in \mathcal{R}^j(c^{k};\mathfrak{D};q_1,\cdots,q_n). \label{E:qFcnotation2}
            \end{align}
           	We remark that the notation \eqref{E:qFcnotation}, \eqref{E:qFcnotation2} also carries
           	with it the implication that the functions $(q_1, \cdots, q_n)$ have the aforementioned properties. 
        \end{definition}

        \begin{remark}
        	The notation $\mathfrak{F}_c \in 
        	\mathcal{R}^j(c^{k};\mathfrak{D};q_1,\cdots,q_n)$ represents an abuse of notation in the sense that in Definition
        	\ref{D:RingDef}, the arguments of the function $\mathfrak{F}_c(y^1,\cdots,y^n)$ are fixed, while in 
        	Definition \ref{D:qRingDef}, we are allowing ourselves the freedom to shift the point of view as 
        	to what are the arguments of the expression $\mathfrak{F}_c$ by allowing ourselves to ``shift around powers of $c.$'' 
        	At the beginning of Section \ref{SS:ApplicationtoENkappac}, we 
        	explain why this freedom can be useful. As a simple example, if $\partial_t \Phi \in H^2,$ 
        	$\|\partial_t \Phi \|_{L^{\infty}} \leq 1,$ and $\mathfrak{F}_c = c^{-2} \partial_t \Phi,$
       		then we have that $\mathfrak{F}_c \in \mathcal{R}^2(c^{-2};[-1,1];\partial_t \Phi)$
       		and also that $\mathfrak{F}_c \in \mathcal{R}^2(c^{-1};[-1,1];c^{-1} \partial_t \Phi).$
      	\end{remark}

     		\begin{definition} 						\label{D:qIDef}
     				Let $\mathfrak{D}, q_1, \cdots, q_n,$ and $\bar{q}_1, \bar{q}_2, \cdots, \bar{q}_n$ be as in
            Definition \ref{D:qRingDef}. Then we define $\mathcal{I}^j(c^{k};\mathfrak{D};q_1,\cdots,q_n)$ to be the
            sub-ring contained in $\mathcal{R}^j(c^{k};\mathfrak{D};q_1,\cdots,q_n)$ consisting of all such $c-$indexed
            expressions $\mathfrak{F}_c$ such that the following estimate holds:
            \begin{align} \label{E:IDef}
                \|\mathfrak{F}_c \|_{H^j} \leqc
                c^{k} \cdot C(\mathfrak{D};\|q_1\|_{H_{\bar{q}_1}^j},\cdots,\|q_n\|_{H_{\bar{q}_n}^j}).
            \end{align}

        		If $\mathfrak{F}_c$ is such an expression, then we indicate this by writing
            \begin{align} \label{E:Fcnotation}
            	\mathfrak{F}_c(q_1, \cdots, q_n) & \in \mathcal{I}^j(c^{k};\mathfrak{D};q_1,\cdots,q_n) \\
            	& or \notag \\
            	\mathfrak{F}_c & \in \mathcal{I}^j(c^{k};\mathfrak{D};q_1,\cdots,q_n) \notag \\
            	& or \notag \\
            	\mathfrak{F}_c & = \mathscr{O}^j(c^{k};\mathfrak{D};q_1,\cdots,q_n). \label{E:ONotation}
            \end{align} 
        \end{definition}

				\begin{remark}
					This definition is highly motivated by the inequality \eqref{E:ModifiedSobolevEstimateConstantArray} of Appendix 
						\ref{A:SobolevMoser}.
				\end{remark}
				
				\begin{remark}
            We also emphasize that in our applications below, \emph{the functions $q_i$ and constants $\bar{q}_i$ may themselves depend 
            on the parameter $c,$ even though we do not always explicitly indicate this dependence.} Typically, the $q_i$ will be 
            quantities related to solutions of the EN$_{\kappa}^c$ system, and the $\bar{q}_i$ will be equal to the components of 
            either \eqref{E:InitialConstant}, \eqref{E:scrVcConstant}, or \eqref{E:VcConstant}, perhaps scaled by a power of $c.$
        \end{remark}

        \begin{remark} \label{R:ArgumentsOmitted}
            In the notation $\mathcal{R}(\cdots),\mathcal{I}(\cdots),$ and $\mathscr{O}^j(\cdots),$
            we often omit the argument $\mathfrak{D}.$ In this case, it is understood that
            there is an implied set $\mathfrak{D}$ that is to be inferred from context; frequently
            $\mathfrak{D}$ is to be inferred from $L^{\infty}$ estimates on the $q_i$ that follow from
            Sobolev embedding. Also, we omit the argument $c^{k}$ when $k=0.$ Furthermore, we have chosen to omit
            dependence on the constants $\bar{q}_i$ since, as will be explained at the beginning of
            Section \ref{SS:ApplicationtoENkappac}, their definitions will be clear from context. We
            will occasionally omit additional arguments when the context is clear.
        \end{remark}

				\subsection{Functions with $c-$independent properties: useful lemmas}
				The following three lemmas provide the core structure for analyzing the Sobolev norms of terms
				appearing in the EN$_{\kappa}^c$ system. They are especially useful for keeping track of powers of
				$c.$ Their proofs are based on the Sobolev-Moser estimates that are stated as propositions in Appendix \ref{A:SobolevMoser}.
				We assume throughout this section that the functions $q_1,\cdots,q_n$ have the properties stated in Definition \ref{D:qRingDef}.
				
				\begin{lemma} \label{L:HjbySubtractingConstant}
            If $j \geq 2$ and $\mathfrak{F}_c(y^1,\cdots,y^n) \in \mathcal{R}^j(c^{k};\mathfrak{D};y^1,\cdots,y^n),$
            then 
           	\begin{align}
            	\mathfrak{F}_c \circ (q_1,\cdots,q_n) - \mathfrak{F}_c \circ(\bar{q}_1,\cdots,\bar{q}_n) \in 
            	\mathcal{I}^j(c^{k};\mathfrak{D};q_1,\cdots,q_n).
            \end{align}
        \end{lemma}
        \begin{proof}
        		We emphasize that the conclusion of Lemma \ref{L:HjbySubtractingConstant} is exactly the statement that 
        		$\|\mathfrak{F}_c \circ(q_1,\cdots,q_n) - \mathfrak{F}_c \circ(\bar{q}_1,\cdots,\bar{q}_n)\|_{H^j} \leqc
        		c^{k} \cdot C(\|q_1\|_{H_{\bar{q}_1}^j},\cdots,\|q_n\|_{H_{\bar{q}_n}^j}).$ Its proof follows from definitions 
        		\ref{D:RingDef}, \ref{D:qRingDef}, and \ref{D:qIDef}, and from \eqref{E:ModifiedSobolevEstimateConstantArray}.
        \end{proof}

        \begin{lemma} \label{L:MultiplyFcGcSobolevEstimate}
            Suppose that $\mathfrak{F}_c \in \mathcal{R}^j(c^{k_1};\mathfrak{D};q_1,\cdots,q_n),$ \\
            $\mathfrak{G}_c \in \mathcal{R}^j(c^{k_2};\mathfrak{D};q_1,\cdots,q_n),$ and $\mathfrak{H}_c \in 
             \mathcal{I}^j(c^{k_3};\mathfrak{D};q_1,\cdots,q_n).$ Then
            \begin{align}
            	\mathfrak{F}_c \cdot \mathfrak{G}_c & \in 
            		\mathcal{R}^j(c^{k_1+k_2};\mathfrak{D};q_1,\cdots,q_n) \qquad \mbox{if}  \ j \geq 0							\\
            		& \mbox{and} \notag \\
            	\mathfrak{F}_c \cdot \mathfrak{H}_c & \in 
            		\mathcal{I}^j(c^{k_1+k_3};\mathfrak{D};q_1,\cdots,q_n) \qquad \mbox{if}  \ j \geq 2.
            \end{align}
       	\end{lemma}    
        \begin{proof}
            Lemma \ref{L:MultiplyFcGcSobolevEstimate} follows from the product rule for derivatives and \eqref{E:ModifiedSobolevEstimate}.
        \end{proof}

        \begin{remark} \label{R:RingRemark}
            Lemma \ref{L:MultiplyFcGcSobolevEstimate} shows that for $k \leq 0,$ $\mathcal{R}^j(c^{k};\mathfrak{D};q_1,\cdots,q_n)$ is a 
            ring, i.e., it is closed under products. We frequently use this property in this article without explicitly mentioning it.
        \end{remark}

        \begin{remark} \label{R:1OverFisintheRing}
           Lemma \ref{L:MultiplyFcGcSobolevEstimate} can easily be used to show that if \\ $\mathfrak{F}_c(y^1,\cdots,y^n) \in
           \mathcal{R}^j(c^0;\mathfrak{D};y^1,\cdots,y^n)$ and if there exists a constant $\widetilde{C}(\mathfrak{D}) > 0$
           such that $\widetilde{C}(\mathfrak{D}) \leqc \inf_{(y^1,\cdots,y^n)\in \mathfrak{D}} 
           |\mathfrak{F}_c(y^1,\cdots,y^n)|,$ then \\ 
           $1/\mathfrak{F}_c\circ(q_1,\cdots,q_n) \in \mathcal{R}^j(c^0;\mathfrak{D};q_1,\cdots,q_n).$
        \end{remark}
        
        \begin{remark} \label{R:FcIsInIN} 
            Lemma \ref{L:MultiplyFcGcSobolevEstimate} shows that if $\mathfrak{F}_c(y^1,\cdots,y^n) \in
           \mathcal{R}^j(c^0;\mathfrak{D};y^1,\cdots,y^n)$ and
            $\mathfrak{F}_c \circ (\bar{q}_1,\cdots,\bar{q}_n) = 0,$ then 
            $\mathfrak{F}_c \circ (q_1,\cdots,q_n) \in \mathcal{I}^j(c^{k};\mathfrak{D};q_1,\cdots,q_n).$ In particular,
            if $\bar{q}=0,$ then any monomial $q^k$ for $k>0$ is an element of $\mathcal{I}^j(q).$
        \end{remark}

        \begin{remark} \label{R:LargecSobolevConstantRemark}
            Lemma \ref{L:MultiplyFcGcSobolevEstimate} shows in
            particular that for $k \leq 0,$ $\mathcal{I}^j(c^{k};\mathfrak{D};q_1,\cdots,q_n)$ is an ideal in
            $\mathcal{R}^j(\mathfrak{D};q_1,\cdots,q_n).$
        \end{remark}

        \begin{remark} \label{R:BoundedThroughALimit}
            If $k \leq 0$ and there exists a fixed function $\mathfrak{F}_{\infty} \in \mathcal{R}^j(\mathfrak{D};y^1,\cdots,y^n)$
            such that $\mathfrak{F}_c - \mathfrak{F}_{\infty} \in
            \mathcal{R}^j(c^{k};\mathfrak{D};y^1,\cdots,y^n),$
            then it follows that $|\mathfrak{F}_c|_{j,\mathfrak{D}} \leqc
            |\mathfrak{F}_{\infty}|_{j,\mathfrak{D}} + 1,$ so that the family of functions
            $\mathfrak{F}_c$ is uniformly bounded in the norm
            $|\cdot|_{j,\mathfrak{D}}$ for all large $c.$ A
            similar remark using the $\|\cdot\|_{H^j}$ norm applies if $\mathfrak{F}_{\infty} \in 
            \mathcal{I}^j(\mathfrak{D};q_1,\cdots,q_n)$
            and $\mathfrak{F}_c - \mathfrak{F}_{\infty} \in \mathcal{I}^j(c^{k};\mathfrak{D};q_1,\cdots,q_n).$ We often make use of these
            observations in this article without explicitly mentioning it.
        \end{remark}

        \begin{lemma} \label{L:TimeDifferentiatedSobolevEstimate}
            Suppose that $j \geq 3,$ $k_1 + k_2 = k_0,$ and that \\
            $\mathfrak{F}_c \in \mathcal{R}^j(c^{k_0};\mathfrak{D}_1;q_1,\cdots,q_n).$ Assume further that
            for $1 \leq i \leq n,$ we have that  $q_i\in \bigcap_{k=0}^{k=1} C^k([0,T],H_{\bar{q}_i}^{j-k})$ and that for 
            all large $c,$ that \\
            $c^{k_2} \big(\partial_t q_1,\cdots,\partial_t q_n\big) ([0,T] \times \mathbb{R}^3)\subset \mathfrak{D}_2.$
            Then on $[0,T],$ we have that 
            \begin{align}
            	\partial_t \big(\mathfrak{F}_c \big) \in
            	\mathcal{I}^{j-1}(c^{k_1};\mathfrak{D}_1 \times \mathfrak{D}_2;q_1,\cdots,q_n,c^{k_2}\partial_t 	
            	q_1,\cdots,c^{k_2}\partial_t q_n).
         		\end{align}
        \end{lemma}
            \begin{proof}
             Lemma \ref{L:TimeDifferentiatedSobolevEstimate} follows from the chain rule,
             Lemma \ref{L:MultiplyFcGcSobolevEstimate}, and Remark \ref{R:FcIsInIN}.
             We emphasize that the constant term associated to $c^{k_2}\partial_t q_i$ is $0,$ so that 
             on the right-hand side of the definition \eqref{E:IDef} of $\mathcal{I}^{j-1}(\cdots),$ we are measuring $c^{k_2}\partial_t 
             q_i$ in the $H^{j-1}$ norm. 
            \end{proof}

        \begin{corollary} \label{C:TimeDifferentiatedSobolevEstimate}
            Let $\partial_a$ be a first-order spatial coordinate derivative operator. Suppose that $j \geq 3,$ $k_1 + k_2 = k_0,$ and that
            $\mathfrak{F}_c \in \mathcal{R}^j(c^{k_0};\mathfrak{D}_1;q_1,\cdots,q_n).$
            Assume that for all large $c,$ we have that $c^{k_2}\big(\partial_a q_1,\cdots,\partial_a q_n\big) ([0,T] \times 
            \mathbb{R}^3)\subset \mathfrak{D}_2.$ Then on $[0,T],$ we have that 
            \begin{align}
            	\partial_a \big(\mathfrak{F}_c \big)
            	\in \mathcal{I}^{j-1}(c^{k_1};\mathfrak{D}_1 \times \mathfrak{D}_2;q_1,\cdots,q_n,c^{k_2}\partial_a 
            	q_1,\cdots,c^{k_2}\partial_a q_n).
           	\end{align}
        \end{corollary}
            \begin{proof}
                The proof of Corollary \ref{C:TimeDifferentiatedSobolevEstimate} is virtually identical to the proof of
                Lemma \ref{L:TimeDifferentiatedSobolevEstimate}. \\
            \end{proof}

        \subsection{Applications to the EN$_{\kappa}^c$ system} \label{SS:ApplicationtoENkappac} 
				We will now apply these lemmas to the EN$_{\kappa}^c$ system. Let us first make a few remarks about our use of the  
        norms $\| \cdot \|_{H^j,\bar{q}_i}$ that appear on the right-hand side 
        of \eqref{E:IDef} and the constant term $\bar{q}_i$ associated to $q_i.$ For the remainder of this 
        article, it is to be understood that the constant term associated to $c^{k} \Vbold$ is $c^{k} \Vb_c,$ that the constant term
        associated to $c^{k} \scrV$ is $c^{k} \scrVb_c,$ and the constant term associated to both $D \Vbold$ and $D \scrV$ is 
        $\mathbf{0},$ where $\scrVb_c$ and $\Vb_c$ are defined in \eqref{E:scrVcConstant} and \eqref{E:VcConstant} 
        respectively. In other words, when estimating $c^{k} \Vbold$ using a $j-$th order Sobolev norm, it is understood that we are using
        the norm $\|\cdot\|_{H_{c^{k}\Vb_c}^j},$ and similarly for the other state-space arrays. The relationship between the arrays 
        $\Vbold$ and $\scrV$ is always understood to be the one implied by \eqref{E:Vboldarray} and 
        \eqref{E:scrVboldarray}. We furthermore emphasize that $\Vbold$ (or $\scrV$) will represent a solution array to the 
        EN$_{\kappa}^c$ system, and therefore will implicitly depend on $c$ through the $c-$dependent initial data $\VID_c$ (see 
        \eqref{E:ENkappacVID}) and through the $c$ dependence 
        of the EN$_{\kappa}^c$ system itself. The fact that the constant arrays $\scrVb_c$ and $\Vb_c$ depend on the parameter $c$ does 
        not pose any difficulty. For as we shall see, $\scrVb_c$ is contained in the fixed compact set $\mathfrak{K}$ for all large $c,$ 
        and $\Vb_c$ is contained in the fixed compact set $K$ for all large $c,$ where the sets $\mathfrak{K}$ and 
        $K$ were introduced at the beginning of Section \ref{S:cDependence}. Therefore, the $L^{\infty}$ estimates 
        of the constants $\scrVb_c$ and $\Vb_c$ that we will need can be made independently of all large $c.$  
        
        In addition to the above remarks, we add that we will have available a-priori estimates
        that guarantee that $\Vbold \in \bigcap_{k=0}^{k=2}C^k([0,T],H_{\Vb_c}^{N-k})$
        for a fixed integer\footnote{The relevance of $N \geq 4$ is described in Section \ref{S:IVPc}.} \\
        $N \geq 4$ on our time interval $[0,T]$ of interest, which are hypotheses that are relevant for Lemma 
        \ref{L:TimeDifferentiatedSobolevEstimate} and Corollary 
        \ref{C:TimeDifferentiatedSobolevEstimate}. Our a-priori estimates will also ensure that
        all of the relevant quantities are contained in an appropriate fixed compact convex set, so that  
        the ``hypotheses on the $q_i$'' described in Definition \ref{D:qRingDef} will always be satisfied. Consequently, we will often 
        omit the dependence of the running constants $C(\cdots)$ on such sets. The relevant a-priori 
        estimates (``Induction Hypotheses'') are described in detail in Section \ref{SS:InductionHypotheses}. 
 
 				Let us now provide a clarifying example and also elaborate upon the idea that it is sometimes useful to shift the 
 				point of view as to what are the arguments of a family $\mathfrak{F}_c(\cdots).$ For example, consider the expression 
 				$\mathfrak{F}_c \eqdef c^{-2}\partial_t \Phi,$ where $\Phi$ is a solution variable in the EN$_{\kappa}^c$ system
        depending on $c$ through the initial data $\VID_c$ and through the $c-$dependence of the
        system itself. If it is known that $c^{-1}\|\partial_t\Phi\|_{H^3}$ is uniformly bounded by $L$ for all large $c,$ 
        then we have that $\mathfrak{F}_c \in \mathcal{I}^3(c^{-1};c^{-1}\partial_t \Phi)$ since $c^{-1}\|c^{-1}\partial_t\Phi\|_{H^3} 
        \leqc c^{-1}L.$ If it also turns out that $\|\partial_t \Phi\|_{H^3}$ is uniformly bounded for all large $c,$ then have that 
        $\mathfrak{F}_c \in \mathcal{I}^3(c^{-2};\partial_t \Phi).$ If both estimates are true, then we indicate this by writing 
        $\mathfrak{F}_c \in 
        \mathcal{I}^3(c^{-1};c^{-1}\partial_t \Phi) \cap \mathcal{I}^3(c^{-2};\partial_t \Phi)$  or $\mathfrak{F}_c = 
        \mathscr{O}^3(c^{-1};c^{-1}\partial_t \Phi) \cap \mathscr{O}^3(c^{-2};\partial_t \Phi).$ These kinds of estimates will enter into
        our continuous induction argument in Section \ref{SS:UniforminTimeExistence}, in which we will first prove a bound for 
        $c^{-1}\partial_t \Phi,$ and then use it to obtain a bound for $\partial_t \Phi;$ see \eqref{E:UniformtimeApriori5} and 
        \eqref{E:UniformtimeApriori7}.
 
				\begin{remark} \label{R:NotAlwaysOptimal}
            For simplicity, we are not always optimal in our estimates. 
        \end{remark}

        The following four lemmas, which provide an analysis of the $c-$dependence of the
        terms appearing in the EN$_{\kappa}^c$ system, will be used heavily in 
        Section \ref{SS:TechnicalLemmas}, which contains most of our technical estimates.
        Before providing the lemmas, we first restate our hypotheses on the equation of state using our new notation \\
        
        \noindent \textbf{Hypotheses on the $c-$Dependence of the Equation of State}
            \begin{align} \label{E:EOScHypothesis1}
                \mathfrak{R}_c(\Ent,p), \ \mathfrak{R}_{\infty}(\Ent,p) & \in \mathcal{R}^{N+1}(\mathfrak{C};\Ent,p) \\
                \mathfrak{R}_c(\Ent,p) - \mathfrak{R}_{\infty}(\Ent,p) &\in
                \mathcal{R}^{N+1}(c^{-2};\mathfrak{C};\Ent,p).  \label{E:EOScHypothesis2}
            \end{align}

        \noindent Recall that the set $\mathfrak{C}$ was introduced at the beginning of Section \ref{S:cDependence} and
        is described in detail in Section \ref{SS:AdmissibleStateSpacec}. We also assume that $\mathfrak{R}_{\infty}(\Ent,p)$ and  
        $\mathfrak{S}_{\infty}^2(\Ent,p)$ are ``physical'' as defined in Section \ref{SS:ENkappacDerivation}, and in particular
        that whenever $\Ent,p >0,$ we have that \\
        $0 < \mathfrak{R}_{\infty}(\Ent,p)$ and $0 < \mathfrak{S}_{\infty}^2(\Ent,p).$
     		Additionally, we note the following simple consequence of 
        \eqref{E:partialRpartialpandStoNegativeTwoRelationshipc}, \eqref{E:EOScHypothesis1}, and \eqref{E:EOScHypothesis2}:
        \begin{align} \label{E:SpeedofSoundcDependence}
            \mathfrak{S}_c^2(\Ent,p) - \mathfrak{S}_{\infty}^2(\Ent,p) &\in
                \mathcal{R}^N(c^{-2};\mathfrak{C};\Ent,p).
        \end{align}

				\begin{remark}
					At the end of this section, we provide an example of a well-known family of equations of state,
        	namely the polytropic equations of state, that satisfy the above hypotheses.
				\end{remark}

				Hypothesis \eqref{E:EOScHypothesis1} ensures that the terms appearing in the EN$_{\kappa}^c$ and EP$_{\kappa}$
				systems are sufficiently differentiable functions of $\scrV,$ thus enabling us to apply the Sobolev-Moser type inequalities
				appearing in Appendix \ref{A:SobolevMoser}. It is strong enough to imply Theorem \ref{T:LocalExistencec} and 
				Theorem \ref{T:EPkappaLocalExistence}. Hypothesis \eqref{E:EOScHypothesis2} is used in our proof of Theorem
				\ref{T:UniformLocalExistenceENkappac} and Theorem \ref{T:NewtonianLimit}. Although a weakened version 
				of Hypothesis \eqref{E:EOScHypothesis2} is sufficient to prove a convergence theorem, we do not pursue this 
				matter here since we are not striving for optimal results.

       \begin{lemma}   \label{L:OrderctoNegativeTwoEstimates}
        		Let $\gamma_c, \Rc, \Rinfinity, \Qc, \Qinfinity, \Wbold,$ and $\scrW$ be the quantities defined in
        		\eqref{E:gammacDefII}, \eqref{E:RcDefII}, \eqref{E:RinfinityDef}, \eqref{E:QcDefII}, \eqref{E:QinfinityRelationship}, 
        		\eqref{E:Wboldarray}, and \eqref{E:scrWboldarray} respectively. Then
        		for $m=0,1,2$ and $\nu = t,1,2,3$ we have the following
        		estimates for all large $c,$ including $c=\infty:$  
            \begin{align}
                (\gamma_c)^2 - 1 &\in \mathcal{R}^{N+1}(c^{-2};\mathbf{v}) \label{E:gammacOrderctoNegativeTwoEstimate}\\
                e^{ \stackrel{+}{_-} 4\Phi/c^2} - 1 &\in \mathcal{R}^{N+1}(c^{m-2};c^{-m}\Phi)
                     \label{E:expPhiOvercSquaredOrderctoNegativeTwoEstimate}\\
                \Rc - \Rinfinity = e^{4\Phi/c^2}\mathfrak{R}_c(\Ent,p) - \mathfrak{R}_{\infty}(\Ent,p)
                    &\in \mathcal{R}^{N+1}(c^{m-2};\Ent,p,c^{-m}\Phi) \label{E:RcOrderctoNegativeTwoEstimate} \\
                \Qc - \Qinfinity = \mathfrak{Q}_c(\Ent,p,\Phi) - \mathfrak{Q}_{\infty}(\Ent,p) &\in 
                \mathcal{R}^N(c^{m-2};\Ent,p,c^{-m}\Phi)
                    \label{E:QcOrderctoNegativeTwoEstimate} \\
                \scrW - \Wbold &\in \mathcal{R}^N(c^{m-2};P,c^{-m}\Phi) \label{E:WscrWOrderctoNegativeTwoEstimate} \\
                \scrW &\in \mathcal{R}^N(\Wbold,c^{-m}\Phi) \label{E:ScrWInTermsofWEstimate} \\
                \partial_{\nu} \scrW - \partial_{\nu} \Wbold & \in \mathcal{I}^{N-1}(c^{m-2};P,\partial_{\nu} P, c^{-m}\Phi, 
                	c^{-m}\partial_{\nu} \Phi) \label{E:SpatialDerivativesWscrWOrderctoNegativeTwoEstimate} \\
                \partial_{\nu} \scrW & \in \mathcal{I}^{N-1}(\Wbold,\partial\Wbold, c^{-m}\Phi, 
                	c^{-m}\partial_{\nu} \Phi). \label{E:DerivativesWscrWOrderctoZeroEstimate}
       		\end{align}
        \end{lemma}
        \begin{proof}
            \eqref{E:gammacOrderctoNegativeTwoEstimate}, and
            \eqref{E:expPhiOvercSquaredOrderctoNegativeTwoEstimate} are easy
            Taylor estimates. \eqref{E:RcOrderctoNegativeTwoEstimate} follows
            from Lemma \ref{L:MultiplyFcGcSobolevEstimate}, \eqref{E:EOScHypothesis1},
            \eqref{E:EOScHypothesis2}, and \eqref{E:expPhiOvercSquaredOrderctoNegativeTwoEstimate}.
            \eqref{E:QcOrderctoNegativeTwoEstimate} then follows from \eqref{E:partialRpartialpandStoNegativeTwoRelationshipc}, 
            \eqref{E:QcFunctionDef}, \eqref{E:QinfinityRelationship}, Lemma \ref{L:MultiplyFcGcSobolevEstimate}, 
            \eqref{E:SpeedofSoundcDependence}, and
            \eqref{E:RcOrderctoNegativeTwoEstimate}. Since
            $P - p =(1 - e^{-4\Phi/c^2})P,$ \eqref{E:WscrWOrderctoNegativeTwoEstimate} follows
            from \eqref{E:expPhiOvercSquaredOrderctoNegativeTwoEstimate},
            Lemma \ref{L:MultiplyFcGcSobolevEstimate}, and that
            the fact that $\scrW$ and $\Wbold$ differ only in that
            the second component of $\scrW$ is $p,$ while the
            second component of $\Wbold$ is $P.$ \eqref{E:ScrWInTermsofWEstimate} 
            is a simple consequence of \eqref{E:WscrWOrderctoNegativeTwoEstimate}. 
            \eqref{E:SpatialDerivativesWscrWOrderctoNegativeTwoEstimate}
            follows from \eqref{E:WscrWOrderctoNegativeTwoEstimate}, Lemma \ref{L:TimeDifferentiatedSobolevEstimate}, and Corollary
            \ref{C:TimeDifferentiatedSobolevEstimate}. \eqref{E:DerivativesWscrWOrderctoZeroEstimate} then follows
            easily from \eqref{E:SpatialDerivativesWscrWOrderctoNegativeTwoEstimate}.
        \end{proof}
        
        The next lemma connects the $c-$asymptotic behavior of an expression written in terms of
        the state-space array $\Wbold$ to the $c-$asymptotic behavior of the same expression written in terms of 
        the state-space array $\scrW.$
        \begin{lemma} \label{L:ScrWInTermsofWEstimate}
        	If $0 \leq j \leq N$ and $\mathfrak{F}_c \in \mathcal{R}^j(c^{k};\scrW),$ then for
        	$m=0,1,2,$ we have that
        	\begin{align}
        	\mathfrak{F}_c \in \mathcal{R}^j(c^{k};\Wbold,c^{-m}\Phi).
        	\end{align}
        \end{lemma}
        
        \begin{proof}
        	Lemma \ref{L:ScrWInTermsofWEstimate} follows easily from expressing $\scrW$ in terms of $\Wbold$ and $c^{-m} \Phi$ via
        	\eqref{E:ScrWInTermsofWEstimate} and applying the chain rule. 
        \end{proof}
        
        \begin{lemma} \label{L:MatricesLargecNormEstimates}
        Let $\Ac^{\nu}(\scrW,\Phi), \ \nu=0,1,2,3,$ denote the matrix-valued functions of $\scrW$ and $\Phi$ introduced in Section
        \ref{SS:EOVc}. Let the $c-$dependent relationship between $\scrW$ and $\Wbold,\Phi$ be defined by
        \eqref{E:Wboldarray} and \eqref{E:scrWboldarray}. Then for all large $c$ including $c=\infty,$ and for $m=0,1,2,$
       	we have that
            \begin{align}
               	\Ainfinity^{\nu}(\scrW), \ \big(\Ainfinity^{0}(\scrW)\big)^{-1} 
                	& \in \mathcal{R}^{N}(\scrW) \cap \mathcal{R}^{N}(\Wbold,c^{-m}\Phi) 
                	\label{E:MatricescIsInfinityNormEstimates} \\
               	\Ac^{\nu}(\scrW,\Phi), \ \big(\Ac^{0}(\scrW,\Phi)\big)^{-1} 
                	& \in \mathcal{R}^{N}(\scrW,c^{-m}\Phi) \cap \mathcal{R}^{N}(\Wbold,c^{-m}\Phi) 
                	\label{E:MatricescIsFiniteNormEstimates} \\
                \Ac^{\nu}(\scrW,\Phi) - \Ainfinity^{\nu}(\scrW) & \in \mathcal{R}^{N}(c^{m-2};\scrW,c^{-m}\Phi)
                	\cap \mathcal{R}^{N}(c^{m-2};\Wbold,c^{-m}\Phi) \label{E:MatricesLargecNormEstimates}\\
                \big(\Ac^{0}(\scrW,\Phi)\big)^{-1} - \big(\Ainfinity^{0}(\scrW)\big)^{-1} & \in
                    \mathcal{R}^{N}(c^{m-2};\scrW,c^{-m}\Phi) \cap \mathcal{R}^{N}(c^{m-2};\Wbold,c^{-m}\Phi). 
                    \label{E:InverseMatricesLargecNormEstimates}
            \end{align}
        \end{lemma}
            \begin{proof}
                \eqref{E:MatricescIsInfinityNormEstimates} - \eqref{E:InverseMatricesLargecNormEstimates}
                follow from \eqref{E:A0c} - \eqref{E:AjInfinity}, Remark \ref{R:1OverFisintheRing}, Lemma 
                \ref{L:MultiplyFcGcSobolevEstimate}, Lemma \ref{L:OrderctoNegativeTwoEstimates}, Lemma \ref{L:ScrWInTermsofWEstimate}, the 
                determinant-adjoint formula for the inverse of a matrix, and the hypotheses \eqref{E:EOScHypothesis1},
                \eqref{E:EOScHypothesis2} on the equation of state.
        		\end{proof}

        \begin{lemma} \label{L:InhomogeneousTermsLargecSobolevEstimates} 
            Let $\mathfrak{B}_{\infty}(\scrW,\partial\Phi) \eqdef \Big(0,0,-\mathfrak{R}_{\infty}(\Ent,p)\partial_1
            \Phi, -\mathfrak{R}_{\infty}(\Ent,p)\partial_2 \Phi,-\mathfrak{R}_{\infty}(\Ent,p)\partial_3 \Phi)\Big)$ denote the
            right-hand side of the EP$_{\kappa}$ equations \eqref{E:EPkappa1}, \eqref{E:EPkappa2}, \eqref{E:EPkappa3}, and let \\
            $\mathfrak{B}_{c}(\scrW,\Phi,D\Phi)$ denote the right-hand side of the EN$_{\kappa}^c$ equations \eqref{E:ENkappac1} -
            \eqref{E:ENkappac3}. Let the $c-$dependent relationship between $\scrW$ and $\Wbold,\Phi$ be defined by
        		\eqref{E:Wboldarray} and \eqref{E:scrWboldarray}. Then for all large $c$ including $c= \infty,$ 
        		and for $m=0,1,2$ and $n=0,1,$ we have that
            \begin{align} 
            	\mathfrak{B}_{\infty}(\scrW,\partial\Phi) &\in 
            	\mathcal{I}^{N}(c^{n};\scrW,c^{-n}\partial\Phi) \cap \mathcal{I}^{N}(c^{n};\Wbold,c^{-m}\Phi,c^{-n}\partial\Phi)
            		\label{E:BInfinityIsInIN} \\
            	\mathfrak{B}_c(\scrW,\Phi,D\Phi) &\in \mathcal{I}^{N}(\scrW,c^{-m}\Phi,\partial\Phi,c^{-m}\partial_t\Phi) \cap 
            		\mathcal{I}^{N}(\Wbold,c^{-m}\Phi,\partial\Phi,c^{-m}\partial_t\Phi) \label{E:BcIsInIN} \\
            	& and \notag \\
        			\mathfrak{B}_c(\scrW,\Phi,D\Phi) &= \mathfrak{B}_{\infty}(\scrW,\partial\Phi) + 
                \mathcal{O}^{N}(c^{m-2};\Wbold,c^{-m}\Phi,c^{-m} D \Phi). \label{E:BcBinfinityDifferenceEstimate}
            \end{align}
          
          \end{lemma}
            \begin{proof}
                \eqref{E:BInfinityIsInIN} - \eqref{E:BcBinfinityDifferenceEstimate} all follow from combining the facts  
                \\ $\mathfrak{B}_{\infty}(\scrWb_{\infty}, \mathbf{0})=0$ and
                $\mathfrak{B}_c(\scrWb_c,\Phibar_c, \mathbf{0})=0$ with Remark \ref{R:FcIsInIN}, 
                Lemma \ref{L:MultiplyFcGcSobolevEstimate}, Lemma \ref{L:OrderctoNegativeTwoEstimates}, and Lemma 
                \ref{L:ScrWInTermsofWEstimate}. 
         		\end{proof}
         		
         		\begin{remark} \label{R:WorstTerm}
         			The fact that
         			$\mathfrak{B}_{\infty}(\scrW,\partial\Phi) \in \mathcal{I}^{N}(c^{1};\Wbold,c^{-m}\Phi,c^{-1}\partial\Phi)$
         			plays a distinguished role in the proof of Lemma \ref{L:UniformtimeApriori2}; $\mathfrak{B}_{\infty}(\scrW,\partial\Phi)$ 
         			will be one of the factors in the ``worst 
         			error term'' because it can grow like $c^1$ if we only have control over the size of
         			$c^{-1}\partial\Phi.$
         		\end{remark}
        
        \begin{remark}
        	Many of the above lemmas are valid for other values of $m$ and $n;$ we stated the lemmas
        	for the values of $m$ and $n$ that we plan to use later.
        \end{remark}

        \begin{example}
            As an enlightening example, we discuss the non-relativistic limit of \emph{polytropic}
            equations of state, that is, equations of state of the form \\ $\rho = m_0c^2n + \frac{A_c({\Ent})}{\gamma 
            -1}n^{\gamma},$ where $m_0$ denotes the rest mass of a fluid element, $n$ denotes the proper number density, and 
            $\gamma > 1.$ Let us assume that \\
            $A_c, A_{\infty} 
            \in \mathfrak{R}^{N+1}(\Pi_{1}(\mathfrak{C});\Ent),$ 
            that $A_{\infty}>0$ on $\Pi_{1}(\mathfrak{C}),$ and that \\
            $A_c - A_{\infty} \in 
            \mathcal{R}^{N+1}(c^{-2};\Pi_{1}(\mathfrak{C});\Ent),$ where $\Pi_{1}(\mathfrak{C})$ is the projection of the set
            $\mathfrak{C}$ introduced at the beginning of Section \ref{S:cDependence} onto the first axis. Some 	
            omitted calculations show that Hypotheses \ref{E:EOScHypothesis1} and \ref{E:EOScHypothesis2} then hold, and that
                \begin{align}
                    R_c &=e^{4 \Phi/c^2}\mathfrak{R}_c(\Ent,p)=\frac{m_0P^{1/\gamma}e^{4\Phi/c^2(1-1/\gamma)}}{A^{1/\gamma}_c(\Ent)}
                    + \frac{P}{c^2(\gamma-1)} \\
                    Q_c &=\mathfrak{Q}_c(\Ent,p,\Phi)= \gamma P \\
                    R_{\infty}&=\mathfrak{R}_{\infty}(\Ent,p)=\frac{m_0p^{1/\gamma}}{A^{1/\gamma}_{\infty}(\Ent)} 
                    	\label{E:NewtonianPolytropic}\\
                    Q_{\infty} &=\mathfrak{Q}_{\infty}(\Ent,p)= \gamma p.
                \end{align}
        \end{example}
        In the isentropic case $\Ent(t,\sforspace) \equiv \Entbar,$ \eqref{E:NewtonianPolytropic} can be rewritten in the familiar 
        form $p = C \cdot (R_{\infty})^{\gamma},$ where $C$ is a constant.

 \section{Energy Currents} \label{S:EnergyCurrentsc}
    In this section we provide energy currents and discuss two key properties: $\mathbf{i})$ for a fixed $c,$ they are positive 
    definite in the variations $\Wdot$ when contracted against certain covectors, and $\mathbf{ii})$ their divergence is lower order in 
    the variations. In Section \ref{S:JdotUniformPositiveDefinitec}, we will see
    that the positivity property is uniform for all large $c.$ A general framework 
    for the construction of energy currents for hyperbolic systems derivable from a Lagrangian
    is developed in \cite{dC2000}. The role of energy currents is to replace the energy principle
    available for symmetric hyperbolic systems by providing integral identities, or more generally,
    integral inequalities, that enable one to control Sobolev norms of solutions\footnote{As we shall see, the energy currents 
   	$\Jscrdotc$ do not control the variations $\Phidot$ or $D
    \Phidot;$ these terms are controlled through a separate argument based on the lemmas and propositions of Appendix 
    \ref{A:Klein-GordonEstimates}.} to the EOV$_{\kappa}^c.$ This technique will be used in our proofs of Lemma 
    \ref{L:UniformtimeApriori1} and Theorem \ref{T:NewtonianLimit}.

    \subsection{The definition of an energy current} \label{SS:Re-scaledEnergyCurrent} 
    Given a variation $\Wdot: \mathcal{M} \rightarrow \mathbb{R}^{5}$
    and a BGS\footnote{Recall that we also refer to $\scrWtilde$ as the BGS when $c=\infty.$} $\Vboldt: \mathcal{M} \rightarrow 
    \mathbb{R}^{10}$ as defined in Section \ref{SS:EOVc}, we define the energy current to be the vectorfield $\Jscrdotc$ with
    components $\Jscrdotc^0,$ $\Jscrdotc^j,$ $j=1,2,3,$ in the global rectangular coordinate system given by
    \begin{align}  \label{E:RescaledEnergyCurrent}
            \Jscrdotc^0 &\eqdef \dot{\Ent}^2  + \frac{\dot{P}^2}{\wtQc} +
            2c^{-2}(\wtg)^2 (\wtv_k \dot{v}^k) \dot{P} \\
            & \hspace{1in} + (\wtg)^2 \left[\wtRc + c^{-2}\wtP \right] \cdot \left[\dot{v}_k \dot{v}^k +
            c^{-2}(\wtg)^2(\wtv_k \dot{v}^k)^2 \right] \notag \\ 
            \Jscrdotc^j &\eqdef \wtv^j \dot{\Ent}^2  + \frac{\wtv^j}{\wtQc} \dot{P}^2 + 2 \left[\dot{v}^j
                + c^{-2}(\wtg)^2 \wtv^j \wtv_k \dot{v}^k \right] \cdot \dot{P} \notag \\
            & \hspace{1in} + (\wtg)^2\wtv^j \left[\wtRc + c^{-2}\wtP\right] \cdot \left[\dot{v}_k \dot{v}^k + c^{-2}(\wtg)^2(\wtv_k 
            \dot{v}^k)^2 \right].   \notag 
    \end{align}
    In the case $c=\infty,$ we define for $j=1,2,3:$
    \begin{align} \label{E:RescaledEnergyCurrentInfinity}
        \Jscrdotinfinity^0 &\eqdef \dot{\Ent}^2  + \frac{\dot{p}^2}{\widetilde{Q}_{\infty}}
            + \widetilde{R}_{\infty} \dot{v}_k \dot{v}^k \\
        \Jscrdotinfinity^j &\eqdef \wtv^j \dot{\Ent}^2  + \frac{\wtv^j}{\widetilde{Q}_{\infty}} \dot{p}^2 + 2 \dot{v}^j \dot{p}
            + \widetilde{R}_{\infty} \wtv^j \dot{v}_k \dot{v}^k. \notag
        \end{align}
        We note that formally, $\lim_{c \to \infty} \Jscrdotc = \Jscrdotinfinity,$ a fact that will be rigorously justified in Section 
        \ref{S:JdotUniformPositiveDefinitec}.

        The energy current \eqref{E:RescaledEnergyCurrent} is very closely related to the energy current $\dot{J}$
        introduced in \cite{jS2008a}, where the following changes have been made. First, we have dropped the terms from $\dot{J}$ 
        corresponding to the variations of the potential $\Phidot$ and its derivatives, for we will bound these terms in a Sobolev norm 
        using a separate argument. Second, the expression for $\Jscrdotc$ is constructed using the velocity
        state-space variable $\mathbf{v}$ \eqref{E:vdef} and variations $\dot{\mathbf{v}},$
        as opposed to the variables $U^j \eqdef e^{\phi} u^j$ and variations $\dot{U}^j$ that appear in the expression
        for $\dot{J}.$ Finally, we emphasize that the formula for $\Jscrdotc^{\nu}$
        applies in a rectangular coordinate system with $x^0 = t,$ whereas in the formula
        for $\dot{J}^{\nu}$ provided in \cite{jS2008a}, the rectangular coordinate system is such that $x^0 = ct,$ even though $c$ was 
        set equal to unity in \cite{jS2008a}.
        
       \begin{remark}
        	A similar current was used by Christodoulou in \cite{dC2007} to analyze the motion of a relativistic fluid 
        	evolving in Minkowski space.
       \end{remark}

      	\subsection{The positive definiteness of $\xi_{\mu} \Jscrdotc^{\mu}$ for $\xi \in \mathcal{I}_x^{s*+}$}
        \label{SS:JscrdotcPositiveDefinite} 
				 As discussed in detail in \cite{jS2008a}, for $\xi$ belonging to a certain subset of the cotangent space at $x,$ which
				we denote by $T_x^* \mathcal{M},$ the quadratic 
        form\footnote{We write ``$\xi_{\mu} \Jscrdotc^{\mu}(\Wdot,\Wdot)$'' to emphasize the point of view that $\xi_{\mu} 
        \Jscrdotc^{\mu}$ is a quadratic form in $\Wdot.$} $\xi_{\mu} \Jscrdotc^{\mu}(\Wdot,\Wdot)$ is positive
        definite in $\Wdot$ if $\wtP >0.$ To elaborate upon this, we follow
        Christodoulou \cite{dC2007} and introduce the \emph{reciprocal acoustical metric} $\widetilde{h}^{-1},$  
				a quadratic form on $T_x^* \mathcal{M}$ 
				with components that read (for $j,k = 1,2,3$)
        \begin{align}
            (\widetilde{h}^{-1})^{00} &\eqdef -c^{-2} - (\widetilde{\gamma}_c)^2\big[\mathfrak{S}_c^{-2}(\widetilde{\Ent},\widetilde{p}) - 	
            	c^{-2}\big]\\
            (\widetilde{h}^{-1})^{0j} &\eqdef (\widetilde{h}^{-1})^{j0} = - 
            	(\widetilde{\gamma}_c)^2\big[\mathfrak{S}_c^{-2}(\widetilde{\Ent},\widetilde{p}) 							- c^{-2}\big]\widetilde{v}^j\\
            (\widetilde{h}^{-1})^{jk} &\eqdef \delta^{jk} - (\widetilde{\gamma}_c)^2\big[\mathfrak{S}_c^{-2}(\widetilde{\Ent},\widetilde{p}) - 								c^{-2}\big]\widetilde{v}^j \widetilde{v}^k
        \end{align}
        in the global rectangular coordinate system. Recall that
        the function $\mathfrak{S}_c$ is defined in \eqref{E:SigmaSquaredc}. 
        
        Recall that for a hyperbolic system of PDEs, the characteristic subset\footnote{\cite{jS2008a} contains
        a detailed discussion of the notion of the characteristic subset of $T^*_x 
        \mathcal{M}$ in the context of the EN$_{\kappa}^{c=1}$ system.}of $T^*_x 
        \mathcal{M}$ is the union of several sheets. If we restrict our attention to 
        the truncated\footnote{By ``truncated EOV$_{\kappa}^c$'' we mean the system that results upon deleting the variable $\Phidot$
        and the equation \eqref{E:EOVcKlein-Gordon} that it satisfies.} EOV$_{\kappa}^c$ \eqref{E:EOVc1} - \eqref{E:EOVc3}, then
        omitted calculations imply that the inner sheet is the sound cone at $x,$ which can be described in coordinates as 
        $\{ \zeta \in T^*_x \mathcal{M} \ | \ (\widetilde{h}^{-1})^{\mu \nu} \zeta_{\mu} \zeta_{\nu} = 0 \}.$
        The interior of the positive component of the sound cone, which we 
        denote by $\mathcal{I}_x^{s*+},$ can be described in coordinates as
        \begin{align}
        \mathcal{I}_x^{s*+} \eqdef
                \{ \zeta \in T^*_x \mathcal{M} \ | \ (\widetilde{h}^{-1})^{\mu \nu} \zeta_{\mu} \zeta_{\nu} < 0 \ \mbox{and} \ \zeta_0 > 
                0 \}. \label{E:SoundConec}
      	\end{align} 
      	\noindent We remark that the characteristic subsets of the $T^*_x \mathcal{M}$ in the complete EOV$_{\kappa}^c$ system 
      	\eqref{E:EOVc1} - \eqref{E:EOVcKlein-Gordon} feature an additional sheet: the light cone (also known as the ``gravitational 
      	null cone''), \emph{which is contained inside the sound cone}.\footnote{As discussed in \cite{jS2008a}, one can also 
        define the sound cone and light cone subsets of the tangent space at $x,$ which we denote by $T_x \mathcal{M},$ by introducing 
        the notion of the dual to a sheet of the 
        characteristic subset of $T_x^* \mathcal{M}.$ The duality reverses the aforementioned containment so that in $T_x \mathcal{M},$ 
        the sound cone is contained inside of the light cone. 
        This is perhaps the more familiar picture, for it corresponds to our intuitive notion of sound traveling
        more slowly than light.}
        
        It follows from the general construction of energy currents as presented in \cite{dC2000} that
        $\xi_{\mu} \Jscrdotc^{\mu}(\Wdot,\Wdot)$ is positive definite whenever $\wtP>0$ and $\xi$ belongs to the interior of the 
        positive component of the sound cone in $T_x^* \mathcal{M}:$
        \begin{align}           \label{E:PositiveDefinitec}
            \xi_{\mu} \Jscrdotc^{\mu}(\Wdot,\Wdot) &> 0
            \ \mbox{if} \ |\Wdot|>0, \ \wtP>0, \ \mbox{and} \ \xi \in \mathcal{I}_x^{s*+}.
        \end{align}
        \noindent The inequality \eqref{E:PositiveDefinitec} allows us to use the quadratic form $\xi_{\mu} \Jscrdotc^{\mu}(\Wdot,\Wdot)$ 
        to estimate the $L^2$ norms of the variations $\Wdot,$ provided that we estimate the BGS $\Vboldt.$ We will 
        discuss this issue further in Section \ref{S:JdotUniformPositiveDefinitec}.
				
				In contrast, the energy current $\dot{J}$ from \cite{jS2008a} has the property that $\xi_{\mu} 
        \dot{J}^{\mu}$ is a positive definite quadratic form in $\Vdot$ only for $\xi$ belonging to the interior of the positive 
        component of the light cone in $T_x^* \mathcal{M};$  $\xi_{\mu} \Jscrdotc^{\mu}(\Wdot,\Wdot)$ is positive definite for a larger 
        set of $\xi$ than is $\xi_{\mu} \dot{J}^{\mu}(\Vdot,\Vdot)$ because $\Jscrdotc$ does not contain terms involving the variations 
        of the potential $\Phidot,$ and therefore does not control the propagation of gravitational waves.
        
       \begin{remark}
            Because $\lim_{c \to \infty} \mathfrak{S}_c^{-2}(\widetilde{\Ent},\widetilde{p}) = 
            \mathfrak{S}_{\infty}^{-2}(\widetilde{\Ent},\widetilde{p}) > 0,$ it follows that for all large $c,$ the covector 
            with coordinates $(1,0,0,0)$ is an element of $\mathcal{I}_x^{s*+}.$ Therefore, $\Jscrdotc^0(\Wdot,\Wdot)$ is positive
            definite for all large $c.$ We also observe that $\Jscrdotinfinity^0(\Wdot,\Wdot),$ which is 
            defined in \eqref{E:RescaledEnergyCurrentInfinity}, is 
            manifestly positive
        definite in the variations if $\widetilde{p} > 0,$ for by our fundamental assumptions on the equation of state,
       	$\widetilde{p} > 0 \implies \widetilde{R}_{\infty} > 0$ and $\widetilde{Q}_{\infty} > 0.$
        \end{remark}

        \subsection{The divergence of the energy current} 
          As described in \cite{jS2008a}, if the variations $\Wdot$ are solutions of the EOV$_{\kappa}^c$ \eqref{E:EOVc1}
         - \eqref{E:EOVc3}, then we can compute $\partial_{\mu} 
         \Big(\Jscrdotc^{\mu}\Big)$ and use the equations \eqref{E:EOVc1}
         - \eqref{E:EOVc3} for substitution to eliminate the terms\footnote{Showing this via a calculation is an arduous task. The 
         lower-order divergence property is a generic feature of energy currents constructed in the manner described in 
         \cite{dC2000}, but we require its explicit form in order to analyze its $c-$ dependence.} containing the derivatives of 
         $\Wdot:$

        \begin{align}
            \begin{split}   \label{E:RescaledEnergyCurrentDivergence}
                 & \partial_{\mu} \Big(\Jscrdotc^{\mu}\Big)  =
                        \Big \lbrace \partial_t\left(\frac{1}{\wtQc}\right) + \partial_j\left(\frac{\wtv^j}{\wtQc}\right) \Big \rbrace 
                        \cdot \dot{P}^2 											\\
                   & +  2 c^{-2} (\widetilde{\gamma}_c)^2 \Big \lbrace \partial_t \wtv_k + \wtv_k \partial_j \wtv^j + 
                   \wtv^j \partial_j \wtv_k + 2 c^{-2} (\widetilde{\gamma}_c)^2 \wtv_k \left(\wtv_j \partial_t \wtv^j + \wtv^j \wtv_a 
                   \partial_j \wtv^a \right) \Big \rbrace \cdot \dot{P} \dot{v}^k  \\
                &  + \Big\lbrace \partial_t \Big[(\widetilde{\gamma}_c)^2\big(\wtRc + c^{-2} \wtP \big) \Big]
                    + \partial_j \Big[(\widetilde{\gamma}_c)^2\big(\wtRc + c^{-2}\wtP \big) \wtv^j \Big] \Big \rbrace \cdot
                    \Big \lbrace \dot{v}_k \dot{v}^k + c^{-2} (\widetilde{\gamma}_c)^2(\wtv_k \dot{v}^k)^2 \Big \rbrace        \\
                & + 2 c^{-2} (\widetilde{\gamma}_c)^4 \Big \lbrace \wtRc + c^{-2}\wtP \Big \rbrace  \\ 
                	& \qquad \cdot \Big \lbrace \wtv_k \dot{v}^k \dot{v}^j \partial_t \wtv_j
                    + \wtv_k \dot{v}^k \dot{v}^a \wtv^j \partial_j \wtv_a + c^{-2} (\widetilde{\gamma}_c)^2(\wtv_k \dot{v}^k)^2 \big( 
                    \wtv_j \partial_t \wtv^j + \wtv_a \wtv^j \partial_j \wtv^a \big) \Big \rbrace \\
                & + 2 \dot{\Ent}\mathit{f} + 2 \frac{\dot{P}\mathit{g}}{\wtQc} + 2 \dot{v}_j \mathit{h}^{(j)}.
            \end{split}
        \end{align}
	
			We observe here that in the case $c=\infty,$ \eqref{E:RescaledEnergyCurrentDivergence} reduces to the more palatable expression
  			\begin{align}
            \begin{split}   \label{E:RescaledEnergyCurrentDivergenceInfinity}
                 	\partial_{\mu} \Big(\Jscrdotinfinity^{\mu}\Big) &= \Big \lbrace \partial_t 
                 	\left(\frac{1}{\widetilde{Q}_{\infty}}\right) + 
                        \partial_j \left(\frac{\wtv^j}{\widetilde{Q}_{\infty}}\right) \Big \rbrace \cdot \dot{p}^2 
                        + \Big\lbrace \partial_t \big(\widetilde{R}_{\infty}\big) 
                    + \partial_j \big(\widetilde{R}_{\infty} \wtv^j \big) \Big\rbrace \cdot \dot{v}_k \dot{v}^k   \\
         					& + 2 \dot{\Ent}\mathit{f} + 2\frac{\dot{P} \mathit{g}}{\widetilde{Q}_{\infty}} + 2 \dot{v}_j 
         					\mathit{h}^{(j)}.
            \end{split}
        \end{align}

	\section{The Initial Data and the Uniform-in-$c$ Positivity of the Energy Currents}             \label{S:IVPc}
        In this section we describe a class of initial data for which our energy methods allow us to
        rigorously take the limit $c \to \infty$ in the EN$_{\kappa}^c$ system.
        The Cauchy surface we consider is $\{(t,\sforspace) \in \mathcal{M} \ | \ t=0\}.$

        \subsection{An $H^N$ perturbation of a uniform quiet fluid} 
        Initial data for the EP$_{\kappa}$ system are denoted by
        \begin{align} \label{E:EPkappaData}
            \scrVID_{\infty}(\sforspace) \eqdef
            (\EntID,\pID,\vID^1,\vID^2,\vID^3, \PhiID_{\infty},\PsiID_0,\PsiID_1,\PsiID_2,\PsiID_3),
        \end{align}
        where $\PsiID_0(\sforspace) \eqdef \partial_t
        \Phi(0,\sforspace)$ and $\PsiID_{j} \eqdef \partial_j
        \PhiID_{\infty}(\sforspace).$ We assume that $\scrVID_{\infty}$ is an $H^N$ perturbation
        of the constant state $\scrVb_{\infty}$, where
        \begin{equation}                                                                                \label{E:InitialConstant}
            \scrVb_{\infty} \eqdef (\Entbar,\pbar,0,0,0,\Phibar_{\infty},0,0,0,0),
        \end{equation}
        $\Entbar, \pbar$ are positive constants, and
        the constant $\Phibar_{\infty}$ is the unique solution to
        \begin{align}                                                                               \label{E:PhibarInfinityDef}
            {\kappa}^2 \Phibar_{\infty} + 4\pi G \mathfrak{R}_{\infty}(\bar{{\Ent}},\bar{p}) =
            0.
        \end{align}
        The constraint \eqref{E:PhibarInfinityDef} must be satisfied in order for equation
        \eqref{E:EPkappa4} to be satisfied by $\scrVb_{\infty}.$ By an $H^N$ perturbation, we mean that
        \begin{align}
            \|\scrWID_{\infty}\|_{H_{\scrWb_{\infty}}^N} < \infty,
        \end{align}
        where we use the notation $\scrWID_{\infty}$ and
        $\scrWb_{\infty}$ to refer to the first 5 components of
        $\scrVID_{\infty}$ and $\scrVb_{\infty}$ respectively. \emph{We emphasize that a further positivity restriction on the 
        initial data $\pID$ and $\EntID$ is introduced in Section \ref{SS:AdmissibleStateSpacec}, and that throughout this 
        article,} $N$ \emph{is a fixed integer satisfying}
        \begin{align} \label{E:NDef}
            N \geq 4.
        \end{align}

        \begin{remark}
        We require $N \geq 4$ so that Corollary \ref{C:SobolevCorollary} and Remark \ref{R:SobolevCalculusRemark} can be applied to 
        conclude that $l \in \bigcap_{k=0}^{k=2} C^k([0,T], H^{N-k}),$ where $l$ is defined in \eqref{E:lUniformTimeDef}; 
        this is a necessary hypothesis for Proposition \ref{P:Klein-GordonNocDependence}, which we use in
        our proof of Theorem \ref{T:UniformLocalExistenceENkappac}.
        \end{remark}

        Although we refer to $\PhiID_{\infty}$ and
        $\PsiID_{\nu}, \ \nu=0,1,2,3,$ as ``data,'' in the EP$_{\kappa}$ system, these 5 quantities are determined by $\EntID,\pID,\vID^1,\vID^2,\vID^3$
        through the equations (\ref{E:EPkappa2}'), \eqref{E:EPkappa4}, and \eqref{E:PhibarInfinityDef}, together with
        vanishing conditions at infinity on $\PhiID_{\infty} - \Phibar_{\infty}$ and
        $\PsiID_0:$
        \begin{align}
            \Delta \PhiID_{\infty} - \kappa^2 (\PhiID_{\infty} - \Phibar_{\infty}) & = 4 \pi G
            \big[\mathfrak{R}_{\infty}(\EntID,\pID) - \mathfrak{R}_{\infty}(\Entbar,\pbar) \big] \label{E:PhiIDinfinityIdentity}\\
            \Delta \PsiID_0 - \kappa^2 \PsiID_0 =
            - 4 \pi G \partial_t|_{t=0} \big(\mathfrak{R}_{\infty}(\Ent,p)\big)
            &= - 4 \pi G \partial_k \big(\mathfrak{R}_{\infty}(\EntID,\pID) \vID^k \big), \label{E:Psi0Identity}
        \end{align}
        where the integral kernel from \eqref{E:EPkappaPotential} can be used to compute $\PhiID_{\infty} - \Phibar_{\infty}$ and
        $\PsiID_0.$ We will nevertheless refer to the array $\scrVID_{\infty}$ as the ``data'' for the EP$_{\kappa}$ system.

        \begin{remark} \label{R:PhiIDinHNPlusOne}
            Remark \ref{R:KernelRemark} implies that $\PhiID_{\infty} \in H_{\Phibar_{\infty}}^{N+2}$ and 
            $\PsiID_{\nu} \in H^{N+1}$ for \\ 
            $\nu = 0,1,2,3.$
        \end{remark}

        We now construct data for the EN$_{\kappa}^c$ system from
        $\scrVID_{\infty}.$ Depending on which set of state-space variables we are working with,
        we denote the data for the EN$_{\kappa}^c$ system by
        \begin{align}
             \scrVID_c &\eqdef(\EntID,\pID,\vID^1,\vID^2,\vID^3,\PhiID_c,\mathring{\Psi}_0,\PsiID_1,\PsiID_2,\PsiID_3)
                 \label{E:ENkappacscrVID} \\
            \mbox{or} \ \VID_c 
            &\eqdef(\EntID,e^{4\PhiID_c/c^2}\pID,\vID^1,\vID^2,\vID^3,\PhiID_c,\mathring{\Psi}_0,\PsiID_1,\PsiID_2,\PsiID_3),
                \label{E:ENkappacVID}
        \end{align}
        where unlike in the EP$_{\kappa}$ case,
        $\PhiID_c,\mathring{\Psi}_0,\PsiID_1,\PsiID_2,$ and
        $\PsiID_3$ are data in the sense that the EN$_{\kappa}^c$ system is under-determined if
        they are not prescribed. We have chosen the data 
        $\EntID,\pID,\vID^1,\vID^2,\vID^3,\PsiID_0,\PsiID_1,\PsiID_2,\PsiID_3$ for the EN$_{\kappa}^c$ system
        to be the same as the data for the EP$_{\kappa}$ system,
        but for technical reasons described below and indicated in \eqref{E:PhiBarc} and \eqref{E:PhiIDcDef}, our requirement that
        there exists a constant background state typically constrains the datum $\PhiID_c$ so that it differs from $\PhiID_{\infty}$
        by a small constant that vanishes as $c \to \infty.$

        As in the EP$_{\kappa}$ system, we assume that $\scrVID_c$ is an
        $H^N$ perturbation of
        the constant state of the form (depending on which collection of state-space variables we are
        working with)
        \begin{align}
            \scrVb_c &\eqdef(\Entbar,\pbar,0,0,0,\Phibar_c,0,0,0,0) \label{E:scrVcConstant}\\
            \mbox{or} \ \Vb_c &\eqdef(\Entbar,\Pbar_c,0,0,0,\Phibar_c,0,0,0,0), \label{E:VcConstant}
        \end{align}
        where $\Entbar$ and $\pbar$ are the same constants appearing in
        $\scrVb_{\infty},$  $\Phibar_c$ is the unique solution to
        \begin{align}                                                                             \label{E:PhiBarc}
            {\kappa}^2 \Phibar_c + 4 \pi G e^{4 \Phibar_c / c^2}\big[\mathfrak{R}_{c}(\bar{{\Ent}},\bar{p}) - 3c^{-2} \bar{p} \big] = 
            0,
        \end{align}
        and $\Pbar_c \eqdef e^{4 \Phibar_c / c^2}\pbar.$ The
        constraint \eqref{E:PhiBarc} must be satisfied in order for equation \eqref{E:ENkappac4}
        to be satisfied by $\pbar, \Entbar,$ and $\Phibar_c.$
        Although the background potential $\Phibar_c$ for the EN$_{\kappa}^c$ system is
        not in general equal to the background potential $\Phibar_{\infty}$ for the
        EP$_{\kappa}$ system, it follows
        from the hypotheses \eqref{E:EOScHypothesis1} and \eqref{E:EOScHypothesis2} on the $c$-dependence of $\mathfrak{R}_c$
        that
        \begin{align} \label{E:PhibarcConvergence}
            \lim_{c \to \infty} \Phibar_{c} = \Phibar_{\infty}.
        \end{align}

        We now define the initial datum $\PhiID_c$ appearing in the arrays
        \eqref{E:ENkappacscrVID} and \eqref{E:ENkappacVID} by
        \begin{align} \label{E:PhiIDcDef}
            \PhiID_c &\eqdef \PhiID_{\infty} - \Phibar_{\infty} + \Phibar_{c}, 
        \end{align}
        which ensures that the deviation of $\PhiID_c$ from the background potential $\Phibar_c$ matches the
        deviation of $\PhiID_{\infty}$ from the background potential $\Phibar_{\infty}.$ We denote the first 5 components of
        $\scrVID_c,$ $\VID_c,$ $ \scrVb_c,$ and $\Vb_c$ by $\scrWID_c, \WID_c, \scrWb_c,$ and
        $\Wb_c$ respectively.

        \begin{remark}
            We could weaken the hypotheses by allowing the
            initial data for the EN$_{\kappa}^c$ system to deviate
            from the initial data for the EP$_{\kappa}$ system by an $H^N$
            perturbation that decays to $0$ rapidly enough as $c \to \infty.$ For
            simplicity, we will not pursue this analysis here.
        \end{remark}

        \subsection{The sets $\mathcal{O}, \mathcal{O}_2, \mathfrak{O}_2, \mathfrak{C}, K,$ and $\mathfrak{K}$} 
        \label{SS:AdmissibleStateSpacec} 
           In order to avoid studying the free boundary problem, and in order
            to avoid singularities in the energy currents \eqref{E:RescaledEnergyCurrent} and 
            \eqref{E:RescaledEnergyCurrentInfinity}, we assume that
            the initial pressure, energy density, and speed of sound are
            uniformly bounded from below by a positive constant. According
            to our assumptions \eqref{E:EOSAssumptions} on the equation of state,
            to achieve this uniform bound, it is sufficient to make the following further
            assumption on the initial
            data: that $\scrWID_{\infty}(\mathbb{R}^3)$ is contained in a compact subset of the
            following open subset $\mathcal{O}$ of the state-space $\mathbb{R}^{5},$
            the \emph{admissible subset of truncated state-space}, defined by

            \begin{equation}                                                                \label{E:AdmissibleSubset}
                \mathcal{O} \eqdef \big\{ \scrW =(\Ent,p,v^1,v^2,v^3) \in \mathbb{R}^{5} \ | \ \Ent > 0, p > 0 \big\}.
            \end{equation}

            Therefore, we assume that $\scrWID_{\infty}(\mathbb{R}^3) \subset \mathcal{O}_1$ and
            $\scrWb_{\infty} \in \mathcal{O}_1,$ where
            $\mathcal{O}_1$ is a precompact open set with $\mathcal{O}_1
            \Subset \mathcal{O},$ and $``\Subset''$ means that ``the closure is compact and contained in the interior of.'' We then fix 
            convex precompact open subsets
            $\mathcal{O}_2$ and $\mathfrak{O}_2 $ with $\mathcal{O}_1 \Subset \mathcal{O}_2 \Subset 
            \mathfrak{O}_2 \Subset \mathcal{O},$ and define $\mathfrak{C}$ to be the projection of $\bar{\mathfrak{O}}_2$ onto the first
            two axes, where $\bar{\mathfrak{O}}_2$ denotes the closure of $\mathfrak{O}_2.$ We assume that with this 
            definition of $\mathfrak{C},$ hypotheses \eqref{E:EOScHypothesis1} and \eqref{E:EOScHypothesis2} are 
            satisfied by the equation of state. Consequently, property \eqref{E:PhibarcConvergence} shows that for all large $c$
            including $c = \infty,$ $\WID_c (\mathbb{R}^3) \Subset \mathcal{O}_2$ and $\Wb_c \in \mathcal{O}_2;$ also note that 
            for all $c,$ $\scrWID_c = \scrWID_{\infty} = \WID_{\infty}.$ 

            We now address the variables $\big(\Phi,\partial_t \Phi, \partial_1 \Phi, \partial_2 \Phi,\partial_3
            \Phi\big).$ In Section \ref{S:UniforminTimeLocalExistence}, we will use energy estimates to prove the existence of 
            an interval $[0,T]$ and a cube of the form $[-a,a]^5$ such that for all large $c$ including 
            $c=\infty,$ we have \\
            $\big(\Phi,\partial_t \Phi, \partial_1 \Phi, \partial_2 \Phi,\partial_3
            \Phi\big)([0,T] \times\mathbb{R}^3) \subset [-a,a]^5.$ Furthermore, 
            it will follow from the discussion in Section \ref{S:UniforminTimeLocalExistence} that for all large $c$ including 
            $c=\infty,$ we have \\
            $\big(\PhiID_c,\PsiID_0, \PsiID_1, \PsiID_2,\PsiID_3\big)(\mathbb{R}^3) \Subset 
            \mbox{Int}([-a,a]^5).$ The compact convex set $K,$ then, as given in \eqref{E:Kdef} below, will be defined to be
            $\bar{\mathcal{O}}_2 \times [-a,a]^5.$ It follows from the above discussion that for all large $c$ including
            $c=\infty,$ we have $\VID_c(\mathbb{R}^3) \Subset \mbox{Int}(K)$ and $\Vb_c \in \mbox{Int}(K).$ Our goal will be to show that 
            the solution $\Vbold_c$ to \eqref{E:ENkappac1} - \eqref{E:PDefcII} launched by the initial data $\VID_c$ exists on a time 
            interval $[0,T]$ that is independent of (all large) $c$ and remains in $K.$

            We now discuss the simple construction of $\mathfrak{K}:$ based on the above construction, it 
            follows from definitions \eqref{E:Wboldarray} - \eqref{E:scrVboldarray} that for all large $c$ including $c=\infty,$ we have 
            $\Vbold \in K \implies \scrW \in \bar{\mathfrak{O}}_2.$ As given in \eqref{E:frakKdef}, we will then
            define the compact convex\footnote{Proposition \ref{P:SobolevTaylor}
            requires the convexity of $K$ and  $\mathfrak{K},$ and the estimate
            \eqref{E:ModifiedSobolevEstimateConstantArray} also requires that $\Vb_c \in K$ and $\scrVb_c \in \mathfrak{K}.$ In practice, 
            $K$ and $\mathfrak{K}$ can be chosen to be cubes.} set $\mathfrak{K}$ by $\mathfrak{K} \eqdef \bar{\mathfrak{O}}_2 \times 
            [-a,a]^5,$ so that for all large $c$ including $c=\infty,$ we also have that $\Vbold \in K \implies \scrV \in \mathfrak{K}.$ 
            As in 
            the previous discussion, it follows that for all large $c$ including $c=\infty,$ we have $\scrVID_c(\mathbb{R}^3) 
            \Subset \mbox{Int}(\mathfrak{K})$ and $\scrVb_c \in \mbox{Int}(\mathfrak{K}).$ 
       			
       			\subsection{The uniform-in-$c$ positive definiteness of $\Jscrdotc^0$} \label{S:JdotUniformPositiveDefinitec}
            As mentioned at the beginning of Section \ref{S:EnergyCurrentsc}, we
            will use the quantity $\|\Jscrdotc^0(t)\|_{L^1}$ to control $\|\Wdot(t)\|_{L^2}^2,$ where $\Jscrdotc$ is 
            an energy current for the variation $\Wdot$ with coefficients defined by a BGS $\Vboldt.$
            Since we seek estimates that are uniform in $c,$ we will
            show that under some simple assumptions on the BGS $\Vboldt,$ it
            follows that $\Jscrdotc^0$ is uniformly positive definite in $\Wdot$ for all large $c.$ Let us now formulate this 
            precisely as a lemma.
            
            \begin{lemma} \label{L:UniformPositivitywithc}
            Let $\Jscrdotc$ be the energy current \eqref{E:RescaledEnergyCurrent} for the variation $\Wdot$ defined by the BGS $\Vboldt.$  
            Assume that $\Wboldt(t,\sforspace) \in \bar{\mathcal{O}}_2$ and that $|\Phit(t,\sforspace)| \leq Z,$
            where $\Wboldt$ denotes the first 5 components of $\Vboldt,$ and $\bar{\mathcal{O}}_2$ is defined in Section 
            \ref{SS:AdmissibleStateSpacec}. Then there exists a constant 
            $C_{\bar{\mathcal{O}}_2,Z}$ with $0 < C_{\bar{\mathcal{O}}_2,Z} < 1$ such that
            \begin{equation}                                                                        \label{E:UniformPositivitywithc}
                C_{\bar{\mathcal{O}}_2,Z}|\Wdot|^2 \leq
                \Jscrdotc^0(\Wdot,\Wdot) \leq C_{\bar{\mathcal{O}}_2,Z}^{-1} |\Wdot|^2
            \end{equation}
            holds for all large $c$ including $c={\infty}.$ 
            \end{lemma}
            \begin{proof}
            It is sufficient prove inequality \eqref{E:UniformPositivitywithc}
            when $|\Wdot| = 1$ since it is invariant under any rescaling of
            $\Wdot.$ Let $\scrWtilde, \scrVtilde$ be the arrays related to the arrays
            $\Wboldt, \Vboldt$ as defined in \eqref{E:Wboldarray} - \eqref{E:scrVboldarray}. Our assumptions imply the existence of a 
            compact set $\mathfrak{D}$ depending only on $\bar{\mathcal{O}}_2$ and $Z$ such that for all large $c,$
            $\scrVtilde(t,\sforspace) \in \mathfrak{D}.$
            
            Recall that $\Jscrdotinfinity$ is defined in \eqref{E:RescaledEnergyCurrentInfinity} and that
            $\Jscrdotinfinity^0$ is manifestly positive definite in the variations\footnote{To be consistent the notation used in 
            formula \eqref{E:RescaledEnergyCurrentInfinity},
            it would be ``more correct'' to use the symbol $\scrWdot$ to denote the variations appearing as arguments in 
            $\Jscrdotinfinity(\cdot,\cdot).$ However, for the purposes of this proof, there is no harm in identifying $\Wdot = \scrWdot$ 
            since in this context, these placeholder variables merely represent the arguments of $\Jscrdotinfinity$ when viewed as a 
            quadratic form in the variations.} $\Wdot$ if $\widetilde{p} > 0.$ If we view $\Jscrdotinfinity^0$ as a function of
            $(\Wdot,\scrWtilde),$ then by uniform continuity, there is a constant $0<C(\mathfrak{D})<1$
            such that \\
            $C(\mathfrak{D})|\Wdot|^2 \leq \Jscrdotinfinity^0 \leq C(\mathfrak{D})^{-1}
            |\Wdot|^2$ holds on the compact set $\lbrace |\Wdot| = 1 \rbrace
            \times \Pi_5(\mathfrak{D}),$ where $\Pi_5(\mathfrak{D})$ is the projection of $\mathfrak{D}$ onto
            the first five axes. Furthermore, if we also view $\Jscrdotc^0$ as a function of
            $(\Wdot,\scrVtilde),$ then by Lemma \ref{L:MultiplyFcGcSobolevEstimate}, Lemma \ref{L:OrderctoNegativeTwoEstimates}, 
            \eqref{E:RescaledEnergyCurrent}, and \eqref{E:RescaledEnergyCurrentInfinity} we have that $\Jscrdotc^0 = \Jscrdotinfinity^0 + 
            \mathfrak{F}_c \cdot |\Wdot|^2,$ where $\mathfrak{F}_c 
            \in \mathcal{R}^{N}(c^{-2};\mathfrak{D};\scrVtilde).$ \eqref{E:UniformPositivitywithc} now easily follows:
            $C_{\bar{\mathcal{O}}_2,Z}$ can be any positive number that is strictly smaller than 
            $C(\mathfrak{D}).$ 
            \end{proof}
            
            \begin{remark} \label{R:NoZdependence}
            	If $c=\infty,$ then the coefficients of the quadratic form $\Jscrdotinfinity^0$ are independent of
            	$\Phi.$ It follows that in this case, the constant $C_{\bar{\mathcal{O}}_2}$ from
            	\eqref{E:UniformPositivitywithc} is independent of $Z.$
            \end{remark}

    \section{Smoothing the Initial Data} \label{S:SmoothingtheData}
    
    		For technical reasons, we need to smooth the initial data. Without smoothing,
    		the terms on the right-hand sides of \eqref{E:IterateInhomogeneouscf} - \eqref{E:IterateInhomogeneouschj}
    		involving the derivatives of the initial data could be unbounded in the $H^N$ norm.
   			To begin, we fix a Friedrich's mollifier $\chi(\sforspace);$ i.e., $\chi \in
        C_c^\infty(\mathbb{R}^3), \ \mbox{supp}(\chi) \subset \{\sforspace \ | \
        |\sforspace| \leq 1 \}, \chi \geq 0,$ and $\int \chi \ d^3\sforspace = 1.$  For $\epsilon > 0,$
        we set $\chi_\epsilon(\sforspace) \overset{\mbox{\tiny{def}}}{=}
        {\epsilon}^{-3} \chi(\frac{\sforspace}{\epsilon}).$ We smooth the first 5 components $\scrWID_{\infty}$ of the data
        $\scrVID_{\infty}$ defined in \eqref{E:EPkappaData} with $\chi_\epsilon,$ defining
        $\chi_\epsilon \scrWID_{\infty} \in C^{\infty}$ by
        \begin{equation}
            \chi_\epsilon \scrWID_{\infty}(\sforspace) \eqdef \int_{\mathbb{R}^3}
            \chi_\epsilon(\sforspace- \sforspace')\scrWID_{\infty}(\sforspace') \,
            d^3 \sforspace'.
        \end{equation}
    		Note that we do not smooth the data $(\PhiID_c, \PsiID_0) \in H_{\Phibar_{c}}^{N+2} \times H^{N+1}$ because 
    		by Remark \ref{R:PhiIDinHNPlusOne} and definition \eqref{E:PhiIDcDef}, they already have sufficient regularity.

        The following property of such a mollification is well known:
        \begin{align}
            \underset{\epsilon \rightarrow 0^+}{\lim}
                \|\chi_\epsilon \scrWID_{\infty} - \scrWID_{\infty} \|_{H^N} &=0
                \label{E:Mollificationc1}.
        \end{align}
                We will choose below an $\epsilon_0>0.$ Once chosen, we define
        
                \begin{align}
            {^{(0)}\scrWID} &\eqdef \big(\EntIDSmoothed,\pIDSmoothed,{^{(0)}\mathring{\mathbf{v}}}\big)
                \eqdef  \chi_{\epsilon_0} \scrWID_{\infty}    \label{E:IterateInitialValuec1}\\
            {^{(0)}\WID_c}    &\eqdef \big(\EntIDSmoothed,e^{4 \PhiID_c/c^2}\cdot\pIDSmoothed,{^{(0)}\mathring{\mathbf{v}}}\big),                                  \label{E:IterateInitialValuec2}
        \end{align}
        where $\PhiID_c$ is defined in \eqref{E:PhiIDcDef}. By Sobolev embedding, the assumptions on the initial data
            $\WID_c$, which are the first 5 components of the data $\VID_c$ defined in \eqref{E:ENkappacVID},
            by Lemma \ref{L:MultiplyFcGcSobolevEstimate}, by \eqref{E:expPhiOvercSquaredOrderctoNegativeTwoEstimate}, and
            by the mollification property \eqref{E:Mollificationc1}, $\exists \lbrace \Lambda_1>0 \land \epsilon_0 > 0 \rbrace$ such 
            that

        \begin{align}
            \mbox{for all large $c$}, \ \|\Wbold - \WIDSmoothed_c \|_{H^N} &\leq \Lambda_1 \Rightarrow
                \Wbold \in \bar{\mathcal{O}}_2                               \label{E:WellDefinedc}\\
            \|\WIDSmoothed_c - \WID_c\|_{H^N} &\leqc
                C_{\bar{\mathcal{O}}_2,Z} \cdot \frac{\Lambda_1}{2},                                  \label{E:IntialValueEstimatec}
        \end{align}
        where $\bar{\mathcal{O}}_2$ is defined in Section \ref{SS:AdmissibleStateSpacec},
        and $C_{\bar{\mathcal{O}}_2,Z}$ is the constant from \eqref{E:UniformPositivitywithc}. Here, $Z$ is a fixed constant that will 
        serve as an upper bound for $\| \Phi(t) \|_{L^{\infty}}$ on a certain time interval, where $\Phi$ will
        be a solution variable in the EN$_{\kappa}^c$ system. We explain this fixed value of $Z,$ 
        given in expression \eqref{E:Zdef} below, in detail in Section \ref{SS:TechnicalLemmas}. Note that 
        according to this reasoning, $\Lambda_1 = \Lambda_1(\bar{\mathcal{O}}_2;Z).$

        \begin{remark}
        \label{R:SmoothedDataLargecNPlusOneEstimate}
            Because $\|\VID_c\|_{H_{\Vb_c}^{N}}, \|\VID_c \|_{L^{\infty}}, \|\VIDSmoothed_c \|_{H_{\Vb_c}^{N+1}}$ and 
            $\|\VIDSmoothed_c \|_{L^{\infty}}$ enter into our 
            Sobolev estimates below, it is an important fact that these quantities are uniformly
            bounded for all large $c.$ By \eqref{E:PhibarcConvergence}, \eqref{E:PhiIDcDef},
            definition \eqref{E:IterateInitialValuec2}, and Sobolev embedding, to obtain uniform bounds for $\VIDSmoothed_c,$
            we only need to show that \\
            $\mid \mid e^{4 \PhiID_c/c^2}\cdot\pIDSmoothed \mid \mid_{H_{e^{4 \Phibar_c/c^2}\cdot\pbar}^{N+1}}$
            is uniformly bounded for all large $c.$ This fact follows from Lemma \ref{L:HjbySubtractingConstant},
            Lemma \ref{L:MultiplyFcGcSobolevEstimate}, and \eqref{E:expPhiOvercSquaredOrderctoNegativeTwoEstimate}. 
        		Such a uniform bound is used, for example, in the estimate \eqref{E:InitialDataInhomogeneousLargecSobolevEstimateB}.
        		We can similarly obtain the uniform bounds for $\VID_c;$ we use such a bound, for example, in the proof of
        		\eqref{E:PartialtlofZerobound}.   
        \end{remark}

\section{Uniform-in-Time Local Existence for EN$_{\kappa}^c$} \label{S:UniforminTimeLocalExistence}

    In this section we prove our first important theorem, namely that there is a uniform time interval $[0,T]$
    on which solutions to the EN$_{\kappa}^c$ system having the initial data $\VID_c$ exist, as long as $c$ is large enough.

    \subsection{Local existence and uniqueness for EN$_{\kappa}^c$ revisited} 
		Let us first recall the following local existence result proved in \cite{jS2008a}, in which it was not yet shown that
    the time interval of existence can be chosen independently of all large $c.$

    \begin{theorem} \label{T:LocalExistencec} {\bf(EN$_{\kappa}^c$ Local Existence Revisited)}                         		 
    		Let $\VID_c(\sforspace)$ be initial data \eqref{E:ENkappacscrVID} for the EN$_{\kappa}^c$ system
        \eqref{E:ENkappac1} - \eqref{E:PDefcII}
        that are subject to the conditions described in Section \ref{S:IVPc}. Assume that the 
        equation of state is ``physical'' as described in Section \ref{SS:ENkappacDerivation}. Then
        for all large (finite) $c,$ there exists a
        $T_c > 0$ such that \eqref{E:ENkappac1} - \eqref{E:PDefcII}
        has a unique classical solution $\Vbold \in C_b^2([0,T_c] \times \mathbb{R}^3)$ of the form \\ 
        $\Vbold =({\Ent},P,v^1,v^2,v^3,\Phi,\partial_t \Phi,\partial_1
        \Phi,\partial_2 \Phi,\partial_3 \Phi)$ with $\Vbold(0,\sforspace)=\VID_c(\sforspace).$
        The solution satisfies $\Vbold([0,T_c] \times \mathbb{R}^3) \subset K,$
        where the ($c-$independent) compact convex set $K$ is defined 
        in \eqref{E:Kdef}. Furthermore, $\Vbold \in \bigcap_{k=0}^{k=2}C^k([0,T_c],H_{\Vb_c}^{N-k})$ and \\ 
        $\Phi \in C_b^3([0,T_c] \times \mathbb{R}^3) \cap \bigcap_{k=0}^{k=3}C^k([0,T_c],H_{\Phibar_c}^{N+1-k}),$
        where the constants $\Vb_c$ and $\Phibar_c$ are defined by \eqref{E:VcConstant} and  \eqref{E:PhiBarc} respectively.
    \end{theorem}

    \begin{remark} \label{R:ExtraBanachSpaceDifferentiability}
				Although they are not explicitly proved in \cite{jS2008a}, the facts that \\ 
				$\Vbold \in C_b^2([0,T_c] \times \mathbb{R}^3)$ and that 
				$\Vbold$ is twice differentiable in $t$ as a
        map from $[0,T_c]$ to $H_{\Vb_c}^{N-2}$ follow from our assumption that $N \geq
        4$ \big(i.e., for $N \geq 4,$ it can be shown that 
        $\Vbold \in C_b^{N-2}([0,T_c] \times \mathbb{R}^3)\cap \bigcap_{k=0}^{k=N-2}
        C^k([0,T_c],H_{\Vb_c}^{N-k}$)\big). Also, by Corollary \ref{C:SobolevCorollary}, we 
        have that 
         $p \in \bigcap_{k=0}^{k=2}C^k([0,T_c],H_{\pbar}^{N-k}),$ since $p = P e^{-4\Phi/c^2}.$
        
        The proof of the claim that $T_c$ can be chosen such that $\Vbold([0,T_c] \times \mathbb{R}^3) \subset K$ is based on the fact
        $\VID_c(\mathbb{R}^3) \Subset \mbox{Int}(K)$ (see Section \ref{SS:AdmissibleStateSpacec}),
        together with the continuity result from the theorem and Sobolev embedding. 
     		
    \end{remark}

    \begin{remark}
        The case $c=\infty$ is discussed separately in Theorem \ref{T:EPkappaLocalExistence}.
    \end{remark}

    \begin{remark}
        The local existence theorem in \cite{jS2008a} was proved using the relativistic state-space variables 
        $U^{\nu} \eqdef e^{\phi}u^{\nu}.$ However, the form of the Newtonian change of variables made in sections 
        \ref{SS:ENkappacDerivation} and \ref{SS:NewtonianReformulationc}, together with Corollary \ref{C:SobolevCorollary}, allows us
        to conclude Sobolev regularity in one set of state-space variables if the same regularity is known in the other set of variables.
    \end{remark}
		
		The following corollary, which slightly extends the lifespan of the solution and
		also allows us to conclude stronger regularity properties from weaker regularity assumptions, will soon be used in our proof of 
		Proposition \ref{P:ContinuationPrinciple}.
    \begin{corollary} \label{C:ContinuationPrinciple}
        Let $\Vbold(t,\sforspace)$ be a solution to the EN$_{\kappa}^c$ system \eqref{E:ENkappac1} - \eqref{E:PDefcII}
        that has the regularity properties $\Vbold \in C_b^1([0,T] \times \mathbb{R}^3) \cap L^{\infty}([0,T],H_{\Vb_c}^{N}).$
        Let $\mathcal{O}$ be the admissible subset of truncated state-space
        defined in \eqref{E:AdmissibleSubset}, and let $\Pi_5: \mathbb{R}^{10} \rightarrow \mathbb{R}^5$ denote projection
        onto the first $5$ axes. Assume that $\Vbold([0,T] \times \mathbb{R}^3) \subset K$ and that $\Vb_c \in \mbox{Int}(K),$
        where $K \subset \mathbb{R}^{10}$ is a compact convex set such that $\Pi_5 (K) \Subset \mathcal{O}.$
      	Then there exists an $\epsilon > 0$ such that 
      	
      	\begin{align} \label{E:ContinuationPrinciple}
            \Vbold \in C_b^2([0,T + \epsilon] \times \mathbb{R}^3) \cap \bigcap_{k=0}^{k=2}C^k ([0,T + \epsilon],H_{\Vb_c}^{N-k}).
        \end{align}
        
    \end{corollary}
      \begin{proof}
        We apply Theorem \ref{T:LocalExistencec} to
        conclude\footnote{Theorem \ref{T:LocalExistencec} can be easily modified to obtain a solution
        that exists both ``forward'' and ``backward'' in time.} that for each $T' \in [0,T],$ there exists an $\epsilon > 0,$ depending
        on $T',$ and a solution $\Vboldt$ to the EN$_{\kappa}^c$ system such that \\ 
        $\Vboldt \in C_b^2([T' - \epsilon,T' + \epsilon]\times 
        \mathbb{R}^3) \cap \bigcap_{k=0}^{k=2}C^k([T'-\epsilon,T'+\epsilon],H_{\Vb_c}^{N-k})$ and such that 
        $\Vboldt(T')=\Vbold(T').$ Furthermore, the uniqueness argument from \cite{jS2008a}, which is based on local energy estimates, can 
        be easily modified to show that solutions to the EN$_{\kappa}^c$ system are unique in the class 
        $C^1([T'-\epsilon,T'+\epsilon] \times \mathbb{R}^3).$ Therefore $\Vbold \equiv \Vboldt$ on their common slab of spacetime 
        existence. Corollary \ref{C:ContinuationPrinciple} thus follows. 
      \end{proof}

    In addition to Theorem \ref{T:LocalExistencec}, our proof of Theorem 
    \ref{T:UniformLocalExistenceENkappac} also requires an additional key ingredient, namely the following continuation principle for 
    Sobolev norm-bounded solutions:

    \begin{proposition} \textbf{(Continuation Principle)}                                 \label{P:ContinuationPrinciple}
        Let $\VID_c(\sforspace)$ be initial data \eqref{E:ENkappacscrVID} for the EN$_{\kappa}^c$ system
        \eqref{E:ENkappac1} - \eqref{E:PDefcII} that are subject to the conditions described in Section
        \ref{S:IVPc}, and let $T>0.$ Assume that $\Vbold \in C^1([0,T)\times \mathbb{R}^3)\cap 
        \bigcap_{k=0}^{k=1}C^k([0,T),H_{\Vb_c}^{N-k})$ is the unique classical solution existing on $[0,T)$ launched by
        $\VID_c(\sforspace).$ Let $\mathcal{O}$ be the admissible subset of truncated state-space
        defined in \eqref{E:AdmissibleSubset}, and let $\Pi_5: \mathbb{R}^{10} \rightarrow \mathbb{R}^5$ denote projection
        onto the first $5$ axes. Assume that there are constants $M_1, M_2 >0,$ a compact set $K \subset \mathbb{R}^{10}$ 
        with $\Pi_5 (K) \Subset \mathcal{O},$ and a set $U \Subset \mbox{Int}(K)$ such that the following three estimates hold for any 
        $T' \in [0,T):$ 
        \begin{enumerate}
            \item $\mid\mid\mid \Vbold \mid\mid\mid_{H_{\Vb_c}^N,T'} \leq M_1$
            \item $\mid\mid\mid \partial_t \Vbold \mid\mid\mid_{H^{N-1},T'} \leq M_2$
            \item $\Vbold([0,T'] \times \mathbb{R}^3) \subset U.$
        \end{enumerate}
        Then there exists an $\epsilon > 0$ such that 
        \begin{align}
            \Vbold \in C_b^2([0,T + \epsilon] \times \mathbb{R}^3) &\cap \bigcap_{k=0}^{k=2}C^k([0,T + \epsilon],H_{\Vb_c}^{N-k}) 
            \label{E:ExtraRegularity} \\ 
            &\mbox{and} \  \ \Vbold([0,T + \epsilon] \times \mathbb{R}^3) \subset K. \notag
        \end{align}
    		\end{proposition}
    		
    		\begin{remark}
    			Hypothesis $(2)$ is redundant; it can be deduced from hypothesis $(1)$ by using the equations
    			to solve for $\partial_t \Vbold$ and then applying \eqref{E:ModifiedSobolevEstimate}.
    		\end{remark}    
        
        \begin{proof}
        We will first show that there
        exists a $\Vbold^* \in H_{\Vb_c}^N$
        such that
        \begin{align}
            \lim_{n \to \infty} \|\Vbold(T_n) - \Vbold^*\|_{H^{N-1}}=0    \label{E:ConvergenceinHNMinusOne}
        \end{align}
        holds for any sequence $\lbrace T_n \rbrace$ of time values
        converging to $T$ from below.

        If $\lbrace T_n \rbrace$ is such a sequence, then hypothesis $(2)$ implies that \\
        $\|\Vbold(T_j) - \Vbold(T_k)\|_{H^{N-1}} \leq M_2 |T_j -
        T_k|.$ By the completeness of $H^{N-1},$ there exists a $\Vbold^* \in H_{\Vb_c}^{N-1}$ such
        that \eqref{E:ConvergenceinHNMinusOne} holds, and it is easy to check that
        $\Vbold^*$ does not depend on the sequence $\lbrace T_n \rbrace.$
        By hypothesis $(1),$ we also have that $\lbrace\Vbold(T_n)\rbrace$ converges weakly in
        $H_{\Vb_c}^N$ to $\Vbold^*$ as $n \to \infty$ and that $\|\Vbold^*\|_{H_{\Vb_c}^N} \leq M_1.$ We now fix a number $N'$ with
        $5/2 < N' < N.$ By Proposition \ref{P:SobInterpolation}, we have that $\lim_{n \to \infty} \|\Vbold(T_n) -
        \Vbold^*\|_{H^{N'}}=0.$ Consequently, if we define $\Vbold(T)\eqdef \Vbold^*,$ it follows 
        that $\Vbold \in C^0([0,T],H_{\Vb_c}^{N'}) \cap L^{\infty}([0,T],H_{\Vb_c}^{N}).$ Using the fact that
        $N' > 5/2,$ together with the embedding of $H^{N'}(\mathbb{R}^3)$ into appropriate H\"{o}lder spaces, it
        can be shown that \\
        $\Vbold \in C^0([0,T],H_{\Vb_c}^{N'})
        \implies \Vbold, \partial \Vbold \in C_b^0([0,T] \times \mathbb{R}^3);$ i.e., we can
        continuously extend $\Vbold, \partial \Vbold$ to the slab $[0,T] \times\mathbb{R}^3.$ 
        
        To conclude that
        $\Vbold \in C_b^1([0,T] \times \mathbb{R}^3),$ we will show that $\partial_t \Vbold$ extends continuously 
        to $[0,T] \times \mathbb{R}^3.$ To this end, we use the EN$_{\kappa}^c$ equations to solve for $\partial_t \Vbold:$ 
        \begin{align} \label{E:PartialtVIsolated}
        	\partial_t \Vbold = \mathfrak{F}(\Vbold,\partial \Vbold),
        \end{align}
        where $\mathfrak{F} \in C^N.$ Since $\Vbold, \partial \Vbold \in C_b^0([0,T] \times \mathbb{R}^3),$ the right-hand
        side of \eqref{E:PartialtVIsolated} has been shown to extend continuously so that it is an element of $C_b^0([0,T] \times 
        \mathbb{R}^3).$ Furthermore, since $\Vbold \in C^1([0,T) \times \mathbb{R}^3)$ by assumption, it follows from the previous
        conclusions and elementary analysis that 
        $\partial_t \Vbold$ exists classically on $[0,T] \times \mathbb{R}^3$ and that 
        $\partial_t \Vbold \in C_b^0([0,T] \times 
        \mathbb{R}^3),$ thus implying that $\Vbold \in C_b^1([0,T] \times \mathbb{R}^3).$ The additional conclusions in 
        \eqref{E:ExtraRegularity} now follow from Corollary \ref{C:ContinuationPrinciple} and continuity.
        \end{proof}
        
        \begin{remark}
        	Proposition \ref{P:ContinuationPrinciple} shows that if the solution $\Vbold$ blows up at time $T,$ then either 
        	$\lim_{T' \uparrow T} \mid\mid\mid \Vbold \mid\mid\mid_{H_{\Vb_c}^N,T'} = \infty, \
        	\lim_{T' \uparrow T} \mid\mid\mid \partial_t \Vbold \mid\mid\mid_{H^{N-1},T'} =\infty,$
        	or $\Vbold(T',\mathbb{R}^3)$ escapes\footnote{We are assuming here
        	that on the set $\lbrace (\Ent,p) \ | \ \Ent>0, \ p>0 \rbrace,$ 
        	the function $\mathfrak{R}_c$ is ``physical'' as described in Section \ref{SS:ApplicationtoENkappac}
        	and is and sufficiently regular.} every compact subset of $\mathcal{O} \times 
        	\mathbb{R}^5$ as $T' \uparrow T,$ where $\mathcal{O}$ is defined in \eqref{E:AdmissibleSubset}.
        \end{remark}

    \begin{remark}
    	Although the main theorems in this article require that $N \geq 4,$ Corollary \ref{C:ContinuationPrinciple} and Proposition 
    	\ref{P:ContinuationPrinciple} are also valid for $N =3,$ except that the conclusion $\Vbold \in C_b^2([0,T + \epsilon] \times 
    	\mathbb{R}^3)$ must be replaced with $\Vbold \in C_b^1([0,T + \epsilon] \times \mathbb{R}^3),$ and the conclusion
    	$\Vbold \in C^2([0,T + \epsilon],H_{\Vb_c}^{N-2})$ does not hold. 
    \end{remark}

    \subsection{The uniform-in-time local existence theorem} 							\label{SS:UniforminTimeExistence} 
        We now state and prove the uniform time of existence theorem.

        \begin{theorem}    \label{T:UniformLocalExistenceENkappac} {\bf (Uniform Time of Existence)}
        Let $\scrVID_{\infty}$ denote initial data \eqref{E:EPkappaData} for
       	the EP$_{\kappa}$ system \eqref{E:EPkappa1} - \eqref{E:QinfinityRelationship}
       	that are subject to the conditions described in Section
       	\ref{S:IVPc}. Let $\VID_c$ denote the corresponding initial data \eqref{E:ENkappacVID} for
       	the EN$_{\kappa}^c$ system \eqref{E:ENkappac1} - \eqref{E:PDefcII}
       	constructed from $\scrVID_{\infty}$ as described in Section \ref{S:IVPc}, and 
       	let ${^{(0)}\WID_c}$ denote the smoothing \eqref{E:IterateInitialValuec2} of the first $5$ components of
       	$\VID_c$ as described in Section \ref{S:SmoothingtheData}. Assume that the 
       	$c-$indexed equation of state satisfies the hypotheses \eqref{E:EOScHypothesis1} and \eqref{E:EOScHypothesis2}
       	and is ``physical'' as described in sections \ref{SS:ENkappacDerivation} and \ref{SS:ApplicationtoENkappac},
        and let $K$ be the fixed compact subset of $\mathbb{R}^{10}$ defined in \eqref{E:Kdef}. Then there exist $c_0 > 0$ and $T > 0,$
        with $T$ not depending on $c,$ such that for $c \geq c_0,$ $\VID_c$ launches a unique classical solution $\Vbold$ to 
        \eqref{E:ENkappac1} -  \eqref{E:PDefcII}
        that exists on the slab $[0,T] \times \mathbb{R}^3$ and that has the properties
        $\Vbold(0,\sforspace)=\VID_c(\sforspace)$ and $\Vbold([0,T] \times \mathbb{R}^3) \subset K.$ 
        The solution is of the form $\Vbold =(\Ent,P,v^1,v^2,v^3,\Phi,\partial_t \Phi, \partial_1
        \Phi,\partial_2 \Phi,\partial_3 \Phi)$ and has the regularity properties \\
        $\Vbold \in C_b^2([0,T] \times \mathbb{R}^3)\cap 
        \bigcap_{k=0}^{k=2}C^k([0,T],H_{\Vb_c}^{N-k})$ and \\
        $\Phi \in C_b^3([0,T] \times \mathbb{R}^3) \cap \bigcap_{k=0}^{k=3}C^k([0,T],H_{\Phibar_c}^{N+1-k}),$
        where the constants $\Vb_c$ and $\Phibar_c$ are defined by \eqref{E:VcConstant} and \eqref{E:PhiBarc} respectively. 
        Furthermore, with $p \eqdef P e^{-4 \Phi/c^2},$ there exist constants $\Lambda_1,\Lambda_2,L_1,L_2,L_3,L_4 >0$ such that

        \begin{subequations}
            \begin{align}
            \mid\mid\mid \Wbold - \WIDSmoothed_c \mid\mid\mid_{H^N,T} \leqc \Lambda_1
                    \label{E:UniformTime1a}\\
            \mid\mid\mid \Phi - \PhiID_c \mid\mid\mid_{H^{N+1},T} \leqc \Lambda_2 \label{E:UniformTime1b}\\
            \mid\mid\mid \partial_t \Wbold \mid\mid\mid_{H^{N-1},T} \leqc L_1         \label{E:UniformTime1c}\\
            \mid\mid\mid \partial_t \Phi \mid\mid\mid_{H^{N},T} \leqc L_2 \label{E:UniformTime1d}\\
            \mid\mid\mid \partial_t^2 \Ent \mid\mid\mid_{H^{N-2},T}, \ \mid\mid\mid \partial_t^2 p \mid\mid\mid_{H^{N-2},T}
                \leqc L_3  \label{E:UniformTime1e}\\
            c^{-1} \mid\mid\mid \partial_t^2 \Phi \mid\mid\mid_{H^{N-1},T} \leqc L_4. \label{E:UniformTime1f}
            \end{align}
        \end{subequations}
    \end{theorem}

        \subsubsection{Outline of the structure of the proof of Theorem \ref{T:UniformLocalExistenceENkappac}}
            We prove Theorem \ref{T:UniformLocalExistenceENkappac} via the method of continuous induction (``bootstrapping''). After
            defining the constants $\Lambda_1, \Lambda_2, L_2',$ and $L_4,$ we make
            the assumptions \eqref{E:UniformTime1aAgain} - \eqref{E:UniformTime1fPrimeAgain}. These
            assumptions hold at $\tau =0$ and therefore, by Theorem \ref{T:LocalExistencec}, there exists a maximal
            interval $\tau \in [0,T_c^{max})$ on which the solution exists and on which the assumptions hold. Based on these estimates, we 
            use a collection of technical lemmas derived from energy estimates to conclude that the bounds 
            \eqref{E:UniformtimeApriori1} - \eqref{E:UniformtimeApriori7} hold for \\
            $\tau \in [0,T_c^{max}).$ It is important
            that the constants appearing on the right-hand sides of \eqref{E:UniformtimeApriori1} - \eqref{E:UniformtimeApriori7}
            do not depend on $c,$ if $c$ is large enough. We can therefore apply Proposition 
            \ref{P:ContinuationPrinciple} to conclude that for all large $c,$ the solution can be extended to a
            uniform interval $[0,T].$ The closing of the induction argument is largely due to the fact that the source
            term for the Klein-Gordon equation satisfied by $\Phi,$ which is the right-hand side of \eqref{E:ENkappac4},
            ``depends on $\Phi$ only through $c^{-2} \Phi.$''
            
        \subsubsection{Proof of Theorem \ref{T:UniformLocalExistenceENkappac}} \label{SSS:ProofofUniformTime} 
				To begin, we remark that for the remainder of this article, we indicate dependence of the running constants on  
				$\|\WID_c\|_{H_{\Wb_c}^{N}},$ $\|\WIDSmoothed_c\|_{H_{\Wb_c}^{N+1}},$ $\|\PhiID_c\|_{H_{\Phibar_c}^{N+1}},$ and 
				$\|\PsiID_0\|_{H^N}$ by writing $C(id).$ By 
				Remark \ref{R:SmoothedDataLargecNPlusOneEstimate}, any constant $C(id)$ can be chosen
    		to be independent of all large $c.$ 	
        
        We now introduce some notation that will be used throughout the proof, and also in the
        following section, where we have placed the proofs of the technical lemmas. Let $\Vbold$ denote the local in time solution
        to the EN$_{\kappa}^c$ system \eqref{E:ENkappac1} -
        \eqref{E:PDefcII} launched by the initial data $\VID_c$
        as furnished by Theorem \ref{T:LocalExistencec}. With $\Wbold$ denoting the first $5$ components of
        $\Vbold,$ we suggestively define
        \begin{align}
            \Wdot(t,\sforspace) &\eqdef \Wbold(t,\sforspace) - \WIDSmoothed_c(\sforspace) \label{E:DotWcUniformTime}\\
            \Phidot &\eqdef \Phi - \PhiID_c \label{E:PhidotcUniformTime},
        \end{align}
        where $\PhiID_c$ is defined in \eqref{E:PhiIDcDef} and
        $\WIDSmoothed_c(\sforspace)$ is defined in \eqref{E:IterateInitialValuec2} with the help of
        \eqref{E:Zdef}. We remark that this choice of $\WIDSmoothed_c(\sforspace)$ is explained in more detail below.
        
        It follows from the fact that $\Wbold$ is a solution to \eqref{E:ENkappac1} - \eqref{E:ENkappac3}
        that $\Wdot$ is a solution to the EOV$_{\kappa}^c$ \eqref{E:EOVc1} - \eqref{E:EOVc3}
        defined by the BGS $\Vbold$ with initial data \\
        $\Wdot(0,\sforspace) = \WID_c(\sforspace) -
        \WIDSmoothed_c(\sforspace).$ The inhomogeneous terms in the
        EOV$_{\kappa}^c$ satisfied by $\Wdot$ are given by $\mathbf{b}=(f,g,\cdots,h^{(3)}),$ where for $j=1,2,3$
            \begin{align}
                f &= - \vsubc^k \partial_k  [{^{(0)}\mathring{\Ent}}] \label{E:IterateInhomogeneouscf}\\
                g &=  (4P-3\Qc)[\partial_t (c^{-2}\Phi) + \vsubc^k \partial_k (c^{-2}\Phi)] - \vsubc^k
                    \partial_k  [e^{4\PhiID_c/c^2}\cdot\pIDSmoothed]
                    \label{E:IterateInhomogeneouscg}\\
                & \qquad - \Qc \partial_k  [{^{(0)}\mathring{v}^k}] - c^{-2}(\gammac)^2\Qc\vsubc_k \vsubc^a \partial_a 
                    [{^{(0)}\mathring{v}^k}] \notag \\
                h^{(j)} & =  \big(3c^{-2}P - \Rc\big)\big(\partial_j \Phi +
                    (\gammac)^{-2}\vsubc^j[\partial_t
                    (c^{-2}\Phi) + \vsubc^k \partial_k (c^{-2} \Phi)]\big) \label{E:IterateInhomogeneouschj}\\
               	& \qquad -(\gammac)^2(\Rc + c^{-2}P)\big(\vsubc^k \partial_k
                    [{^{(0)}\mathring{v}^j}] + c^{-2}(\gammac)^2 \vsubc^j \vsubc_k \vsubc^a \partial_a [{^{(0)}\mathring{v}^k}]\big)  
                    \notag \\
               	& \qquad - \partial_j [e^{4\PhiID_c/c^2}\cdot\pIDSmoothed]
                    - c^{-2}(\gammac)^2 \vsubc^j \vsubc^k \partial_k [e^{4\PhiID_c/c^2}\cdot\pIDSmoothed].    \notag
            \end{align}

        In order to show that the hypotheses of Proposition \ref{P:ContinuationPrinciple} are satisfied, we
        will need to estimate $\partial_{\vec{\alpha}}\Wdot$ in $L^2.$ Therefore, we study
        the equation that $\partial_{\vec{\alpha}}\Wdot$ satisfies: for $0 \leq |{\vec{\alpha}}| \leq N,$
        we differentiate the EOV$_{\kappa}^c$ defined by the BGS $\Vbold$
        with inhomogeneous terms $\mathbf{b}$
        to which $\Wdot$ is a solution, obtaining that $\partial_{\vec{\alpha}} \Wdot$
        satisfies
        \begin{align} \label{E:EOVcUniformTime}
            \Ac^{\mu}(\scrW,\Phi) \partial_{\mu} \big(\partial_{\vec{\alpha}}\Wdot \big) &=
            \mathbf{b}_{\vec{\alpha}},
        \end{align}
        where (suppressing the dependence of the $\Ac^{\nu}(\cdot)$ on $\scrW$ and $\Phi$)
        \begin{align}
            \mathbf{b}_{\vec{\alpha}} \eqdef \Ac^0 \partial_{\vec{\alpha}} \left((\Ac^0)^{-1}\mathbf{b} \right) +
            \mathbf{k}_{\vec{\alpha}} \label{E:balphaDefc}
        \end{align}
        and 
        \begin{align}
            \mathbf{k}_{\vec{\alpha}} \eqdef \Ac^0 \big[(\Ac^0)^{-1}\Ac^k \partial_k (\partial_{\vec{\alpha}} \Wdot)
            -\partial_{\vec{\alpha}} \big((\Ac^0)^{-1} \Ac^k \partial_k \Wdot \big)\big]  \label{E:kalphaDefc}.
        \end{align}
        Thus, each $\partial_{\vec{\alpha}}\Wdot$ is a solution the EOV$_{\kappa}^c$
        defined by the \emph{same} BGS $\Vbold$ with inhomogeneous terms
        $\mathbf{b}_{\vec{\alpha}}.$ Furthermore, $\Phidot$ is a solution to the EOV$_{\kappa}^c$
        equation \eqref{E:EOVcKlein-Gordon} with $\Phidot(0,\sforspace) = 0,$ and the inhomogeneous
        term $l$ on the right-hand side of \eqref{E:EOVcKlein-Gordon} is 
        \begin{align} \label{E:lUniformTimeDef}
            l \eqdef (\kappa^2 - \Delta) \PhiID_c + 4 \pi G (\Rc - 3c^{-2}P).
        \end{align}
        We will return to these facts in Section \ref{SS:TechnicalLemmas}, where we will use them in the proofs of some technical
        lemmas.

        As an intermediate step in our proof of \eqref{E:UniformTime1a} - \eqref{E:UniformTime1f}, 
        we will prove the following weaker version of \eqref{E:UniformTime1d}:
        \begin{align}
            c^{-1} \mid\mid\mid \partial_t \Phi \mid\mid\mid_{H^{N},T} \leqc L_2'
            \tag{\ref{E:UniformTime1d}'}.
        \end{align}
        We now define the constants $\Lambda_1, \Lambda_2,$
        $L_2',$ and $L_4.$ We will then use a variety of energy estimates to
        define $L_1, L_2,$ and $L_3$ in terms of these four constants and to show that
        \eqref{E:UniformTime1a} - \eqref{E:UniformTime1f} are
        satisfied if $T$ is small enough. First, to motivate our definitions of $L_2', L_4,$ and $\Lambda_2,$ see inequalities 
        \eqref{E:Klein-GordonPartialtPhiHNNormBound2}
        and \eqref{E:Klein-GordonPartialSquaredtPhiHNMinusOneNormBound2} of
        Proposition \ref{P:Klein-GordonPhiBounds} and inequality \eqref{E:Klein-GordonPhiNocBound}
        of Corollary \ref{C:Klein-GordonPhiBounds}, and
        let $C_0(\kappa)$ denote the constant that appears throughout the lemma and its corollary.
        By a non-optimal application of Lemma \ref{L:Klein-GordonIntitialEnergyEstimates}, we have that
    \begin{align}
         C_0(\kappa)\big(c^{-1}\|\PsiID_0\|_{H^{N}} + \|l(0)\|_{H^{N-1}}\big)
            \leqc 1/2 \eqdef L_2'/2\\
        C_0(\kappa)\big(c\|l(0)\|_{H^{N-1}} + \|(\Delta - \kappa^2)\PsiID_0 - \partial_t l(0)\|_{H^{N-2}} \big) \leqc 1 \eqdef L_4.
    \end{align}
    Note also the trivial (and not optimal) estimate \\ 
    $(C_0(\kappa))^2c^{-2}\|\PsiID_0\|_{H^N}^2 \leqc 1/4 
    \eqdef (\Lambda_2)^2/4.$ With these considerations in mind, we have thus defined
    \begin{align}
        \Lambda_2 &\eqdef  1 \label{E:Lambda2Def}\\
        L_2' &\eqdef 1 \label{E:L2PrimeDef} \\
        L_4 &\eqdef 1. \label{E:Lambda4Def}
    \end{align}

   To define $\Lambda_1,$ we first define $Z=Z(id;\Lambda_2)$ to be the constant appearing in \eqref{E:Zdef}. Using this value of $Z,$
   which we emphasize depends only on $\Lambda_2$ and the initial data $\scrWID_{\infty}$ for the EP$_{\kappa}$ system,
   we then choose $\Lambda_1$ so that \eqref{E:WellDefinedc} and \eqref{E:IntialValueEstimatec} hold. Note that it is exactly
   at this step in the proof that the smoothing $\WIDSmoothed_c,$ which is defined in \eqref{E:IterateInitialValuec2}, of the initial 
   data $\WID_c,$ which are the first $5$ components of \eqref{E:ENkappacVID}, is fixed.

   We find it illuminating to display the dependence of other constants that will appear below on $\Lambda_1,\Lambda_2,L_2',L_4.$ 
   Therefore, we continue to refer to \eqref{E:Lambda2Def} - \eqref{E:Lambda4Def} by the symbols $\Lambda_2, L_2',$ and $L_4$ 
   respectively, even though they are equal to $1.$

    We now carry out the continuous induction in detail. Let $T^{max}_c$ be the maximal time for which the solution
    $\Vbold$ exists and satisfies the estimates \eqref{E:UniformTime1a},
    \eqref{E:UniformTime1b}, (\ref{E:UniformTime1d}'), and
    \eqref{E:UniformTime1f}; i.e.,
    \begin{align} \label{E:TmaxcDef}
        T^{max}_c \eqdef \sup \Big\lbrace T \ |\  \Vbold & \in \bigcap_{k=0}^{k=2}C^k([0,T],H_{\Vb_c}^{N-k}), \\ 
            & \mbox{and} \ \eqref{E:UniformTime1a}, \eqref{E:UniformTime1b}, \    
            \mbox{(\ref{E:UniformTime1d}')}, \mbox{and} \ \eqref{E:UniformTime1f} \ \mbox{hold} \Big\rbrace.  \notag 
    \end{align}
    Note that the set we are taking the $\sup$ of necessarily contains positive values of $T$ since for all large $c,$ 
    the relevant bounds are satisfied at $T=0,$ and therefore by Theorem \ref{T:LocalExistencec}, also for short times.
    Lemmas \ref{L:UniformtimeApriori1}, \ref{L:UniformtimeApriori2},
    \ref{L:PartialtandSquaredEntandPSobolevBounds}, \ref{L:UniformtimeAprioriPhidot},
    \ref{L:UniformtimeAprioriPartialtPhidot}, and inequalities
    \eqref{E:PartialtSquaredPhiImprovedBound} and \eqref{E:PartialtPhiImprovedBound} of Lemma
    \ref{L:UniformTimeImprovedBound} supply the following estimates which are valid for $0 \leq \tau < T^{max}_c:$
    \begin{align}
        &\mid\mid\mid \Wdot \mid\mid\mid_{H^N,\tau} \leqc \big[\Lambda_1/2 + \tau \cdot C(\Lambda_1,\Lambda_2,L_1,L_2')\big] \cdot
                \mbox{\textnormal{exp}} \big(\tau \cdot C(\Lambda_1,\Lambda_2,L_1,L_2') \big) \label{E:UniformtimeApriori1} \\
        &\mid\mid\mid\partial_t \Wbold \mid\mid\mid_{H^N,\tau} \leqc L_1(\Lambda_1,\Lambda_2,L_2')\label{E:UniformtimeApriori2}\\
        &\mid\mid\mid \partial_t^2 \Ent \mid\mid\mid_{H^{N-2},\tau}, \mid\mid\mid \partial_t^2 p \mid\mid\mid_{H^{N-2},\tau}
            \leqc L_3(\Lambda_1,\Lambda_2,L_1,L_2',L_4)\label{E:UniformtimeApriori3} \\
        &\mid\mid\mid \Phidot \mid\mid\mid^2_{H^{N+1},\tau} \leqc \frac{(\Lambda_2)^2}{4} +\tau \cdot C(\Lambda_1,\Lambda_2,L_2) + 
        \tau^2 \cdot C(\Lambda_1,\Lambda_2,L_1,L_2',L_3,L_4) \label{E:UniformtimeApriori4} \\
        &c^{-1} \mid\mid\mid \partial_t \Phi \mid\mid\mid_{H^{N},\tau} \leqc L_2'/2 + \tau \cdot 
        C(\Lambda_1,\Lambda_2,L_1,L_2')\label{E:UniformtimeApriori5}  \\
        &c^{-1} \mid\mid\mid \partial_t^2 \Phi \mid\mid\mid_{H^{N-1},\tau} \leqc L_4/2 + \tau \cdot
        C(\Lambda_1,\Lambda_2,L_1,L_2',L_3,L_4) \label{E:UniformtimeApriori6} \\
        & \mid\mid\mid \partial_t \Phi \mid\mid\mid_{H^{N},\tau} \leqc
            L_2(\Lambda_1,\Lambda_2,L_1,L_2')/2 + \tau \cdot C(\Lambda_1,\Lambda_2,L_1,L_2',L_3,L_4)
            \label{E:UniformtimeApriori7}.
    \end{align}

    We apply the following sequence of reasoning to interpret the above inequalities:
    first $L_1$ in \eqref{E:UniformtimeApriori2} is determined through the known constants
    $\Lambda_1,\Lambda_2,$ and $L_2'.$ Then $L_3$ in \eqref{E:UniformtimeApriori3} is
    determined through the known constants $\Lambda_1,\Lambda_2,L_1,L_2',$ and
    $L_4.$ Then $L_2$ in \eqref{E:UniformtimeApriori7} is determined through $\Lambda_1,\Lambda_2,L_1,$ and $L_2'.$ 
    Finally, the remaining constants $C(\cdots)$ in
    \eqref{E:UniformtimeApriori1} - \eqref{E:UniformtimeApriori6} are all
    determined through $\Lambda_1,\Lambda_2,L_1,L_2',L_3,L_4.$
    
    By Sobolev embedding and \eqref{E:PhibarcConvergence}, there exists a cube $[-a,a]^5$ (depending on the initial data, $\Lambda_1,$
    and $L_2$) such that for all large $c,$ the assumptions $\mid\mid \Phidot \mid\mid_{H^{N}} \leq \Lambda_1$ and  
    $\mid\mid \partial_t \Phi \mid\mid_{H^{N}} \leq L_2$ together imply that
    \begin{align} 
    	\Big(\Phidot,\partial_1 \Phidot,\partial_2\Phidot,\partial_3 \Phidot,\partial_t \Phi\Big)([0,T] \times \mathbb{R}^3) 	
    	\subset [-a,a]^5.
  	\end{align}
    \noindent Motivated by these considerations, we define both for use now and use later in the article the following compact sets:
    \begin{align}	
    	K \eqdef \bar{\mathcal{O}}_2 \times [-a,a]^5 \label{E:Kdef} \\
    	\mathfrak{K} \eqdef \bar{\mathfrak{O}}_2 \times [-a,a]^5. \label{E:frakKdef}
    \end{align}
    \noindent Here, $\mathcal{O}_2$ and $\mathfrak{O}_2$ are the sets defined in Section \ref{SS:AdmissibleStateSpacec}.

    We now choose $T$ so that when $0 \leq \tau \leq T,$ it
    algebraically follows that the right-hand sides of \eqref{E:UniformtimeApriori1} 
    and \eqref{E:UniformtimeApriori4} - \eqref{E:UniformtimeApriori7} are \emph{strictly} less than 
    $\Lambda_1, (\Lambda_2)^2, L_2', L_4,$ and $L_2$ respectively. Note that $T$ may be
    chosen independently of (all large) $c.$ We now show that $T_c^{max} \leq T$ is impossible. 
    
    Assume that $T_c^{max} \leq T.$ Then observe that the right-hand sides of \eqref{E:UniformtimeApriori1}
    and \eqref{E:UniformtimeApriori4} - \eqref{E:UniformtimeApriori7} are \emph{strictly} less than 
    $\Lambda_1, (\Lambda_2)^2, L_2', L_4,$ and $L_2$ respectively when $\tau = T_c^{max}.$
    Therefore, by the construction of the set $K$ described above,
    by \eqref{E:WellDefinedc}, and by Sobolev embedding, we conclude that for all large $c,$ $\Vbold([0,T^{max}_c) \times
    \mathbb{R}^3)$ is contained in the \emph{interior} of $K.$ Consequently, we may apply
    Proposition \ref{P:ContinuationPrinciple} to extend the solution in time beyond $T_c^{max},$ thus contradicting
    the definition of $T_c^{max}.$ Note that this argument also shows that 
    $\Vbold([0,T] \times \mathbb{R}^3) \subset K.$ This completes the proof of Theorem \ref{T:UniformLocalExistenceENkappac}. 
    \hspace{5.3in} $\square$

    \subsection{The technical lemmas} \label{SS:TechnicalLemmas} 
     We now state and prove the technical lemmas quoted in the proof
    of Theorem \ref{T:UniformLocalExistenceENkappac}. We will require some auxiliary lemmas
    along the way. Throughout this section, we assume the hypotheses of Theorem
    \ref{T:UniformLocalExistenceENkappac} and we use the notation from Section \ref{SSS:ProofofUniformTime}; i.e., $\Vbold$ denotes
    the solution, $\Wbold$ denotes its first 5 components, the relationship between
    $\Wbold$ and $\scrW$ is given by \eqref{E:Wboldarray} and \eqref{E:scrWboldarray}, $\Wdot$ and $\Phidot$ are 
    defined in \eqref{E:DotWcUniformTime} and \eqref{E:PhidotcUniformTime}
    respectively, $l$ is defined in \eqref{E:lUniformTimeDef}, and
    so forth. All of the estimates in this section hold on the time interval $\tau \in [0,T_c^{max}),$ where $T_c^{max}$ is defined 
    in \eqref{E:TmaxcDef}.

    \subsubsection{The induction hypotheses} \label{SS:InductionHypotheses} 
   	By the definition of $T_c^{max},$ we have the 
    following bounds, where $\Lambda_2, L_2',$ and $L_4$ are defined in 	
    \eqref{E:Lambda2Def} - \eqref{E:Lambda4Def} respectively, and we will soon elaborate on our choice $\Lambda_1$:
    \begin{align}
        \mid\mid\mid \Wbold - \WIDSmoothed_c \mid\mid\mid_{H^N,\tau} \leqc \Lambda_1 \label{E:UniformTime1aAgain}\\
        \mid\mid\mid \Phi - \PhiID_c \mid\mid\mid_{H^{N+1},\tau} \leqc \Lambda_2 \label{E:UniformTime1bAgain}\\
        c^{-1}\mid\mid\mid \partial_t \Phi \mid\mid\mid_{H^{N},\tau} \leqc
        L_2' \label{E:UniformTime1dPrimeAgain} \\
        c^{-1} \mid\mid\mid \partial_t^2 \Phi \mid\mid\mid_{H^{N-2},\tau} \leqc L_4
        \label{E:UniformTime1fPrimeAgain}.
  	\end{align}
  	We note the following easy consequence of \eqref{E:PhiIDcDef} and \eqref{E:UniformTime1bAgain}:  
    \begin{align}       \tag{\ref{E:UniformTime1bAgain}'}
    		\mid\mid\mid \Phi & - \Phibar_c \mid\mid\mid_{H^N,\tau}  \\ 
    			& \leq \mid\mid\mid \Phi - \PhiID_c \mid\mid\mid_{H^N,\tau}
    			+ \mid\mid\mid \PhiID_c - \Phibar_c \mid\mid\mid_{H^N,\tau} \leqc \Lambda_2 + C(id) \eqdef 
    			C(id;\Lambda_2). \notag
    \end{align}
    It then follows from \eqref{E:PhibarcConvergence}, (\ref{E:UniformTime1bAgain}'), and Sobolev embedding that
    \begin{align}                  \label{E:Zdef}
    	\mid\mid\mid \Phi \mid\mid\mid_{L^{\infty},\tau} \leqc Z(id;\Lambda_2). 
    \end{align}
    
    Let us recall how $\Lambda_1$ was chosen: using the value of $Z$ in \eqref{E:Zdef}, which depends only on the data 
    $\scrWID_{\infty}$ for the EP$_{\kappa}$ system and the known constant $\Lambda_2,$ we have chosen a constant $\Lambda_1 > 0$ 
    such that \eqref{E:WellDefinedc} and 
    \eqref{E:IntialValueEstimatec} hold. As discussed in sections \ref{S:SmoothingtheData} and \ref{SSS:ProofofUniformTime}, such a 
    choice of $\Lambda_1$ also involves fixing the smoothing ${^{(0)}\scrWID}$ of $\scrWID_{\infty},$
    which then defines ${^{(0)}\WID_c}$ via equation \eqref{E:IterateInitialValuec2}. We emphasize that it is this choice of 
    ${^{(0)}\WID_c}$ and $\Lambda_1$ that appear in \eqref{E:UniformtimeApriori1}, \eqref{E:UniformTime1a}, and  
    \eqref{E:UniformTime1aAgain}.
    
    By \eqref{E:IntialValueEstimatec} and \eqref{E:UniformTime1aAgain}, we also have that
    \begin{align} 				\tag{\ref{E:UniformTime1aAgain}'}
    	\mid\mid\mid \Wbold & - \Wb_c \mid\mid\mid_{H^N,\tau} \\
    		&\leq \ \mid\mid\mid \Wbold - \WIDSmoothed_c \mid\mid\mid_{H^N,\tau} \ + \
    			\mid\mid\mid \WIDSmoothed_c - \WID_c \mid\mid\mid_{H^N,\tau} 
    			\ + \ \mid\mid\mid \WID_c- \bar{\Wbold}_c \mid\mid\mid_{H^N,\tau} \notag \\
    			&\leqc \ \Lambda_1 + C(id;\Lambda_1) \ + \ C(id) \ \eqdef \ C(id;\Lambda_1). \notag 
    \end{align}
    Furthermore, by Lemma \ref{L:HjbySubtractingConstant}, \eqref{E:ScrWInTermsofWEstimate} with $m=0,$ and
    (\ref{E:UniformTime1aAgain}'), we have that
    \begin{align} \label{E:scrWHNBound}
    	\mid\mid\mid \scrW - \scrWb_c \mid\mid\mid_{H^N,\tau} \leqc C(id;\Lambda_1,\Lambda_2).
    \end{align}

    We also observe that \eqref{E:WellDefinedc}, \eqref{E:UniformTime1aAgain}, and 
    the definition of $\mathfrak{O}_2$ given in Section \ref{SS:AdmissibleStateSpacec} together imply that
  	for all large $c,$ we have that $\Wbold([0,T_c^{max}) \times \mathbb{R}^3) \subset \bar{\mathcal{O}}_2$
  	and $\scrW([0,T_c^{max}) \times \mathbb{R}^3) \subset \bar{\mathfrak{O}}_2.$

    In our discussion below, we will refer to \eqref{E:UniformTime1aAgain} - \eqref{E:scrWHNBound} 
    (\ref{E:UniformTime1aAgain}'), and (\ref{E:UniformTime1bAgain}') as the \emph{induction hypotheses}. Sobolev embedding and
    the induction hypotheses, which for all large $c$ are satisfied at $\tau=0,$ together imply that 
    $\Wbold,$ $\partial \Wbold,$ $\scrW,$ $\partial \scrW,$ $\Phi,$ $\partial \Phi,$ $c^{-1} \partial_t \Phi,$ and
    $c^{-1} \partial_t^2 \Phi$ are each contained in a compact convex set (depending only on the initial data, $\Lambda_1,$ $\Lambda_2,$
    $L_2',$ and $L_4$) on $[0,T_c^{max}) \times \mathbb{R}^3.$ As stated in Remark \ref{R:ArgumentsOmitted}, we will make use of this fact 
    without explicitly mentioning it every time.

    \subsubsection{Proofs of the technical lemmas}
    \begin{lemma} \label{L:lFirstEstimates}
    	Consider the quantity $l$ defined in \eqref{E:lUniformTimeDef}. Then for $m=0,1,2,$ we have that
        \begin{align}
            (4 \pi G)^{-1} l &= \mathfrak{R}_{\infty}(\Ent,p) - \mathfrak{R}_{\infty}(\EntID,\pID) + \mathfrak{F}_c
                \label{E:lInitialEstimate} \\
            (4 \pi G)^{-1}\partial_t l &= \partial_t \big(\mathfrak{R}_{\infty}(\Ent,p) \big)
                + \mathfrak{G}_c \label{E:PartialtlInitialEstimateA} \\
            (4 \pi G)^{-1}\partial_t^2 l &= \partial_t^2
                \big(\mathfrak{R}_{\infty}(\Ent,p) \big)
                + \mathfrak{H}_c, \label{E:PartialtSquaredlInitialEstimate}
        \end{align}
        where 
        \begin{align}
        \mathfrak{F}_c &\in \mathcal{I}^N(c^{m-2};\Ent,p,c^{-m}\Phi) \label{E:lInitialFcEstimate}\\
            \mathfrak{G}_c &\in \mathcal{I}^{N-1}(c^{m-2};\Ent,p,c^{-m}\Phi,\partial_t \Ent,\partial_t p,c^{-m}\partial_t \Phi)
            \label{E:PartialtlInitialGcEstimate} \\
        \mathfrak{H}_c &\in \mathcal{I}^{N-2}(c^{m-2};\Ent,p,c^{-m}\Phi,\partial_t \Ent,\partial_t
            p,c^{-m}\partial_t \Phi,\partial_t^2 \Ent,\partial_t^2p,c^{-m}\partial_t^2 \Phi). 
            \label{E:PartialtSquaredlInitialHcEstimate}
        \end{align}
        \end{lemma}

        \begin{proof}
            It follows from the discussion in Section \ref{S:IVPc} that
            \begin{align} \label{E:lExpression}
                (4 \pi G)^{-1}l = \big(e^{4\Phi/c^2}\mathfrak{R}_c(\Ent,p) & - e^{4\Phibar_c/c^2}\mathfrak{R}_c(\Entbar,\pbar)\big) \\
                    & + 3 c^{-2} \big(e^{4\Phibar_c/c^2}\pbar - e^{4\Phi/c^2}p \big) 
                    + \mathfrak{R}_{\infty}(\Entbar,\pbar) - \mathfrak{R}_{\infty}(\EntID,\pID). \notag
            \end{align}
            Therefore, \eqref{E:lInitialEstimate} + \eqref{E:lInitialFcEstimate} follows from Lemma
            \ref{L:HjbySubtractingConstant}, Lemma \ref{L:MultiplyFcGcSobolevEstimate}, and Lemma \ref{L:OrderctoNegativeTwoEstimates}. 
            \eqref{E:PartialtlInitialEstimateA} + \eqref{E:PartialtlInitialGcEstimate}
            and \eqref{E:PartialtSquaredlInitialEstimate} + \eqref{E:PartialtSquaredlInitialHcEstimate} then follow from Lemma 
            \ref{L:TimeDifferentiatedSobolevEstimate}.
        \end{proof}

        \begin{lemma}       \label{L:UniformtimeApriori2}
            \begin{align}
                \mid\mid\mid\partial_t \Wbold \mid\mid\mid_{H^{N-1},\tau}, \ \mid\mid\mid\partial_t \scrW
                \mid\mid\mid_{H^{N-1},\tau} \leqc 
                C(id;\Lambda_1,\Lambda_2,L_2')\eqdef
                L_1(id;\Lambda_1,\Lambda_2,L_2'). \label{E:UniformtimeApriori2Again}
            \end{align}
        \end{lemma}
            \begin{proof}
            By using the EN$_{\kappa}^c$ equations \eqref{E:ENkappac1} -
            \eqref{E:ENkappac3} to solve for $\partial_t \Wbold$ and applying Lemma \ref{L:MultiplyFcGcSobolevEstimate},
            \eqref{E:SpatialDerivativesWscrWOrderctoNegativeTwoEstimate} in the cases $\nu =1,2,3,$ Lemma 
            \ref{L:ScrWInTermsofWEstimate}, Lemma \ref{L:MatricesLargecNormEstimates}, Lemma 
            \ref{L:InhomogeneousTermsLargecSobolevEstimates}, and Remark \ref{R:WorstTerm}, we have that
                \begin{align}  \label{E:SolveforTimeDerivatives}
                    \partial_t \Wbold & = \big(\Ac^0(\scrW,\Phi)\big)^{-1}\big[-\Ac^k(\scrW,\Phi) \partial_k \Wbold
                        + \mathfrak{B}_c(\scrW,\Phi,D\Phi)\big] \\
                    &= \big(\Ainfinity^0(\scrW)\big)^{-1}\big[-\Ainfinity^k(\scrW) \partial_k \scrW  + 
                    \mathfrak{B}_{\infty}(\scrW,\partial\Phi) \big] \notag \\
                    & \qquad + \mathscr{O}^{N-1}(\Wbold,\partial\Wbold,c^{-1}\Phi,c^{-1}D\Phi)
                        	\cap \mathscr{O}^{N-1}(c^{-2};\Wbold,\partial\Wbold,\Phi,D\Phi). \notag
                \end{align}

            The bound for $\mid\mid\mid\partial_t \Wbold \mid\mid\mid_{H^{N-1},\tau}$ now follows from Lemma 
            \ref{L:MultiplyFcGcSobolevEstimate}, \eqref{E:MatricescIsInfinityNormEstimates}, \eqref{E:BInfinityIsInIN},
            the induction hypotheses, \eqref{E:SolveforTimeDerivatives}, and the definition \eqref{E:ONotation} of \\ 	
            $\mathscr{O}^{N-1}(\Wbold,\partial\Wbold,c^{-1}\Phi,c^{-1}D\Phi).$ The bound for $\mid\mid\mid\partial_t \scrW 
            \mid\mid\mid_{H^{N-1},\tau}$ then follows from the bound for $\mid\mid\mid\partial_t \Wbold \mid\mid\mid_{H^{N-1},\tau},$ 
            \eqref{E:SpatialDerivativesWscrWOrderctoNegativeTwoEstimate} in the case $\nu =t, m=1,$ and
            the induction hypotheses. We remark that we have written the ``intersection term'' on the right-hand side of 
            \eqref{E:SolveforTimeDerivatives} in a form that will be useful in our proofs of 
            Lemma \ref{L:Klein-GordonIntitialEnergyEstimates}, and 
            Lemma \ref{L:PartialtandSquaredEntandPSobolevBounds}; the $``c^{-2}$ decay'' is used in Lemma 
            \ref{L:Klein-GordonIntitialEnergyEstimates} and Corollary \ref{C:ApproximateEPkappaSolutions}, while the ``dependence on 
            $c^{-1} D \Phi$'' is used in Lemma \ref{L:PartialtandSquaredEntandPSobolevBounds}. Similar comments apply to Corollary 
            \ref{C:ENkappacIsAlmostEPkappa} and equation \eqref{E:PartialtpExpression} below. \\
           	\end{proof}

        The following indispensable corollary shows that for large $c,$ the EN$_{\kappa}^c$ system can be written as
        a small perturbation of the EP$_{\kappa}$ system. See also Corollary \ref{C:ApproximateEPkappaSolutions}.
        
        \begin{corollary} \label{C:ENkappacIsAlmostEPkappa} {\bf (EN$_{\kappa}^c$ $\approx$ EP$_{\kappa}$ for Large $c$)}
            \begin{align} \label{E:SolveforTimeDerivativesscrW}
                \partial_t \scrW & = \big(\Ainfinity^0(\scrW)\big)^{-1}\big[-\Ainfinity^k(\scrW) \partial_k \scrW  +
                	\mathfrak{B}_{\infty}(\scrW,\partial\Phi) \big]  \\
                & \qquad + \mathscr{O}^{N-1}(\Wbold,\partial\Wbold,c^{-1}\Phi,c^{-1}D\Phi)
               		\cap \mathscr{O}^{N-1}(c^{-2};\Wbold,\partial\Wbold,\Phi,D\Phi). \notag 
             \end{align}
        \end{corollary}
            \begin{proof}
              	Recall that $\partial_t \Wbold$ and $\partial_t \scrW$ differ only in
                that the second component of $\partial_t \Wbold$ is $\partial_t P,$
                while the second component of $\partial_t \scrW$ is $\partial_t p.$ Therefore, it follows 
                trivially from \eqref{E:SolveforTimeDerivatives} that 
                \eqref{E:SolveforTimeDerivativesscrW} holds for all the 
                components of $\partial_t \scrW$ except for the second component $\partial_t p.$
                
                To handle the component $\partial_t p,$ we first observe that the second component of the array 
                $-\big(\Ainfinity^0(\scrW)\big)^{-1}\big[-\Ainfinity^k(\scrW) \partial_k \scrW + 
                \mathfrak{B}_{\infty}(\scrW,\partial\Phi) \big]$ is equal to \\
                $- v^k \partial_k p - \mathfrak{Q}_{\infty}(\Ent,p)\partial_k v^k.$ It thus follows directly from
               	considering the second component of \eqref{E:SolveforTimeDerivatives} that
                \begin{align} \label{E:PartialtPExpression}
                	\partial_t P &= - v^k \partial_k p - \mathfrak{Q}_{\infty}(\Ent,p)\partial_k v^k \\
                  & \ \ + \ \mathscr{O}^{N-1}(\Wbold,\partial\Wbold,c^{-1}\Phi,c^{-1}D\Phi)
                        	\cap \mathscr{O}^{N-1}(c^{-2};\Wbold,\partial\Wbold,\Phi,D\Phi). \notag
                \end{align}
                Therefore, since $\partial_t p - \partial_t P= (e^{-4 \Phi / c^2} - 1) \partial_t P - 4 (c^{-2} \partial_t \Phi) 
                e^{-4\Phi/c^2} P,$ we 
                use Lemma \ref{L:MultiplyFcGcSobolevEstimate}, \eqref{E:expPhiOvercSquaredOrderctoNegativeTwoEstimate}, 	
                \eqref{E:SpatialDerivativesWscrWOrderctoNegativeTwoEstimate}, \eqref{E:DerivativesWscrWOrderctoZeroEstimate}, Lemma 
                \ref{L:ScrWInTermsofWEstimate}, and \eqref{E:PartialtPExpression} to conclude that 
                \begin{align} \label{E:PartialtpExpression}
                	\partial_t p &= - v^k \partial_k p - \mathfrak{Q}_{\infty}(\Ent,p)\partial_k v^k \\
                  & \ \ + \ \mathscr{O}^{N-1}(\Wbold,\partial\Wbold,c^{-1}\Phi,c^{-1}D\Phi)
                 		\cap \mathscr{O}^{N-1}(c^{-2};\Wbold,\partial\Wbold,\Phi,D\Phi). \notag 
                \end{align}
      		\end{proof}

    \begin{lemma}                                                   \label{L:Klein-GordonIntitialEnergyEstimates}
        There exists a constant $C(id) > 0$ such that
        \begin{align}
            \|l(0)\|_{H^N} & \leqc c^{-2}C(id) \label{E:lofZerobound} \\
            \|(\Delta - \kappa^2) \PsiID_0 - \partial_t l(0)\|_{H^{N-1}} &\leqc c^{-2} C(id) .
            \label{E:PartialtlofZerobound}
        \end{align}
    \end{lemma}
          \begin{proof}
            The estimate \eqref{E:lofZerobound} follows from the estimate \eqref{E:lInitialEstimate}
            for $l(t)$ at $t=0$ and \eqref{E:lInitialFcEstimate} in the case $m=0.$

        To obtain the estimate \eqref{E:PartialtlofZerobound}, first recall that by
        assumption \eqref{E:Psi0Identity} and the chain rule, we have that
        \begin{align} \label{E:PsiIntialDef}
            (4 \pi G)^{-1}(\kappa^2 - \Delta) \PsiID_0 &=
            \partial_k\big(\mathfrak{R}_{\infty}(\EntID,\pID) \vID^k\big)        \\
            &= \frac{\partial \mathfrak{R}_{\infty}}{\partial \Ent}(\EntID,\pID) \vID^k \partial_k \EntID
                + \frac{\partial \mathfrak{R}_{\infty}}{\partial p}(\EntID,\pID) \vID^k \partial_k \pID
                + \mathfrak{R}_{\infty}(\EntID,\pID) \partial_k \vID^k. \notag
        \end{align}
        Furthermore, by Lemma \ref{L:MultiplyFcGcSobolevEstimate}, \eqref{E:PartialtlInitialEstimateA} at $t=0,$ 
        \eqref{E:PartialtlInitialGcEstimate} in the case $m=0,$ the chain rule,
        \eqref{E:ENkappac1}, \eqref{E:PartialtpExpression}, and 
        \eqref{E:partialRpartialpandStoNegativeTwoRelationshipc} + \eqref{E:QcFunctionDef} in the case $c=\infty,$ we have that
        \begin{align} \label{E:PartialtlInitialEstimateC}
            (4 \pi G)^{-1} \partial_t l(0) =
            - \frac{\partial \mathfrak{R}_{\infty}}{\partial \Ent}(\EntID,\pID) \vID^k 
            \partial_k \EntID  &- \frac{\partial \mathfrak{R}_{\infty}}{\partial p}(\EntID,\pID) \vID^k \partial_k \pID \\
            & - \mathfrak{R}_{\infty}(\EntID,\pID) \partial_k \vID^k + \mathscr{O}^{N-1}(c^{-2};id). \notag
        \end{align}
        The estimate \eqref{E:PartialtlofZerobound} now follows from \eqref{E:PsiIntialDef} and
        \eqref{E:PartialtlInitialEstimateC}.

        \end{proof}

    \begin{lemma} \label{L:PartialtandSquaredEntandPSobolevBounds}
        \begin{align}
            \mid\mid\mid\partial_t^2 \Ent\mid\mid\mid_{H^{N-2},\tau}, \ \mid\mid\mid\partial_t^2 p\mid\mid\mid_{H^{N-2},\tau}
                    &\leqc C(id;\Lambda_1,\Lambda_2,L_1,L_2',L_4) \label{E:PartialtSquaredpandEtabound} \\
                    & \eqdef L_3(id;\Lambda_1,\Lambda_2,L_1,L_2',L_4). \notag
      	\end{align}
    \end{lemma}
            \begin{proof}
                To obtain the bound for $\partial_t^2 p,$ differentiate
                each side of the expression \eqref{E:PartialtpExpression} with respect to $t,$ and
                then apply Lemma \ref{L:TimeDifferentiatedSobolevEstimate} to conclude that
              \begin{align}
                \partial_t^2 p &=  -\partial_t \big[v^k \partial_k p + \mathfrak{Q}_{\infty}(\Ent,p) \partial_k v^k
                \big] + \mathfrak{F}_c, \label{E:PartialtSquaredpIsolated}
            \end{align}
                where $\mathfrak{F}_c \in \mathcal{I}^{N-2}(\Wbold,D \Wbold, \partial\partial_t 
                \Wbold,c^{-1} \Phi, c^{-1} D \Phi, c^{-1} \partial \partial_t \Phi ,c^{-1}\partial_t^2 \Phi)$. 
                We now use Lemma \ref{L:MultiplyFcGcSobolevEstimate}, the induction hypotheses, the previously 
                established bounds \eqref{E:UniformtimeApriori2Again} on $\mid\mid\mid \partial_t \Wbold 
                \mid\mid\mid_{H^{N-1},\tau}$ and $\mid\mid\mid \partial_t \scrW 
                \mid\mid\mid_{H^{N-1},\tau},$ and the definition of $\mathcal{I}^{N-2}(\cdots)$ to conclude the estimate 
                \eqref{E:PartialtSquaredpandEtabound} for $\mid\mid\mid\partial_t^2 p\mid\mid\mid_{H^{N-2},\tau}.$

                The estimate for $\partial_t^2 \eta$ is similar, and in fact much
                simpler: use equation \eqref{E:ENkappac1} to solve for $\partial_t
                \eta,$ and then differentiate with respect to $t$ and reason as above.
            \end{proof}

        \begin{lemma}               \label{L:lSobolevBounds}
            \begin{align}
                \mid\mid\mid l \mid\mid\mid_{H^N,\tau} &\leqc C(id;\Lambda_1,\Lambda_2) \label{E:lHNaPrioriBound} \\
                \mid\mid\mid \partial_t l \mid\mid\mid_{H^{N-1},\tau} &\leqc C(id;\Lambda_1,\Lambda_2,L_1,L_2')
                \label{E:PartialtlHNMinusOneaPrioriBound} \\
                \mid\mid\mid \partial_t^2 l \mid\mid\mid_{H^{N-2},\tau} &\leqc 
                C(id;\Lambda_1,\Lambda_2,L_1,L_2',L_3,L_4) \label{E:PartialtSquaredlHNMinusTwoaPrioriBound}.
            \end{align}
        \end{lemma}
            \begin{proof}
                To prove \eqref{E:lHNaPrioriBound}, we first consider the formula for $l$ given in 
                \eqref{E:lInitialEstimate} + \eqref{E:lInitialFcEstimate}. By Lemma \ref{L:HjbySubtractingConstant} and 
                \eqref{E:scrWHNBound}, we have that 
                \begin{align}
                	\mid\mid\mid \mathfrak{R}_{\infty} (\Ent,p)  - 
                	\mathfrak{R}_{\infty}(\EntID,\pID)\mid\mid\mid_{H^N,\tau} & \leq  \
                		\mid\mid\mid \mathfrak{R}_{\infty}(\Ent,p) - \mathfrak{R}_{\infty}(\Entbar,\pbar) \mid\mid\mid_{H^N,\tau} 
                		\notag \\
                	& \qquad \qquad + \mid\mid\mid \mathfrak{R}_{\infty}(\Entbar,\pbar) - \mathfrak{R}_{\infty}(\EntID,\pID) 
                		\mid\mid\mid_{H^N,\tau}
                		\notag \\
                	& \leqc C(id;\Lambda_1,\Lambda_2). 
               	\end{align}
               	
               	To estimate $\mid\mid\mid \mathfrak{F}_c \mid\mid\mid_{H^N,\tau},$ where $\mathfrak{F}_c$ is 
                from \eqref{E:lInitialEstimate}, simply use \eqref{E:lInitialFcEstimate} in the case $m=0$ together with
                (\ref{E:UniformTime1bAgain}') and \eqref{E:scrWHNBound}. The proofs of \eqref{E:PartialtlHNMinusOneaPrioriBound}
                and \eqref{E:PartialtSquaredlHNMinusTwoaPrioriBound} follow similarly from the expressions
                \eqref{E:PartialtlInitialEstimateA}, \eqref{E:PartialtSquaredlInitialEstimate}, \eqref{E:PartialtlInitialGcEstimate} 
                in the case $m=1,$ and \eqref{E:PartialtSquaredlInitialHcEstimate} in the case $m=1,$ together with 
                Lemma \ref{L:MultiplyFcGcSobolevEstimate} and the bounds supplied by the induction hypotheses, Lemma 
                \ref{L:UniformtimeApriori2}, and Lemma \ref{L:PartialtandSquaredEntandPSobolevBounds}.
        	\end{proof}

        \begin{lemma}                                       \label{L:UniformtimeAprioriPartialtPhidot}
            \begin{align}
                c^{-1} \mid\mid\mid \partial_t \Phi \mid\mid\mid_{H^{N},\tau} &\leqc 
                	1/2 + \tau \cdot C(id;\Lambda_1,\Lambda_2,L_1,L_2') \label{E:PartialtPhiaPrioriHNBound} \\
                &\eqdef L_2'/2 + \tau \cdot C(id;\Lambda_1,\Lambda_2,L_1,L_2'). \notag
            \end{align}
        \end{lemma}
            \begin{proof}
                \eqref{E:PartialtPhiaPrioriHNBound}
                follows from definition \eqref{E:L2PrimeDef}, Lemma \ref{L:Klein-GordonIntitialEnergyEstimates}, inequality 
                \eqref{E:PartialtlHNMinusOneaPrioriBound} of Lemma
                \ref{L:lSobolevBounds}, and inequality \eqref{E:Klein-GordonPartialtPhiHNNormBound2} of Proposition 
                \ref{P:Klein-GordonPhiBounds}.
            \end{proof}

        \begin{lemma}                       \label{L:UniformTimeImprovedBound}
        \begin{align}
            \mid\mid\mid \partial_t \Phi \mid\mid\mid_{H^{N},\tau} &\leqc  C(id;\Lambda_1,\Lambda_2,L_1,L_2') + \tau 
            \cdot C(id;\Lambda_1,\Lambda_2,L_1,L_2',L_3,L_4)  \label{E:PartialtPhiImprovedBound} \\
                &\eqdef L_2(id;\Lambda_1,\Lambda_2,L_1,L_2')/2 + \tau \cdot C(id;\Lambda_1,\Lambda_2,L_1,L_2',L_3,L_4)  
                \notag \\
        		c^{-1} \mid\mid\mid \partial_t^2 \Phi
                \mid\mid\mid_{H^{N-1},\tau} &\leqc 1/2 + \tau \cdot
                C(id;\Lambda_1,\Lambda_2,L_1,L_2',L_3,L_4)  \label{E:PartialtSquaredPhiImprovedBound} \\
                &\eqdef L_4/2 + \tau \cdot C(id;\Lambda_1,\Lambda_2,L_1,L_2',L_3,L_4). \notag
        \end{align}
        \end{lemma}
            \begin{proof}
                The estimate \eqref{E:PartialtPhiImprovedBound} follows
                from Lemma \ref{L:Klein-GordonIntitialEnergyEstimates}, inequalities \eqref{E:PartialtlHNMinusOneaPrioriBound}
                and \eqref{E:PartialtSquaredlHNMinusTwoaPrioriBound} of Lemma \ref{L:lSobolevBounds},
                and inequality \eqref{E:Klein-GordonNocDependencePartialtPhi} of Proposition \ref{P:Klein-GordonNocDependence}.
                The estimate \eqref{E:PartialtSquaredPhiImprovedBound} follows from definition \eqref{E:Lambda4Def},
                Lemma \ref{L:Klein-GordonIntitialEnergyEstimates}, inequality \eqref{E:PartialtSquaredlHNMinusTwoaPrioriBound} of
                Lemma \ref{L:lSobolevBounds}, and inequality \eqref{E:Klein-GordonPartialSquaredtPhiHNMinusOneNormBound2} of
                Proposition \ref{P:Klein-GordonPhiBounds}.
            \end{proof}

        \begin{lemma}                                             \label{L:UniformtimeAprioriPhidot}
            \begin{align} \label{E:PhiHNAprioriEstimateAgain}
                \mid\mid\mid \Phidot \mid\mid\mid^2_{H^{N+1},\tau} \leqc
                \frac{(\Lambda_2)^2}{4} + \tau \cdot 
                C(id;\Lambda_1,\Lambda_2,L_2) + \tau^2 \cdot C(id;\Lambda_1,\Lambda_2,L_1,L_2',L_3,L_4).
            \end{align}
        \end{lemma}  
          
            \begin{proof}
                Inequality \eqref{E:PhiHNAprioriEstimateAgain} follows
                from definition \eqref{E:Lambda2Def}, \eqref{E:lHNaPrioriBound}, \eqref{E:PartialtPhiImprovedBound},
                and inequality \eqref{E:Klein-GordonPhiNocBound} of Corollary \ref{C:Klein-GordonPhiBounds}.
            \end{proof}

        \begin{lemma} \label{L:DivergenceJdotLargecL1Bound}
                Let $\Jscrdotc$ be the energy current \eqref{E:RescaledEnergyCurrent} for
                the variation $\Wdot$ defined by the BGS $\Vbold,$
                and let $\mathbf{b} \eqdef (f,g,\cdots,h^{(3)}),$ where
                $f,g,\cdots,h^{(3)}$ are the inhomogeneous terms from the EOV$_{\kappa}^c$ satisfied by
                $\Wdot$ that are defined in \eqref{E:IterateInhomogeneouscf} - \eqref{E:IterateInhomogeneouschj} 
                and that also appear in the expression \eqref{E:RescaledEnergyCurrentDivergence} for the divergence 
                of $\Jscrdotc.$ Then on $[0,T_c^{max}),$ we have that
                \begin{align}
                    \|\partial_{\mu} \big(\Jscrdotc^{\mu}\big)\|_{L^1} \leqc
                    C(id;\Lambda_1,\Lambda_2,L_1,L_2') \cdot \big[\|\Wdot\|_{L^2}^2 +
                    \|\Wdot\|_{L^2}\|\mathbf{b}\|_{L^2}\big].
                \end{align}
       	\end{lemma}
                \begin{proof}
                    We separate the terms on the right-hand side of \eqref{E:RescaledEnergyCurrentDivergence} into
                    two types: those that depend quadratically on the variations, and those that depend linearly on the variations.
                    We first bound (for all large $c$) the $L^1$ norm of the terms that depend quadratically on the variations by 
                    $C(id;\Lambda_1,\Lambda_2,L_1,L_2')\cdot \|\Wdot\|_{L^2}^2.$ This bound follows directly from the fact
                    that the coefficients of the quadratic variation terms can be bounded in $L^{\infty}$ by
                    $C(id;\Lambda_1,\Lambda_2,L_1,L_2');$ such an $L^{\infty}$ bound
                   	may be obtained by combining Remark \ref{R:1OverFisintheRing}, Lemma \ref{L:OrderctoNegativeTwoEstimates}
                   	in the case $m=1,$ Remark \ref{R:BoundedThroughALimit}, the induction hypotheses, 
                   	\eqref{E:UniformtimeApriori2Again}, and Sobolev embedding. 
                   	
                   	We similarly bound the $L^1$ norm of the terms that depend linearly on the variations
                   	by $C(id;\Lambda_1,\Lambda_2) \cdot \|\Wdot\|_{L^2}\|\mathbf{b}\|_{L^2},$ but for
                   	these terms, we also make use of the Cauchy-Schwarz inequality for integrals.
                \end{proof}
            
            We also state here the following corollary that will
            be used in the proof of Theorem
            \ref{T:NewtonianLimit}. 

            \begin{corollary} \label{C:DivergenceJdotInfinityL1Bound}
                Let $\scrV \in C_b^1([0,T] \times \mathbb{R}^3) \cap \bigcap_{k=0}^{k=1} C^k([0,T],H_{\scrVb_c}^{N-k}),$
								and assume that $\scrV([0,T] \times \mathbb{R}^3)\subset \mathfrak{K},$
								where $\mathfrak{K}$ is defined in \eqref{E:frakKdef}. Let
                $\scrWdot$ be a solution to the EOV$_{\kappa}^{\infty}$
                \eqref{E:EOVc1} - \eqref{E:EOVc3} defined by the BGS $\scrW$ with inhomogeneous terms \\
                $\mathbf{b} =(f,g,\cdots,h^{(3)}),$ where $\scrW$ denotes the first 5
                components of $\scrV.$ Let $\Jscrdotinfinity$ be the energy current
                \eqref{E:RescaledEnergyCurrentInfinity} for
                the variation $\scrWdot$ defined by the BGS $\scrW.$ Then on $[0,T],$ we have that
                \begin{align}
                    \|\partial_{\mu} \big(\Jscrdotinfinity^{\mu}\big)\|_{L^1} \leq
                    C(\mathfrak{K};\mid\mid\mid\scrW\mid\mid\mid_{L^{\infty},T},\mid\mid\mid \partial_t 
                    \scrW\mid\mid\mid_{L^{\infty},T}) \cdot \big[\|\scrWdot\|_{L^2}^2 + \|\scrWdot\|_{L^2}\|\mathbf{b}\|_{L^2}\big].
                \end{align}
            \end{corollary}
            \begin{proof}
            	We do not give any details since Corollary \ref{C:DivergenceJdotInfinityL1Bound} can proved by arguing as 
            	we did in our proof of Lemma \ref{L:DivergenceJdotLargecL1Bound}. In fact, the proof of
            	Corollary \ref{C:DivergenceJdotInfinityL1Bound} is simpler: $c$ does not enter into the estimates.
            \end{proof}

        \begin{lemma}                           \label{L:UniformtimeApriori1}
            \begin{align} \label{E:UniformtimeApriori1LemmaStatement}
                \mid\mid\mid \Wdot \mid\mid\mid_{H^N,\tau} \leqc \big[\Lambda_1/2 + \tau \cdot 
                C(id;\Lambda_1,\Lambda_2,L_1,L_2') \big] 
                \cdot \mbox{\textnormal{exp}} \big(\tau \cdot C(id;\Lambda_1,\Lambda_2,L_1,L_2') \big).
            \end{align}
        \end{lemma}
            \begin{proof}

        Our proof of Lemma \ref{L:UniformtimeApriori1} follows from a Gronwall estimate in
        the $H^N$ norm of the variation $\Wdot$ defined in \eqref{E:DotWcUniformTime}. Rather than directly estimating
        the $H^N$ norm of $\Wdot,$ we instead estimate the $L^1$ norm of $\Jscrdotc_{\vec{\alpha}}^0,$
        where $\Jscrdotc_{\vec{\alpha}}$ is the energy current for the variation $\partial_{\vec{\alpha}}\Wdot$ defined by the BGS
        $\Vbold.$ This is favorable because of property \eqref{E:PositiveDefinitec} and because
        by \eqref{E:RescaledEnergyCurrentDivergence}, the divergence of $\Jscrdotc$ is lower order in $\Wdot.$
        We follow the method of proof of local existence from \cite{jS2008a}; the only difficulty
        is checking that our estimates are independent of all large $c.$ An important ingredient
        in our proof is showing that for $0 \leq |{\vec{\alpha}}| \leq N,$ we have the bound
        \begin{align}       \label{E:ComponentsinL2c}
            \|\mathbf{b}_{\vec{\alpha}}\|_{L^2} \leqc C(id;\Lambda_1,\Lambda_2,L_2')\big(1 + \|\Wdot\|_{H^N}\big),
        \end{align}
        where $\mathbf{b}_{\vec{\alpha}}$ is defined in \eqref{E:balphaDefc}.
        Let us assume \eqref{E:ComponentsinL2c} for the moment; we will provide a proof at the end of the proof of the lemma.

        We now let $\Jscrdotc_{\vec{\alpha}}$ denote the energy current \eqref{E:RescaledEnergyCurrent} for
        the variation $\partial_{\vec{\alpha}} \Wdot$ defined by the BGS
        $\Vbold,$ and abbreviating $\Jscrdot_{\vec{\alpha}} \eqdef \Jscrdotc_{\vec{\alpha}}$ to ease the notation, we define
        $\mathscr{E}(t) \geq 0$ by
            \begin{align} \label{E:UniformTimeEnergyc}
                \mathscr{E}^2(t) \eqdef \sum_{|\vec{\alpha}| \leq N} \int_{\mathbb{R}^3}
                \Jscrdot_{\vec{\alpha}}^0(t,\sforspace) \, d^3\sforspace .
            \end{align}
        By \eqref{E:UniformPositivitywithc}, and the Cauchy-Schwarz inequality for sums, we have that
                \begin{align} \label{E:LargecAlternateSobolevNormEstimate}
                    C_{\bar{\mathcal{O}}_2,Z} \|\Wdot\|_{H^N}^2 \leqc \mathscr{E}^2(t) \leqc C^{-1}_{\bar{\mathcal{O}}_2,Z}
                        \|\Wdot\|_{H^N}^2.
                \end{align}
                \noindent Here, the value of $Z=Z(id;\Lambda_2)$ is given by \eqref{E:Zdef}.

                Then by Lemma \ref{L:DivergenceJdotLargecL1Bound},
                \eqref{E:ComponentsinL2c}, \eqref{E:LargecAlternateSobolevNormEstimate},
                with $C \eqdef C(id;\Lambda_1,\Lambda_2,L_1,L_2'),$ we have

                \begin{align} \label{E:UniformTimeDifferentialEnergyInequalityc}
                    2\mathscr{E}\frac{d}{dt}\mathscr{E} &= \sum_{|\vec{\alpha}| \leq N}
                        \int_{\mathbb{R}^3} \partial_{\mu} \Jscrdot_{\vec{\alpha}}^{\mu} \, d^3 \sforspace
                        \leqc C \cdot \sum_{|\vec{\alpha}| \leq N}
                        \big(\|\partial_{\vec{\alpha}}\Wdot\|_{L^2}^2
                        + \|\partial_{\vec{\alpha}}\Wdot\|_{L^2}\|\mathbf{b}_{\vec{\alpha}}\|_{L^2}\big) \\
                        &\leqc C \cdot \big(\|\Wdot\|_{H^N}^2 + \|\Wdot\|_{H^N}\big) \notag
                        \leqc C \cdot \big(\mathscr{E}^2 + \mathscr{E}\big). \notag
                \end{align}
        We now apply Gronwall's inequality to
                \eqref{E:UniformTimeDifferentialEnergyInequalityc}, concluding that 
                \begin{align} \label{LargecAlternateHNEnergyBound}
                    \mathscr{E}(t) \leqc \big[\mathscr{E}(0)+ Ct\big] \cdot \mbox{exp}(Ct).
                \end{align}
        Using \eqref{E:LargecAlternateSobolevNormEstimate} again, it follows from \eqref{LargecAlternateHNEnergyBound} that
                \begin{align}
                    \|\Wdot(t)\|_{H^N} \leqc \big(C_{\bar{\mathcal{O}}_2,Z}^{-1} \|\Wdot(0)\|_{H^N}
                        + C t\big)\cdot\mbox{exp}(C t).
                \end{align}
                Recalling that $\Wdot(0)=\WID_c - \WIDSmoothed_c$ and taking into account inequality
                \eqref{E:IntialValueEstimatec}, the estimate \eqref{E:UniformtimeApriori1LemmaStatement} now follows.

        It remains to show \eqref{E:ComponentsinL2c}. Our proof is based on the Sobolev-Moser
        propositions stated in Appendix \ref{A:SobolevMoser} and the $c-$independent estimates of Section
        \ref{S:cDependence}. With the 5 components of the array $\mathbf{b}$ defined by \eqref{E:IterateInhomogeneouscf} - 
        \eqref{E:IterateInhomogeneouschj}, we first claim that the term $\Ac^0 \partial_{\vec{\alpha}}
        \big((\Ac^0)^{-1}\mathbf{b} \big)$ from \eqref{E:balphaDefc} satisfies
        \begin{align}                                                                           \label{E:SobolevClaimc}
            \| \Ac^0 \partial_{\vec{\alpha}} \big((\Ac^0)^{-1}\mathbf{b} \big) \|_{L^2} \leqc
            C(id;\Lambda_1,\Lambda_2,L_2').
        \end{align}
        Because \eqref{E:MatricescIsFiniteNormEstimates} and the induction hypotheses together imply that \\
        $\|\Ac^0(\scrW,\Phi)\|_{L^{\infty}}
        \leqc C(id;\Lambda_1,\Lambda_2),$ it suffices to bound the
        $H^N$ norm of $(\Ac^0)^{-1}\mathbf{b}$ by the right-hand side of \eqref{E:SobolevClaimc}. To this end,
        we use the induction hypotheses, \eqref{E:MatricescIsFiniteNormEstimates}, 
        Proposition \ref{P:CompositionProductSobolevMoser}, and Remark
        \ref{R:SobolevCalculusRemark}, with 
        $(\Ac^0(\scrW,\Phi))^{-1}$ playing the role of $\mathfrak{F}$ in the proposition and 
        $\mathbf{b}$ playing the role of $G,$ to conclude that
        \begin{align} \label{E:AtbHNEstimatec}
            \|(\Ac^0)^{-1}\mathbf{b}\|_{H^N} \leqc C(id;\Lambda_1,\Lambda_2)\|\mathbf{b}\|_{H^N}.
        \end{align}

        To estimate $\|\mathbf{b}\|_{H^N},$ we first split the array $\mathbf{b}$ into two arrays:
        \begin{align} \label{E:bSplit}
            \mathbf{b} = \mathfrak{B}_c(\scrW,\Phi,D\Phi) + \mathfrak{I}_c(id,\scrW,\Phi),
        \end{align}
        where $\mathfrak{B}_c$ is defined in Lemma \ref{L:InhomogeneousTermsLargecSobolevEstimates} and the 5-component array
        $\mathfrak{I}_c$ comprises the
        terms from the right-hand sides of \eqref{E:IterateInhomogeneouscf} - \eqref{E:IterateInhomogeneouschj}
        containing at least one factor of the smoothed initial data. By Lemma
        \ref{L:MultiplyFcGcSobolevEstimate}, Lemma
        \ref{L:OrderctoNegativeTwoEstimates}, Remark \ref{R:BoundedThroughALimit}, and Remark \ref{R:SmoothedDataLargecNPlusOneEstimate}, 
        we have that
      	\begin{align} \label{E:InitialDataInhomogeneousLargecSobolevEstimateA}
            \mathfrak{I}_c \in \mathcal{I}^{N}(id,\scrW,\Phi),
        \end{align}
        and from \eqref{E:InitialDataInhomogeneousLargecSobolevEstimateA} and the induction hypotheses, 
        it follows that
        \begin{align} \label{E:InitialDataInhomogeneousLargecSobolevEstimateB}
            \|\mathfrak{I}_c(id,\scrW,\Phi)\|_{H^N}\leqc C(id;\Lambda_1,\Lambda_2).
        \end{align}
        Furthermore, by \eqref{E:BcIsInIN} in the case $m=1$ and the induction hypotheses,
        we have that
        \begin{align} \label{E:InhomogeneousLargecSobolevEstimate}
            \|\mathfrak{B}_c(\scrW,\Phi,D\Phi)\|_{H^N} \leqc C(id;\Lambda_1,\Lambda_2,L_2').
        \end{align}
        Combining \eqref{E:bSplit}, \eqref{E:InitialDataInhomogeneousLargecSobolevEstimateB} and
        \eqref{E:InhomogeneousLargecSobolevEstimate}, we have that
        \begin{align}                                           \label{E:bHNEstimatec}
            \|\mathbf{b}\|_{H^N} \leqc
            C(id;\Lambda_1,\Lambda_2,L_2').
        \end{align}
        We now observe that \eqref{E:AtbHNEstimatec} and \eqref{E:bHNEstimatec} together imply \eqref{E:SobolevClaimc}.

        We next claim that the $\mathbf{k}_{\vec{\alpha}}$ terms
        \eqref{E:kalphaDefc} satisfy
        \begin{align}                                                               \label{E:SobolevClaim2c}
            \|\mathbf{k}_{\vec{\alpha}}\|_{L^2} \leqc
            C(id;\Lambda_1,\Lambda_2)\|\Wdot\|_{H^N}.
        \end{align}
        Since $\|\Ac^0(\scrW,\Phi)\|_{L^{\infty}} \leqc
        C(id;\Lambda_1,\Lambda_2),$ to prove \eqref{E:SobolevClaim2c}, it suffices to control the $L^2$ norm of
        $(\Ac^0)^{-1}\Ac^k \partial_k (\partial_{\vec{\alpha}}\Wdot) -\partial_{\vec{\alpha}}
        \left((\Ac^0)^{-1} \Ac^k \partial_k \Wdot \right).$ By the induction hypotheses, \eqref{E:MatricescIsFiniteNormEstimates}, 
        Proposition \ref{P:SobolevMissingDerivativeProposition}, and Remark
        \ref{R:SobolevMissingDerivativeRemark}, with \\
        $(\Ac^0)^{-1}\Ac^k = \Big((\Ac^0)^{-1}\Ac^k\Big)(\scrW,\Phi)$ playing the role 
        of $\mathfrak{F}$ in the proposition, and $\partial_k \Wdot$ playing the role of $G,$ we have (for $0 \leq 
        |\vec{\alpha}| \leq N$) that
        \begin{align}
            \|(\Ac^0)^{-1}\Ac^k \partial_{\vec{\alpha}} (\partial_k \Wdot) -\partial_{\vec{\alpha}}
                \big((\Ac^0)^{-1} \Ac^k \partial_k \Wdot \big)\|_{L^2}
            \leqc C(id;\Lambda_1,\Lambda_2) \|\partial \Wdot\|_{H^{N-1}},
        \end{align}
        from which \eqref{E:SobolevClaim2c} readily follows. This concludes the proof of
        \eqref{E:ComponentsinL2c}, and therefore also the proof of Lemma \ref{L:UniformtimeApriori1}.
       \end{proof}

 \section{The Non-Relativistic Limit of the EN$_{\kappa}^c$ System} \label{S:NonrelativisticLimit}

        In this section, we state and prove our main theorem regarding
        the non-relativistic limit of the EN$_{\kappa}^c$ system.
        Before stating our main theorem, we first state and prove a corollary of Theorem \ref{T:UniformLocalExistenceENkappac}
				that will be used in the proof of Theorem \ref{T:NewtonianLimit}, and we also briefly discuss local existence for the 
				EP$_{\kappa}$ system.

        \subsection{EN$_{\kappa}^c$ well-approximates EP$_{\kappa}$ for large $c$}
        	 The following corollary, which is an extension of Corollary \ref{C:ENkappacIsAlmostEPkappa}, shows that for large $c,$ 
        	 solutions to the EN$_{\kappa}^c$ system are ``almost'' solutions to the EP$_{\kappa}$ system.
					
                \begin{corollary} \label{C:ApproximateEPkappaSolutions}
                    For all large $c,$ the solutions $\scrV=(\scrW,\Phi,D\Phi)$ to the 
                    EN$_{\kappa}^c$ system \eqref{E:ENkappac1} - \eqref{E:PDefcII} furnished by Theorem 
                    \ref{T:UniformLocalExistenceENkappac} satisfy
                    \begin{align} \label{E:ApproximateEPkappaSolutions}
                        \Ainfinity^{\mu}(\scrW) \partial_{\mu} \scrW
                            &= \mathfrak{B}_{\infty}(\scrW,\partial\Phi) + \mathfrak{E}1_c \\
                        \Delta (\Phi - \PhiID_c) - \kappa^2 (\Phi - \PhiID_c) &= 4 \pi G [\mathfrak{R}_{\infty}(\Ent,p)
                        - \mathfrak{R}_{\infty}(\EntID,\pID)] + \mathfrak{E}2_c,
                    \end{align}
                    where
                    \begin{align}
                        \mid\mid\mid \mathfrak{E}1_c\mid\mid\mid_{H^{N-1},T} &\leqc c^{-2} C(id;\Lambda_1,\Lambda_2,L_2)
                            \label{E:ApproximateSolutionErrorSobolevEstimate1}\\
                        \mid\mid\mid \mathfrak{E}2_c\mid\mid\mid_{H^{N-1},T} &\leqc c^{-1}
                        C(id;\Lambda_1,\Lambda_2,L_4),
                            \label{E:ApproximateSolutionErrorSobolevEstimate2}
                    \end{align}
                    and $T$ is from Theorem \ref{T:UniformLocalExistenceENkappac}.
                 \end{corollary}
                   
                    \begin{proof}
                        The estimate \eqref{E:ApproximateSolutionErrorSobolevEstimate1}
                        follows from multiplying each side of \eqref{E:SolveforTimeDerivativesscrW}
                        by $\Ainfinity^0(\scrW),$ and then combining
                        Proposition \ref{P:CompositionProductSobolevMoser}, Remark \ref{R:SobolevCalculusRemark},
                        \eqref{E:UniformTime1d}, and the induction hypotheses from
                        Section \ref{SS:InductionHypotheses}, which are valid on $[0,T];$ 
                        we remark that we are making use of the $\mathscr{O}^{N-1}(c^{-2};\Wbold,\partial\Wbold,\Phi,D\Phi)$
                        estimate on the right-hand side of \eqref{E:SolveforTimeDerivativesscrW}.
                        Similarly, the 
                        estimate \eqref{E:ApproximateSolutionErrorSobolevEstimate2}
                        follows from the fact that \\ $\Delta (\Phi - \PhiID_c) - \kappa^2 (\Phi - \PhiID_c)= c^{-2}\partial_t^2 \Phi 
                        + l,$ where $l$ is given by \eqref{E:lExpression}, together with \eqref{E:lInitialEstimate},
                        \eqref{E:lInitialFcEstimate} in the case $m=0,$ \eqref{E:UniformTime1f}, and the
                        induction hypotheses.
                    \end{proof}

        \subsection{Local existence for EP$_{\kappa}$} 
				 In this section, we briefly discuss local existence for the EP$_{\kappa}$ system.

        \begin{theorem} \label{T:EPkappaLocalExistence} {\bf (Local Existence for EP$_{\kappa}$)}
            Let $\scrVID_{\infty}$ denote initial data \eqref{E:EPkappaData} for
            the EP$_{\kappa}$ system \eqref{E:EPkappa1} - \eqref{E:QinfinityRelationship}
            that are subject to the conditions described in Section
            \ref{S:IVPc}. Assume further that the equation of state is ``physical'' as
            described in sections \ref{SS:ENkappacDerivation} and \ref{SS:ApplicationtoENkappac}. Then there exists a
            $T_{\infty} > 0$ such that 
            \eqref{E:EPkappa1} - \eqref{E:QinfinityRelationship} has a unique classical solution $\scrV_{\infty} \in
            C_b^2([0,T_{\infty}] \times \mathbb{R}^3)$ of the form $\scrV_{\infty} \eqdef (\Ent_{\infty},p_{\infty},v_{\infty}^1,\cdots, 
            \partial_3 \Phi_{\infty}),$ and such that 
            $\scrV_{\infty}(0,\sforspace)=\scrVID_{\infty}(\sforspace).$ Additionally, $T_{\infty}$ can be
            chosen such that \\ 
            $\scrV_{\infty}([0,T_{\infty}] \times \mathbb{R}^3) \subset \mathfrak{K},$
            where the compact convex set $\mathfrak{K}$ is defined in
            \eqref{E:frakKdef}. Finally,  
            $\scrV_{\infty} \in \bigcap_{k=0}^{k=2}C^k([0,T_{\infty}],H_{\scrVb_{\infty}}^{N-k})$ and 
            $\Phi \in  C_b^3([0,T_{\infty}] \times \mathbb{R}^3) \cap 
            \bigcap_{k=0}^{k=3}C^k([0,T_{\infty}],H_{\Phibar_{\infty}}^{N+1-k}).$
       \end{theorem}

       \begin{proof}
            Theorem \ref{T:EPkappaLocalExistence} can be proved
            by an iteration scheme based on the method of energy currents: energy currents ${^{(\infty)}{\dot{\mathscr{J}}}}$ can be 
            used to control $\|\scrW_{\infty}\|_{H_{\scrWb_{\infty}}^{N}},$ while $\|\Phi_{\infty}\|_{H_{\Phibar_{\infty}}^{N+1}}$
            can be controlled using the estimate 
            $\| f \|_{H^2} \leq C \|(\Delta - \kappa^2) f \|_{L^2}$ for $f \in H^2.$ These methods are employed in the
            proof of Theorem \ref{T:NewtonianLimit} below, so we don't
            provide a proof here. Similar techniques are used by Makino in \cite{tM1986}. We remark that these methods 
            apply in particular to the system studied by Kiessling (as described in Section \ref{SS:EPkappac}) in \cite{mK2003}.
       \end{proof}

        \subsection{Statement and proof of the main theorem} 

        \begin{theorem} \label{T:NewtonianLimit} {\bf (The Non-Relativistic Limit of EN$_{\kappa}^c$)}
            Let $\scrVID_{\infty}$ denote initial data \eqref{E:EPkappaData} for
            the EP$_{\kappa}$ system \eqref{E:EPkappa1} - \eqref{E:QinfinityRelationship}
            that are subject to the conditions described in Section
            \ref{S:IVPc}. Let $\scrVID_c$ denote the corresponding initial data \eqref{E:ENkappacscrVID} for
            the EN$_{\kappa}^c$ system \eqref{E:ENkappac1} - \eqref{E:PDefcII}
            constructed from $\scrVID_{\infty}$ as described in Section \ref{S:IVPc}, and assume that the 
            $c-$indexed equation of state satisfies the hypotheses \eqref{E:EOScHypothesis1} and \eqref{E:EOScHypothesis2}
            and is ``physical'' as described in sections \ref{SS:ENkappacDerivation} and \ref{SS:ApplicationtoENkappac}.
            Let $\scrV_{\infty} \eqdef (\Ent_{\infty},p_{\infty},v_{\infty}^1,\cdots, \partial_3 \Phi_{\infty})$
            $\big(\scrV_c \eqdef(\Ent_c,p_c,v_c^1,\cdots, \partial_3 \Phi_c) \big)$ denote the
            solution to the EP$_{\kappa}$ \big(EN$_{\kappa}^c$\big) system launched by $\scrVID_{\infty}$
            $\big(\scrVID_c\big)$ as furnished by Theorem
            \ref{T:EPkappaLocalExistence} (Theorem \ref{T:UniformLocalExistenceENkappac}). By Theorem
            \ref{T:EPkappaLocalExistence} and Theorem \ref{T:UniformLocalExistenceENkappac}, we may assume that for all large $c,$
            $\scrV_{\infty}$ and $\scrV_c$ exist on a common spacetime slab $[0,T] \times
            \mathbb{R}^3,$ where $T$ is the minimum of the two
            times from the conclusions of the theorems.
            Let $\scrW_{\infty}$ and $\scrW_c$ denote the first $5$ components of $\scrV_{\infty}$ and
            $\scrV_c$ respectively. Then there exists a constant $C>0$ such that
            \begin{align}
                \mid\mid\mid \scrW_{\infty} - \scrW_c \mid\mid\mid_{H^{N-1},T} &\leqc c^{-1} \cdot C \label{E:NewtonianConvergence} \\
                \mid\mid\mid\big(\Phi_{\infty} - \Phibar_{\infty}\big) - \big(\Phi_c -
                    \Phibar_c\big)\mid\mid\mid_{H^{N+1},T} &\leqc c^{-1} \cdot C
                    \label{E:ClassicalConvergencePhi}\\
                \lim_{c \to \infty} |\Phibar_{\infty} - \Phibar_{c}| &=0 \label{E:ClassicalConstantConvergence},
         		\end{align}
         		where the constants $\Phibar_{\infty}$ and $\Phibar_{c}$ are defined through the initial data by
          	\eqref{E:PhibarInfinityDef} and \eqref{E:PhiBarc} respectively.
        \end{theorem}
        \begin{remark}
            \eqref{E:NewtonianConvergence}, \eqref{E:ClassicalConvergencePhi}, \eqref{E:ClassicalConstantConvergence}, and Sobolev
            embedding imply that \\
            $\scrW_c \rightarrow \scrW_{\infty}$ uniformly and $\Phi_c \rightarrow \Phi_{\infty}$ 
            uniformly on $[0,T] \times \mathbb{R}^3$ as $c \to \infty.$ Furthermore, the interpolation estimate 
            \eqref{E:SobInterpolation},
            together with the uniform bound \\
            $\mid\mid\mid \scrW_c \mid\mid\mid_{H_{\scrWb_c}^{N},T} \leqc C$
            that follows from combining \eqref{E:ScrWInTermsofWEstimate}, \eqref{E:UniformTime1a}, and \eqref{E:UniformTime1b},
            collectively imply that $\lim_{c \to \infty} \mid\mid\mid \scrW_{\infty} - \scrW_c \mid\mid\mid_{H^{N'},T}=0$
            for any $N' < N.$ The reason that we cannot use our argument to obtain the $H^{N}$ norm on the left-hand side of 
            \eqref{E:NewtonianConvergence} instead of the $H^{N-1}$ norm
            is that the expression \eqref{E:bClassicalLimitdef} for $\mathbf{b}$ already involves one derivative of $\scrW,$ and therefore 
            can only be controlled in the $H^{N-1}$ norm.
        \end{remark}

\begin{proof}
            Throughout the proof, we refer to the constants
            $\Lambda_1, \Lambda_2,$ etc., from the conclusion of Theorem
            \ref{T:UniformLocalExistenceENkappac}. 
            To ease the notation, we drop the subscripts $c$ from the solution $\scrV_c$ and its first 5 components $\scrW_c,$ setting 
            $\scrV \eqdef \scrV_c,$ $\scrW \eqdef \scrW_c,$ etc. We then define with the aid of \eqref{E:PhiIDcDef}
            \begin{align}
                \scrWdot &\eqdef \scrW_{\infty} - \scrW \label{E:WdotClassicalLimit}\\
                \Phidot &\eqdef (\Phi_{\infty} - \Phibar_{\infty}) - (\Phi - \Phibar_c)
                = (\Phi_{\infty} - \PhiID_{\infty}) - (\Phi - \PhiID_c). \label{E:PhidotClassicalLimit}
            \end{align}
            Our proof of Theorem \ref{T:NewtonianLimit} is similar to our proof of
            Lemma \ref{L:UniformtimeApriori1}; we use energy currents and elementary harmonic
            analysis (i.e. Lemma \ref{L:AlterateHNnormestimate}) to obtain a Gronwall estimate for the $H^{N-1}$ norm of the 
            variation $\scrWdot$ defined in \eqref{E:WdotClassicalLimit}. It will also follow from our proof that the $H^{N+1}$ norm 
            of $\Phidot$ is controlled in terms of $\|\scrWdot\|_{H^{N-1}}$ plus a small remainder. We remark that all of the
            estimates in this proof are valid on the interval $[0,T],$ where $T$ is as in the statement of Theorem 
            \ref{T:NewtonianLimit}.
                    
          	From definitions \eqref{E:WdotClassicalLimit} and \eqref{E:PhidotClassicalLimit},
           	it follows that $\scrWdot, \Phidot$ are solutions to the following EOV$_{\kappa}^{\infty}$ defined
           	by the BGS $\scrW_{\infty}:$
                    \begin{align}
                        \Ainfinity^{\mu}(\scrW_{\infty}) \partial_{\mu} \scrWdot
                            &= \mathbf{b} \label{E:EOVClassicalLimit} \\
                        (\Delta - \kappa^2) \Phidot &= l, \label{E:PhidotClassicalLimitEquation}
                    \end{align}
                    where
                    \begin{align}
                        \mathbf{b} & \eqdef \mathfrak{B}_{\infty}(\scrW_{\infty},\partial\Phi_{\infty})
                           - \mathfrak{B}_{\infty}(\scrW,\partial\Phi)
                            + \big[\Ainfinity^{\mu}(\scrW) - \Ainfinity^{\mu}(\scrW_{\infty})\big]\partial_{\mu}\scrW - 
                          	\mathfrak{E}1_c
                            \label{E:bClassicalLimitdef} \\
                        l & \eqdef 4 \pi G \big[\mathfrak{R}_{\infty}(\Ent_{\infty},p_{\infty}) - \mathfrak{R}_{\infty}(\Ent,p)\big]
                            - \mathfrak{E}2_c, \label{E:ClassicallLimitdef}
                    \end{align}
                    $\mathfrak{B}_{\infty}$ is defined in Lemma \ref{L:InhomogeneousTermsLargecSobolevEstimates},
                    and $\mathfrak{E}1_c,$ $\mathfrak{E}2_c$ are
                    defined in Corollary \ref{C:ApproximateEPkappaSolutions}. Note that the definition of $l$ in
                    \eqref{E:ClassicallLimitdef} differs from the definition \eqref{E:lExpression} of $l$ that is used in the
                    proof of Corollary \ref{C:ApproximateEPkappaSolutions}. By comparing
                    \eqref{E:EPkappaData} and \eqref{E:ENkappacscrVID}, we see that the initial condition satisfied by $\scrWdot$ is
                    \begin{align} \label{E:ClassicalLimitInitialConditions}
                        \scrWdot(0) &= \mathbf{0}.
                    \end{align}

                    Differentiating equation \eqref{E:EOVClassicalLimit} with the spatial multi-index operator 
                    $\partial_{\vec{\alpha}},$ we have that
                    \begin{align}
                        \Ainfinity^{\mu}(\scrW_{\infty}) \partial_{\mu} \big( \partial_{\vec{\alpha}} \scrWdot \big) =
                            \mathbf{b}_{\vec{\alpha}},
                    \end{align}
                    where (suppressing the dependence of $\Ainfinity^{\nu}(\cdot)$ on $\scrW_{\infty}$ for $\nu =0,1,2,3$)
                    \begin{align}
                        \mathbf{b}_{\vec{\alpha}} \eqdef \Ainfinity^0 \partial_{\vec{\alpha}} \left((\Ainfinity^0)^{-1}\mathbf{b} \right) +
                            \mathbf{k}_{\vec{\alpha}} \label{E:balphaDefinfinity}
                    \end{align}
          and
          \begin{align}
            \mathbf{k}_{\vec{\alpha}} \eqdef \Ainfinity^0 \big[(\Ainfinity^0)^{-1}\Ainfinity^k \partial_{\vec{\alpha}}
            (\partial_k \scrWdot) -\partial_{\vec{\alpha}} \big((\Ainfinity^0)^{-1} \Ainfinity^k \partial_k \scrWdot
            \big)\big]  \label{E:kalphaDefinfinity}.
            \end{align}
            As an intermediate step, we will show that for $0 \leq |\vec{\alpha}| \leq N-1,$ we have that
                \begin{align}
                    \|\mathbf{b}_{\vec{\alpha}}\|_{L^2} &\leqc C\big(id;\mid\mid\mid \scrW_{\infty} 
                    \mid\mid\mid_{H_{\scrWb_{\infty}}^N,T},\Lambda_1,\Lambda_2,L_1,L_2,L_4 \big)\cdot
                    \big(\|\scrWdot\|_{H^{N-1}} + c^{-1} \big) \label{E:ClassicalbalphaSobolevEstimate}.
                \end{align}

                Let us assume \eqref{E:ClassicalbalphaSobolevEstimate} for the moment and
                proceed as in Lemma \ref{L:UniformtimeApriori1}: we let $\Jscrdotinfinity_{\vec{\alpha}}$
                denote the energy current \eqref{E:RescaledEnergyCurrentInfinity} for
                $\partial_{\vec{\alpha}} \scrWdot$ defined by the BGS
                $\scrW_{\infty},$ and define $\mathscr{E}(t) \geq 0$ by
                \begin{align} \label{E:ClassicalEnergyc}
                    \mathscr{E}^2(t) \eqdef \sum_{|\vec{\alpha}| \leq N-1} \int_{\mathbb{R}^3}
                        \Jscrdot_{\vec{\alpha}}^{0}(t,\sforspace) \, d^3\sforspace,
                \end{align}
                where we have dropped the superscript $(\infty)$ on $\Jscrdot$ to ease the notation.
                By \eqref{E:UniformPositivitywithc}, Remark \ref{R:NoZdependence}, and the Cauchy-Schwarz inequality for sums, we have 		
                that
                \begin{align} \label{E:LargecAlternateSobolevNormEstimateAgain}
                    C_{\bar{\mathcal{O}}_2} \|\scrWdot\|_{H^{N-1}}^2 \leqc \mathscr{E}^2(t) \leqc C_{\bar{\mathcal{O}}_2}^{-1}
                        \|\scrWdot\|_{H^{N-1}}^2.
                \end{align}
                Then by Corollary \ref{C:DivergenceJdotInfinityL1Bound} + Sobolev embedding, \eqref{E:ClassicalbalphaSobolevEstimate}, and
                \eqref{E:LargecAlternateSobolevNormEstimateAgain}, with \\
                $C=C\big(id;\mid\mid\mid\scrW_{\infty}\mid\mid\mid_{H_{\scrWb_{\infty}}^N,T},
                \mid\mid\mid \partial_t \scrW_{\infty}\mid\mid\mid_{H^{N-1},T},\Lambda_1,
                \Lambda_2,L_1,L_2,L_4\big),$ we have that
                \begin{align} \label{E:ClassicalDifferentialEnergyInequality}
                    2 \mathscr{E}\frac{d}{dt}\mathscr{E} &= \sum_{|\vec{\alpha}| \leq {N-1}}
                        \int_{\mathbb{R}^3} \partial_{\mu} \Jscrdot_{\vec{\alpha}}^{\mu} \, d^3 \sforspace
                        \leqc C \cdot \sum_{|\vec{\alpha}| \leq {N-1}} \big(\|\partial_{\vec{\alpha}}\scrWdot\|_{L^2}^2
                        + \|\partial_{\vec{\alpha}}\scrWdot\|_{L^2}\|\mathbf{b}_{\vec{\alpha}}\|_{L^2}\big) \\
                        & \leqc C \cdot \|\scrWdot\|_{H^{N-1}}^2 + c^{-1} C \cdot \|\scrWdot\|_{H^{N-1}}
                            \leqc C \cdot \mathscr{E}^2 + c^{-1} C \cdot \mathscr{E}. \notag
                \end{align}
                Taking into account \eqref{E:ClassicalLimitInitialConditions}, which implies that
                $\mathscr{E}(0)=0,$ we apply Gronwall's inequality to
                \eqref{E:ClassicalDifferentialEnergyInequality}, concluding
                that for $t \in [0,T],$
                \begin{align} \label{E:EisOrderOneOverc}
                    \mathscr{E}(t) \leqc c^{-1} C \cdot t \cdot \mbox{exp}(C \cdot t).
                \end{align}
                From \eqref{E:LargecAlternateSobolevNormEstimateAgain} and \eqref{E:EisOrderOneOverc}, it follows that
                \begin{align} \label{E:MainTheoremWdotConclusion}
                    \mid\mid\mid \scrWdot \mid\mid\mid_{H^{N-1},T} \leqc c^{-1}C \cdot T \cdot \mbox{exp}(T \cdot C),
                \end{align}
                \noindent which implies \eqref{E:NewtonianConvergence}.
                
           			We now return to the proof of \eqref{E:ClassicalbalphaSobolevEstimate}.
                To prove \eqref{E:ClassicalbalphaSobolevEstimate}, we show only that
                the following bound holds, where for the remainder of this proof, we abbreviate \\
                $C=C\big(id;\mid\mid\mid\scrW_{\infty}\mid\mid\mid_{H_{\scrWb_{\infty}}^{N-1},T},\Lambda_1,
                \Lambda_2,L_1,L_2,L_4\big):$
                \begin{align} \label{E:ClassicalbSobolevEstimate}
                    \|\mathbf{b}\|_{H^{N-1}} &\leqc C\cdot\|\scrWdot\|_{H^{{N-1}}} + c^{-1}C.
                \end{align}
                The remaining details, which we leave up to the reader, then follow as in the proof of Lemma 
                \ref{L:UniformtimeApriori1}. By \eqref{E:scrWHNBound}, which is valid for
                $\tau = T,$ and by \eqref{E:ModifiedSobolevEstimate2}, we have
                that
                \begin{align} \label{E:KleinGordonClassicalLimitSobolevDifferenceEstimate}
                    \|\mathfrak{R}_{\infty}(\Ent_{\infty},p_{\infty}) - \mathfrak{R}_{\infty}(\Ent,p)\|_{H^{{N-1}}} \leqc
                        C \cdot \|\scrWdot\|_{H^{{N-1}}},
                \end{align}
                and combining \eqref{E:ApproximateSolutionErrorSobolevEstimate2}, \eqref{E:PhidotClassicalLimitEquation}, 
                \eqref{E:ClassicallLimitdef}, \eqref{E:KleinGordonClassicalLimitSobolevDifferenceEstimate}, and Lemma
                \ref{L:AlterateHNnormestimate}, it follows that
                \begin{align} \label{E:PhidotClassicalLimitSobolevDifferenceEstimate}
                    \|\Phidot\|_{H^{N+1}} \leqc C \cdot \|l\|_{H^{N-1}} \leqc C \cdot\|\scrWdot\|_{H^{{N-1}}}
                        + c^{-1} C.
                \end{align}

                Similarly, taking into account \eqref{E:PhidotClassicalLimitSobolevDifferenceEstimate}, we have that
                \begin{align} \label{E:WdotInhomogeneousClassicalLimitSobolevDifferenceEstimate}
                    \|\mathfrak{B}_{\infty}(\scrW_{\infty},\partial\Phi_{\infty})  -
                    \mathfrak{B}_{\infty}(\scrW,\partial\Phi)\|_{H^{{N-1}}} & \leqc C \cdot(\|\scrWdot\|_{H^{{N-1}}} +
                        \|\partial\Phidot\|_{H^{{N-1}}}) \\
                    & \leqc C \cdot \|\scrWdot\|_{H^{{N-1}}} + c^{-1} C. \notag
                \end{align}

                Finally, by \eqref{E:scrWHNBound} and \eqref{E:UniformtimeApriori2Again}, which
                are both valid for $\tau = T,$ by \eqref{E:ModifiedSobolevEstimate}, and by
                \eqref{E:ModifiedSobolevEstimate2}, we have that
                \begin{align} \label{E:MatrixDifferenceTerm}
                    \big\|\big[\Ainfinity^{\mu}(\scrW) -
                    \Ainfinity^{\mu}(\scrW_{\infty})\big]\partial_{\mu}\scrW \big\|_{H^{N-1}}
                    \leqc C \cdot \|\scrWdot \|_{H^{N-1}}.
                \end{align}

                Inequality \eqref{E:ClassicalbSobolevEstimate} now follows from
                \eqref{E:ApproximateSolutionErrorSobolevEstimate1}, \eqref{E:bClassicalLimitdef},
                \eqref{E:WdotInhomogeneousClassicalLimitSobolevDifferenceEstimate}, and \eqref{E:MatrixDifferenceTerm}. The estimate 
                \eqref{E:ClassicalConvergencePhi} then follows from \eqref{E:PhidotClassicalLimit},
                \eqref{E:MainTheoremWdotConclusion}, and \eqref{E:PhidotClassicalLimitSobolevDifferenceEstimate}, while
                \eqref{E:ClassicalConstantConvergence} is merely a restatement of \eqref{E:PhibarcConvergence}.
            \end{proof}

\section*{Acknowledgments}
I would like to thank Michael Kiessling and A.
Shadi Tahvildar-Zadeh for discussing this project with me and for providing 
comments that were helpful in my revision of the earlier drafts. I would also like
to thank the anonymous referee for providing suggestions that helped me to clarify certain
points and for providing some of the references. Work supported by NSF Grant DMS-0406951. Any opinions,
conclusions, or recommendations expressed in this material are those of the author
and do not necessarily reflect the views of the NSF.

\appendix

\setcounter{section}{0}
   \setcounter{subsection}{0}
   \setcounter{subsubsection}{0}
   \setcounter{paragraph}{0}
   \setcounter{subparagraph}{0}
   \setcounter{figure}{0}
   \setcounter{table}{0}
   \setcounter{equation}{0}
   \setcounter{theorem}{0}
   \setcounter{definition}{0}
   \setcounter{example}{0}
   \setcounter{remark}{0}
   \setcounter{proposition}{0}
   \renewcommand{\thesection}{\Alph{section}}
   \renewcommand{\theequation}{\Alph{section}.\arabic{equation}}
   \renewcommand{\theproposition}{\Alph{section}-\arabic{proposition}}
   \renewcommand{\thecorollary}{\Alph{section}.\arabic{corollary}}
   \renewcommand{\thedefinition}{\Alph{section}.\arabic{definition}}
   \renewcommand{\thetheorem}{\Alph{section}.\arabic{theorem}}
   \renewcommand{\theremark}{\Alph{section}.\arabic{remark}}
   \renewcommand{\thelemma}{\Alph{section}-\arabic{lemma}}
   \renewcommand{\theexample}{\Alph{section}.\arabic{example}}   

\section{Inhomogeneous Linear Klein-Gordon Estimates} \label{A:Klein-GordonEstimates}
  
                In this appendix, we collect together some standard energy estimates for the linear Klein-Gordon equation
                with an inhomogeneous term. We provide some proofs for convenience. Throughout this appendix,
                we abbreviate $L^p=L^p(\mathbb{R}^d)$ and \\ 
                $H^j=H^j(\mathbb{R}^d).$
                \begin{proposition}    \label{P:Klein-Gordon}
                    Let $l \in C^0([0,T],H^{N})$ and $\PsiID_0(\sforspace)\in H^{N},$ where $N
                    \in \mathbb{N}.$ Then there is a unique
                    solution $\Phidot(t,\sforspace):\mathbb{R}\times\mathbb{R}^d \rightarrow
                    \mathbb{R}$ to the equation
                    \begin{align}     \label{E:Klein-Gordon}
                        -c^{-2} \partial_t^2 \Phidot + \Delta \Phidot - \kappa^2 \Phidot = l
                    \end{align}
                    with initial data $\Phidot(0,\sforspace)=0,$ $\partial_t \Phidot(0,\sforspace)=
                    \PsiID_0(\sforspace),$ where $\Delta \eqdef \sum_{i=1}^d \partial_i^2.$ The solution
                    has the regularity property $\Phidot \in \bigcap_{k=0}^{k=1}C^k([0,T], H^{N+1-k}).$
            	\end{proposition}
                    \begin{proof}
                        This is a standard result; consult \cite{cS1995} for a proof.
                    \end{proof}

                \begin{proposition}                                           \label{P:Klein-GordonPhiBounds}
                    Assume the hypotheses of Proposition \ref{P:Klein-Gordon}. Assume further 
                    that \\ $l \in \bigcap_{k=0}^{k=2}C^k([0,T],H^{N-k}).$
                    Then there exists a constant $C_0(\kappa) > 0$ such that
                    \begin{align}
                        &\mid\mid\mid \Phidot \mid\mid\mid_{H^{N+1},T}
                            \leq  C_0(\kappa) \cdot \Big \lbrace c^{-1} \|\PsiID_0\|_{H^{N}}
                            + c T \mid\mid\mid l \mid\mid\mid_{H^{N},T}\Big \rbrace  \label{E:Klein-GordonPhiHNPlusOneNormBound} \\
                        &\mid\mid\mid \partial_t \Phidot \mid\mid\mid_{H^{N},T}
                        \leq  C_0(\kappa) \cdot \Big \lbrace \|\PsiID_0\|_{H^{N}}
                            + c^2 T \mid\mid\mid l \mid\mid\mid_{H^{N},T}\Big \rbrace 
                            \label{E:Klein-GordonPartialtPhiHNNormBound} \\
                        &\mid\mid\mid \partial_t \Phidot \mid\mid\mid_{H^{N},T}
                        \leq  C_0(\kappa) \cdot \Big \lbrace \|\PsiID_0\|_{H^{N}} + c\|l(0)\|_{H^{N-1}}
                            + c T \mid\mid\mid \partial_t l \mid\mid\mid_{H^{N-1},T}\Big \rbrace  \label{E:Klein-GordonPartialtPhiHNNormBound2} \\
                        &\mid\mid\mid \partial_t^2 \Phidot \mid\mid\mid_{H^{N-1},T}
                        \leq  C_0(\kappa) \cdot \Big \lbrace c\|\PsiID_0\|_{H^{N}} + c^2 \|l(0)\|_{H^{N-1}}
                            + c^2 T \mid\mid\mid \partial_t l \mid\mid\mid_{H^{N-1},T}\Big \rbrace  \label{E:Klein-GordonPartialSquaredtPhiHNMinusOneNormBound} \\
                         &\mid\mid\mid \partial_t^2 \Phidot \mid\mid\mid_{H^{N-1},T} \notag \\
                         & \ \leq C_0(\kappa) \cdot \Big \lbrace c^2\|l(0)\|_{H^{N-1}} + c\|(\Delta - \kappa^2)\PsiID_0 - 
                         \partial_t l(0)\|_{H^{N-2}}
                            + c T \mid\mid\mid \partial_t^2 l \mid\mid\mid_{H^{N-2},T}\Big \rbrace  
                            \label{E:Klein-GordonPartialSquaredtPhiHNMinusOneNormBound2} \\
                        &\mid\mid\mid \partial_t^3 \Phidot \mid\mid\mid_{H^{N-2},T} \notag \\
                        & \ \leq C_0(\kappa) \cdot \Big \lbrace c^3\|l(0)\|_{H^{N-1}} + c^2\|(\Delta - \kappa^2)\PsiID_0 - 
                        	\partial_t l(0)\|_{H^{N-2}} + c^2 T \mid\mid\mid \partial_t^2 l \mid\mid\mid_{H^{N-2},T}\Big \rbrace. 
                        	\label{E:Klein-GordonPartialCubedtPhiHNMinusTwoNormBound}
                    \end{align}
                    \end{proposition}

                    \begin{proof}
                    Because $\partial^{(k)} \Phidot$ is a
                    solution to the Klein-Gordon equation \\ $-c^{-2} \partial_t^2 \big(\partial^{(k)}\Phidot \big) +
                    \Delta \big(\partial^{(k)}  \Phidot \big) - \kappa^2 \big( \partial^{(k)}  \Phidot \big) = \partial^{(k)}
                    l,$ we will use standard energy estimates for
                    the linear Klein-Gordon equation to estimate $\mid\mid\mid
                    \Phidot \mid\mid\mid_{H^{N+1},T}.$ Thus, for $0 \leq k \leq N,$ we define $E_{k}(t) \geq 0$ by
                        \begin{align} \label{E:Klein-GordonEnergyDifferentialEnergyEquality}
                            E^2_{k}(t) \eqdef
                                \|\kappa \partial^{(k)} \Phidot(t) \|_{L^2}^2  + \|\partial^{(k+1)}\Phidot \|_{L^2}^2
                                + \| c^{-1} \partial^{(k)} \partial_t \Phidot(t) \|_{L^2}^2.
                        \end{align}
                        We now multiply each side of the equation satisfied by
                        $\partial^{(k)} \Phidot$ by $-\partial^{(k)}\partial_t \Phidot,$ integrate by parts over
                        $\mathbb{R}^d,$ and use H\"{o}lder's inequality to arrive at the following
                        chain of inequalities:
                        \begin{align} \label{E:EkDifferentialInequality}
                            {E_{k}}(t) \frac{d}{dt}  {E_{k}}(t) &= \frac{1}{2}\frac{d}{dt} \left({E^2_{k}}(t)\right) =  \int_{\mathbb{R}^d}
                                \big(-\partial^{(k)} \partial_t \Phidot \big) \cdot \big(\partial^{(k)} l \big) \, d^d \sforspace \\
                            &\leq \|\partial^{(k)}\partial_t \Phidot(t)\|_{L^2} \|\partial^{(k)} l(t) \|_{L^2}, \notag
                        \end{align}
                        where $\big(-\partial^{(k)} \partial_t \Phidot \big) \cdot \big(\partial^{(k)} l
                        \big)$ denotes the array-valued quantity
                        formed by taking the component by
                        component product of the two arrays $-\partial^{(k)} \partial_t
                        \Phidot$ and $\partial^{(k)} l.$

                        If we now define $E(t) \geq 0$ by
                        \begin{align}                    \label{E:Klein-GordonHNPlusOneEnergyforPsi}
                            E^2(t) \eqdef
                            \left(\sum_{k=0}^{N} {E^2_{k}(t)}\right)= \kappa^2\|\Phidot(t)\|^2_{H^{N}}
                            + \|\partial\Phidot(t)\|^2_{H^{N}} + c^{-2}\|\partial_t \Phidot(t)\|^2_{H^{N}},
                        \end{align}
                        it follows from \eqref{E:EkDifferentialInequality} and the Cauchy-Schwarz
                        inequality for sums that
                        \begin{align}   \label{E:EnergySquaredInequality}
                            E(t) \frac{d}{dt} E(t) = \frac{1}{2}\frac{d}{dt} (E^2(t)) \leq \|\partial_t \Phidot\|_{H^{N}}\|l(t)\|_{H^{N}}
                            \leq c E(t)\|l(t)\|_{H^{N}},
                        \end{align}
                        and so
                        \begin{align} \label{E:HNPlusOneKlein-GordonEnergyDifferentialEnergyEquality}
                            \frac{d}{dt} E(t) \leq c \|l(t)\|_{H^{N}}.
                        \end{align}
                        Integrating \eqref{E:HNPlusOneKlein-GordonEnergyDifferentialEnergyEquality} over time, we have the
                        following inequality, valid for $t \in [0,T]:$
                        \begin{align} \label{E:Klein-GordonEnergyHNPlusOneEnergyEquality}
                            E(t) \leq E(0) + c t \mid\mid\mid  l \mid\mid\mid_{H^{N},T}.
                        \end{align}

                        From the definition of $E(t)$ and the initial condition $\Phidot=0,$ we have that
                        \begin{align}                                       \label{E:PhiHNPlusOneLessThanE}
                            \|\Phidot(t)\|_{H^{N+1}} &\leq C(\kappa) E(t) \\
                            \|\partial_t \Phidot(t)\|_{H^{N}} &\leq c E(t) \label{E:PartialtPhiHNLessThanE}\\
                            E(0) &= c^{-1} \|\PsiID_0\|_{H^{N}}. \label{E:BoundonHNPlusOneEnergyattEqualsZero}
                        \end{align}

                        Combining \eqref{E:Klein-GordonEnergyHNPlusOneEnergyEquality},
                        \eqref{E:PhiHNPlusOneLessThanE}, \eqref{E:PartialtPhiHNLessThanE},
                        and \eqref{E:BoundonHNPlusOneEnergyattEqualsZero}, and taking the sup over $t \in[0,T]$ proves
                        \eqref{E:Klein-GordonPhiHNPlusOneNormBound} and \eqref{E:Klein-GordonPartialtPhiHNNormBound}.

                        To prove \eqref{E:Klein-GordonPartialtPhiHNNormBound2}
                        - \eqref{E:Klein-GordonPartialCubedtPhiHNMinusTwoNormBound},
                        we differentiate the Klein-Gordon equation with respect to
                        $t$ (twice to prove \eqref{E:Klein-GordonPartialSquaredtPhiHNMinusOneNormBound2}
                        and \eqref{E:Klein-GordonPartialCubedtPhiHNMinusTwoNormBound})
                        and argue as above, taking into account the initial conditions
                        \begin{align}
                        \label{E:PartialtSquaredPhidotattequals0}
                            \partial_t^2 \Phidot(0) & = -c^{2}l(0) \\
                            \partial_t^3 \Phidot(0) & = c^{2}\big[(\Delta - \kappa^2) \PsiID_0
                            - \partial_tl(0)\big]. \label{E:PartialtCubedPhidotattequals0}
                        \end{align} \\
                    \end{proof}

                    \begin{corollary}                                                       \label{C:Klein-GordonPhiBounds}
                        Assume the hypotheses of Proposition \ref{P:Klein-GordonPhiBounds}, 
                        and let $C_0(\kappa)$ be the constant appearing in the conclusions of the proposition. Then
                        \begin{align}
                            &\mid\mid\mid \Phidot \mid\mid\mid_{H^{N+1},T}^2 \leq \big (C_0(\kappa)\big)^2 \cdot
                                \Big \lbrace c^{-2}\|\PsiID_0\|_{H^N}^2
                                + 2 T \cdot \mid\mid\mid \partial_t \Phidot \mid\mid\mid_{H^{N},T}
                                \cdot \mid\mid\mid l \mid\mid\mid_{H^{N},T} \Big \rbrace. \label{E:Klein-GordonPhiNocBound}
                        \end{align}
                    \end{corollary}
                        \begin{proof}
                            Inequality \eqref{E:EnergySquaredInequality} gives that $\frac{1}{2}\frac{d}{dt} (E^2(t)) \leq 
                            \|\partial_t \Phidot\|_{H^{N}}\|l(t)\|_{H^{N}}.$ Taking into account 
                            \eqref{E:PhiHNPlusOneLessThanE} and \eqref{E:BoundonHNPlusOneEnergyattEqualsZero},
                            the proof of \eqref{E:Klein-GordonPhiNocBound} easily follows.
                        \end{proof}

                     \begin{lemma} \label{L:AlterateHNnormestimate}
                        Let $N \in \mathbb{N},$ and $\mathscr{I} \in H^{N-1}.$
                        Suppose that $\Phidot \in L^2$ and that \\ 
                        $\Delta \Phidot -
                      	\kappa^2 \Phidot  = \mathscr{I}.$  Then $\Phidot \in H^{N+1}$ and
                        \begin{align} \label{E:AlternateHNnormestimate}
                            \|\Phidot\|_{H^{N+1}}
                            \leq C(N,\kappa)
                            \|\mathscr{I}\|_{H^{N-1}}.
                        \end{align}
                     \end{lemma}
                        \begin{proof}
                            For use in this argument, we define the Fourier transform through its action on integrable functions $F$ by
                            $\widehat{F}(\boldsymbol{\xi}) \eqdef
                            \int_{\mathbb{R}^d} F(\sforspace) e^{-2 \pi i \boldsymbol{\xi} \cdot \sforspace} \,
                            d^d \sforspace.$ The following chain of
                            inequalities uses standard results from Fourier analysis, including Plancherel's theorem:
                            \begin{align}                                                       \label{E:AlterateHNnormestimate2}
                            &\|\Phidot\|^2_{H^{2}} \leq C \|(1 + |2\pi \boldsymbol{\xi}|^2)^2 
                            	\widehat{\Phidot}(\boldsymbol{\xi})\|^2_{L^2} 
                           		\leq C(\kappa) \int_{\mathbb{R}^d} (\kappa^2 + |2 \pi 
                            	\boldsymbol{\xi}|^2)^2|\widehat{\Phidot}(\boldsymbol{\xi})|^2 \, 
                            	d^d \boldsymbol{\xi} \\
                            &=C(\kappa)\|(\kappa^2- \Delta)\Phidot \|^2_{L^2}
                            	=C(\kappa)\|\mathscr{I}\|^2_{L^2}, \notag
                            \end{align}
                            and this proves
                            \eqref{E:AlternateHNnormestimate} in the
                            case $N=1.$ To estimate $L^2$ norms of the $k^{th}$ order derivatives of
                            $\Phidot$ for $k \geq 1,$ we differentiate the equation $k$ times
                            to arrive at the equation $\Delta \big(\partial^{(k)}\Phidot\big) -
                            \kappa^2 \big(\partial^{(k)} \Phidot \big) = \partial^{(k)} \mathscr{I},$ and argue
                            as above to conclude that
                            \begin{align} \label{E:AlterateHNnormestimate3}
                                \|\partial^{(k)} \Phidot \|^2_{H^2} \leq
                                C(\kappa) \|\partial^{(k)}
                                \mathscr{I}\|^2_{L^2}.
                            \end{align}
                            Now we add the estimate \eqref{E:AlterateHNnormestimate2}
                            to the estimates \eqref{E:AlterateHNnormestimate3} for $1 \leq k \leq N - 1$ to conclude
                            \eqref{E:AlternateHNnormestimate}.
                        \end{proof}

          \begin{remark}
            The hypothesis $\Phidot \in L^2$ does not follow from the remaining
                assumptions. For example, consider $g (x) = e^x.$ Then $g - \frac{d^2}{dx^2}g
                \in L^2(\mathbb{R}),$ but \\
                $g \not\in H^2(\mathbb{R}).$
          \end{remark}

                \begin{proposition}
                \label{P:Klein-GordonNocDependence}
                    Assume the hypotheses of Proposition \ref{P:Klein-Gordon}. Assume further that \\
                  	$l \in \bigcap_{k=0}^{k=2}C^k([0,T],H^{N-k}).$ Then
                    \begin{align}                                               \label{E:Klein-GordonNocDependencePhi}
                        &\mid\mid\mid \Phidot  \mid\mid\mid_{H^{N+1},T} \ \leq  C(N,\kappa) \\
                        & \ \ \ \cdot \Big \lbrace c^{-1}\|\PsiID_0\|_{H^{N}} + \|l(0)\|_{H^{N-1}}
                            + \mid\mid\mid l \mid\mid\mid_{H^{N-1},T} + \ T \mid\mid\mid \partial_t l \mid\mid\mid_{H^{N-1},T}\Big 			
                            \rbrace \notag
                   	\end{align}
                   	\noindent and
                   	\begin{align}
                     		\mid\mid\mid \partial_t \Phidot \mid\mid\mid_{H^{N},T} \ \leq  C(N,\kappa) 
                        		  \cdot & \Big\lbrace c\|l(0)\|_{H^{N-1}} + \|(\Delta - \kappa^2)\PsiID_0 - \partial_tl(0)\|_{H^{N-2}} 
                        		  \label{E:Klein-GordonNocDependencePartialtPhi}  \\
                            & \qquad + \mid\mid\mid \partial_t l \mid\mid\mid_{H^{N-2},T} +
                            T \mid\mid\mid \partial_t^2 l \mid\mid\mid_{H^{N-2},T}\Big \rbrace. \notag
                    \end{align}
                \end{proposition}

                    \begin{proof}
                        Define $\mathscr{I}
                        \eqdef l + c^{-2}
                        \partial_t^2 \Phidot$ and observe that $\Phidot$ is a solution to
                        \begin{align}
                            \Delta \Phidot - \kappa^2 \Phidot  = \mathscr{I}.
                        \end{align}
                        By inequality \eqref{E:Klein-GordonPartialSquaredtPhiHNMinusOneNormBound} of Proposition 
                        \ref{P:Klein-GordonPhiBounds}, Lemma \ref{L:AlterateHNnormestimate}, and
                        the triangle inequality, we have that
                        \begin{align}
                            &\mid\mid\mid \Phidot \mid\mid\mid_{H^{N+1},T}
                            \leq C(N,\kappa) \mid\mid\mid l + c^{-2} \partial_t^2 \Phidot \mid\mid\mid_{H^{N-1},T} \\
                            & \qquad \leq C(N,\kappa) \cdot \Big \lbrace c^{-1}\|\PsiID_0\|_{H^{N}} + \|l(0)\|_{H^{N-1}}
                            + \mid\mid\mid l \mid\mid\mid_{H^{N-1},T} + \ T \mid\mid\mid \partial_t l 
                            \mid\mid\mid_{H^{N-1},T}\Big \rbrace, 	\notag
                        \end{align}
                        which proves $\eqref{E:Klein-GordonNocDependencePhi}.$

                        Because $\partial_t \Phidot$ satisfies the equation $-c^{-2}
                        \partial_t^2 (\partial_t \Phidot) +
                        \Delta \big(\partial_t \Phidot \big) - \kappa^2 \big(\partial_t \Phidot \big) = \partial_t
                        l,$ we may use a similar argument to prove \eqref{E:Klein-GordonNocDependencePartialtPhi}; we leave the simple
                        modification, which makes use of \eqref{E:Klein-GordonPartialCubedtPhiHNMinusTwoNormBound}, up to the reader.
                    \end{proof}

\section{Sobolev-Moser Estimates} \label{A:SobolevMoser}
		In this Appendix, we use notation that is as consistent as possible with
    our use of notation in the body of the paper. To conserve space, we refer the reader
    to the literature instead of providing proofs: propositions \ref{P:CompositionProductSobolevMoser} and
    \ref{P:SobolevTaylor} are similar to propositions proved in
    chapter 6 of \cite{lH1997}, while Proposition
    \ref{P:SobolevMissingDerivativeProposition} is proved in
    \cite{sKaM1981}. The corollaries and remarks below are straightforward extensions
    of the propositions. With the exception of Proposition \ref{P:SobInterpolation}, which is a standard Sobolev interpolation
    inequality, the proofs of the propositions given in the literature are commonly based on the following version of the
    Gagliardo-Nirenberg inequality \cite{lN1959}, together with repeated use of H\"{o}lder's inequality and/or Sobolev embedding,
    where throughout this appendix, we abbreviate $L^p=L^p(\mathbb{R}^d),H^j=H^j(\mathbb{R}^d),$ and $H_{\Vb}^j = 
    H_{\Vb}^j(\mathbb{R}^d)$:
\begin{lemma}                                                                                           \label{L:GN}
    If $i,k \in \mathbb{N}$ with $0 \leq i \leq k,$ and $\mathbf{V}$ is a
    scalar-valued or array-valued function on $\mathbb{R}^d$ satisfying $\mathbf{V} \in
    L^{\infty}$ and $\|\partial^{(k)} \mathbf{V}\|_{L^2} < \infty,$ then
    \begin{align}
        \| \partial^{(i)} \mathbf{V} \|_{L^{2k/i}} \leq C(k) \|\mathbf{V}\|_{L^{\infty}}^{1 -
        \frac{i}{k}}\|\partial^{(k)} \mathbf{V}\|_{L^2}^{\frac{i}{k}}.
    \end{align}
\end{lemma}

\begin{proposition}                                                                                                \label{P:CompositionProductSobolevMoser}
    Let $K \subset \mathbb{R}^{n}$ be a compact set, and let $j,d \in \mathbb{N}$ with $ j > \frac{d}{2}.$
    Let $\mathbf{V}:\mathbb{R}^d \rightarrow \mathbb{R}^n$ be an element of $H^j,$ and assume that
    $\mathbf{V} \subset K.$ Let $\mathfrak{F} \in C_b^j(K)$ be a $q \times q$ matrix-valued function,
    and let $G \in H^j$ be a $q \times q$ ($q \times 1)$ matrix-valued (array-valued) function. Then the $q \times q$ 	
    $(q \times 1)$ matrix-valued (array-valued) function $(\mathfrak{F} \circ \mathbf{V})G$ is an element of $H^j$ and
    \begin{align} \label{E:CompositionProductSobolevMoser}
        \|(\mathfrak{F} \circ \mathbf{V})G\|_{H^j} \leq C(j,d)|\mathfrak{F}|_{j,K}(1 +
        \|\mathbf{V}\|_{H^j}^j)\|G\|_{H^j}.
    \end{align}
	\end{proposition}

\begin{corollary}                                                                                               \label{C:SobolevCorollary}
    Assume the hypotheses of Proposition
    \ref{P:CompositionProductSobolevMoser} with the following
    changes: $\mathbf{V},G \in C^0([0,T],H^j).$
    Then the $q \times q$ $(q \times 1)$ matrix-valued (array-valued) function $(\mathfrak{F} \circ \mathbf{V})G$ is an element of
    $C^0([0,T],H^j).$
\end{corollary}

\begin{remark}                                                                                          \label{R:SobolevCalculusRemark}
    We often make use of a slight modification of Proposition \ref{P:CompositionProductSobolevMoser}
    in which the assumption $\mathbf{V} \in H^j$ is replaced with the assumption
    $\mathbf{V} \in H_{\Vb}^j,$ where $\bar{\mathbf{V}} \in \mathbb{R}^n$ is a
    constant array. Under this modified assumption, the conclusion of
    Proposition \ref{P:CompositionProductSobolevMoser}
    is modified as follows:
    \begin{align}                                                                                       \label{E:ModifiedSobolevEstimate}
        \|(\mathfrak{F} \circ \mathbf{V})G\|_{H^j} \leq C(j,d)|\mathfrak{F}|_{j,K}(1 +
        \|\mathbf{V}\|_{H_{\Vb}^j}^j)\|G\|_{H^j}.
    \end{align}
    A similar modification can be made to Corollary \ref{C:SobolevCorollary}.
\end{remark}

\begin{proposition}                                                                                             \label{P:SobolevTaylor}
    Let $K \subset \mathbb{R}^{n}$ be a compact convex set, and let $j,d \in 	
    \mathbb{N}$ with \\ 
    $j > \frac{d}{2}.$ Let $\mathfrak{F} \in C_b^j(K)$ be a  scalar or array-valued
    function. Let $\mathbf{V}, \widetilde{\mathbf{V}}: \mathbb{R}^d \rightarrow
    \mathbb{R}^n,$ and assume that $\mathbf{V}, \widetilde{\mathbf{V}}
    \in H^j.$ Assume further that $\mathbf{V}, \widetilde{\mathbf{V}} \subset 
    K.$ Then \\
    $\mathfrak{F} \circ \mathbf{V} - \mathfrak{F} \circ \widetilde{\mathbf{V}} \in H^j$ and
    \begin{align}
        \|\mathfrak{F} \circ \mathbf{V} - \mathfrak{F} \circ \widetilde{\mathbf{V}}\|_{H^j}
        \leq C(j,d,\|\mathbf{V}\|_{H^j},\|\widetilde{\mathbf{V}}\|_{H^j})|\mathfrak{F}|_{j+1,K} \|\mathbf{V} 
        - \widetilde{\mathbf{V}}\|_{H^j}.
    \end{align}
\end{proposition}

\begin{remark}     \label{R:SobolevTaylorCalculusRemark}
    As in Remark \ref{R:SobolevCalculusRemark}, we may replace the hypotheses $\mathbf{V}, \mathbf{\widetilde{V}} \in H^j$ 
    from Proposition
    \ref{P:SobolevTaylor} with the hypotheses $\mathbf{V}, \mathbf{\widetilde{V}} \in
    H_{\Vb}^j,$ in which case the conclusion
    of the proposition is:
    \begin{align}                                                                                       \label{E:ModifiedSobolevEstimate2}
        \|(\mathfrak{F} \circ \mathbf{V}) - (\mathfrak{F} \circ \mathbf{\widetilde{V}}) \|_{H^j} \leq
        C(j,d,\|\mathbf{V}\|_{H_{\Vb}^j},\|\widetilde{\mathbf{V}}\|_{H_{\Vb}^j})|\mathfrak{F}|_{j+1,K} 
        \|\mathbf{V} - \widetilde{\mathbf{V}}\|_{H^j}.
    \end{align}
    Furthermore, a careful analysis of the special case $\widetilde{\mathbf{V}} =
    \bar{\mathbf{V}},$ where $\Vb \in K$ is a constant array, gives the bound
    \begin{align} 														\label{E:ModifiedSobolevEstimateConstantArray}
        \|\mathfrak{F} \circ \mathbf{V} - \mathfrak{F} \circ \bar{\mathbf{V}}\|_{H^j} \leq C(j,d)|\partial 
        \mathfrak{F}/\partial \mathbf{V}|_{j-1,K}(1 + \|\mathbf{V}\|_{H_{\Vb}^j}^{j-1}) 
        (\|\mathbf{V}\|_{H_{\Vb}^j}),
    \end{align}
    in which we require less regularity of $\mathfrak{F}$ than we do in the general case.
\end{remark}

\begin{proposition}                                                                             \label{P:SobolevMissingDerivativeProposition}
        Assume the hypotheses of Proposition \ref{P:CompositionProductSobolevMoser} with the following two changes:
        \begin{enumerate}
            \item Assume $j > \frac{d}{2} + 1.$
            \item Assume that $G \in H^{j-1}.$
        \end{enumerate}
        Let $\vec{\alpha}$ be a spatial
        derivative multi-index such that $1 \leq |\vec{\alpha}| \leq j.$ Then
        \begin{align}       \label{E:SobolevMissingDerivativeProposition}
            \|\partial_{\vec{\alpha}}&\left[(\mathfrak{F} \circ \mathbf{V})G\right] - (\mathfrak{F} \circ
                \mathbf{V})\partial_{\vec{\alpha}}G\|_{L^2} \notag \\
            &\leq C(j,d)|\partial \mathfrak{F}/\partial \mathbf{V}|_{j-1,K}(\|\mathbf{V}\|_{H^j} +
                \|\mathbf{V}\|_{H^j}^j)\|G\|_{H^{j-1}}. 
        \end{align}
\end{proposition}

    \begin{remark}                                                              \label{R:SobolevMissingDerivativeRemark}
        As in Remark \ref{R:SobolevCalculusRemark}, we may replace the assumption $\mathbf{V} \in
        H^j$ in Proposition \ref{P:SobolevMissingDerivativeProposition}
        with the assumption $\mathbf{V} \in H_{\Vb}^j,$ where $\bar{\mathbf{V}}$ is a constant array, in
        which case we obtain
        \begin{align}
             \|\partial_{\vec{\alpha}}&\left[(\mathfrak{F} \circ \mathbf{V})G\right] - (\mathfrak{F} \circ
             \mathbf{V})\partial_{\vec{\alpha}}G\|_{L^2} \notag \\
             &\leq C(j,d)|\partial \mathfrak{F}/\partial \mathbf{V}|_{j-1,K}(\|\mathbf{V}\|_{H_{\Vb}^j} +
                \|\mathbf{V}\|_{H_{\Vb}^j}^j)\|G\|_{H^{j-1}}. 
        \end{align}
    \end{remark}

\begin{proposition}                                                           \label{P:SobInterpolation}
            Let $N',N \in \mathbb{R}$ be such that $0 \leq N' \leq N,$ and assume that $\mathfrak{F} \in H^N.$ Then
            \begin{align} \label{E:SobInterpolation}
                \|\mathfrak{F}\|_{H^{N'}} \leq C(N',d) \|\mathfrak{F}\|_{L^2}^{1-
                N'/N}\|\mathfrak{F}\|_{H^N}^{N'/N}.
            \end{align}
\end{proposition}

\bibliographystyle{amsalpha}
\bibliography{JBib}
\end{document}